%% file: main.tex
\documentclass{report}

\usepackage[utf8]{inputenc} 
\usepackage{lmodern}
\usepackage[T1]{fontenc}
\usepackage[DIV=10]{typearea} 
\usepackage{tikz}

\usepackage{enumitem}

\usepackage{microtype}

\usepackage{amssymb}
\usepackage{amsmath}
\usepackage{amsthm}
\usepackage{thmtools}
\usepackage{thm-restate}
\usepackage{titling}

\usepackage{comment}
\usepackage[normalem]{ulem}

\usepackage[backend=biber, style=alphabetic, backref=true, doi=false, url=false, maxcitenames=3, mincitenames=3, maxbibnames=10, minbibnames=10, sortlocale=en_US]{biblatex}  

\usepackage[procnumbered,ruled,vlined,linesnumbered]{algorithm2e}

\usepackage{cleveref}

\input{commands}

\declaretheorem[numberwithin=section]{theorem}
\declaretheorem[numberlike=theorem]{lemma}

\declaretheorem[numberlike=theorem]{corollary}
\declaretheorem[numberlike=theorem]{definition}
\declaretheorem[numberlike=theorem]{claim}
\declaretheorem[numberlike=theorem]{fact}

\declaretheorem[numberlike=theorem]{invariant}
\declaretheorem[numberlike=theorem]{remark}

\DeclareMathOperator*{\argmax}{arg\,max}

\DontPrintSemicolon
\SetKw{KwAnd}{and}
\SetProcNameSty{textsc}
\SetFuncSty{textsc}

\newcommand{\polylog}{\text{polylog}}
\newcommand{\poly}{\text{poly}}

\declaretheorem[numberwithin=section,refname={Theorem,Theorems},Refname={Theorem,Theorems},name={Theorem}]{thm}

\declaretheorem[style=definition,numberlike=thm,refname={Definition,Definitions},Refname={Definition,Definitions},name={Definition}]{defn}

\global\long\def\Otil{\tilde{O}}
\global\long\def\Ohat{\widehat{O}}
\global\long\def\poly{\mathrm{poly}}

\global\long\def\val{\mathrm{val}}
\global\long\def\vol{\mathrm{vol}}
\global\long\def\dist{\mathbf{dist}}

\global\long\def\congest{\mathsf{cong}}
\global\long\def\len{\mathsf{len}}
\global\long\def\davg{d_{avg}}
\global\long\def\pset{\mathcal{P}}

\global\long\def\Shat{\hat{S}}

\newcommand{\Omegahat}{\widehat{\Omega}}
\newcommand{\Omegatil}{\tilde{\Omega}}

\newcommand{\rmatching}{\textrm{{\sc Robust-Matching}}}

\newcommand{\rwitness}{\textrm{{\sc Robust-Witness}}}

\newcommand{\certifywitness}{\textrm{{\sc Certify-Witness}}}

\newcommand{\pathtowitness}{\textrm{{\sc Forest-From-Witness}}}
\newcommand{\shortoracle}{\textrm{{\sc Path-Inside-Expander}}}

\newcommand{\alphaex}{\alpha_{\textrm{ex}}}

\newcommand{\eps}{\epsilon}

\newcommand{\ignore}[1]{}

\newcommand{\gstar}{G^*}
\newcommand{\vstar}{V^*}

\newcommand{\fout}{\mathcal{F}_{out}}
\newcommand{\fin}{\mathcal{F}_{in}}

\renewcommand{\deg}[0]{\mathbf{deg}}

\makeatletter
\def\@makechapterhead#1{%
  \vspace*{50\p@}%
  {\parindent \z@ \raggedright \normalfont
    \ifnum \c@secnumdepth >\m@ne
        \Large\bfseries \@chapapp\space \thechapter
        \par\nobreak
        \vskip 20\p@
    \fi
    \interlinepenalty\@M
    \huge \bfseries #1\par\nobreak
    \vskip 40\p@
  }}
\makeatother

\title{Near-Optimal Algorithms for Reachability, Strongly-Connected Components and Shortest Paths in Partially Dynamic Digraphs}
\author{
	Maximilian Probst Gutenberg	
}
\date{}

\begin{document}

\begin{titlepage}
   \begin{center}
       \vspace*{1cm}

       {\Huge \textbf{Near-Optimal Algorithms for Reachability, Strongly-Connected Components and Shortest Paths in Partially Dynamic Digraphs}}

        \vspace{0.5cm}

       \textbf{by}
       
       \vspace{0.5cm}

       \Large{\textbf{Maximilian Probst Gutenberg}}
       
       \vspace{0.5cm}
        \textbf{supervised by }\\
        \textbf{Christian Wulff-Nilsen and Mikkel Thorup}

       \vspace{1cm}
        This thesis has been submitted to the PhD School of The Faculty of Science,\\ \vspace{0.2cm}
         University of Copenhagen \\ \vspace{0.2cm} September, 2020
            
       \vspace{0.5cm}

   \end{center}
\end{titlepage}
\pagebreak
\pagenumbering{roman}

\begin{table}[h]
    \begin{tabular}{lp{11cm}}
         \textbf{Thesis Title:} & Near-Optimal Algorithms for Reachability, Strongly-Connected Components and Shortest Paths in Partially Dynamic Digraphs \\
         \textbf{Author:} & Maximilian Probst Gutenberg \\
         \textbf{Affiliation:} & University of Copenhagen, Faculty of Science, Department of Computer Science (DIKU) and Basic Algorithm Research Center (BARC).\\
         \textbf{Advisors:} & Christian Wulff-Nilsen (DIKU and BARC) and Mikkel Thorup (DIKU and BARC) \\
         \textbf{Thesis Committee:} & Valerie King (University of Victoria), Jakob Nordström (DIKU and Lund University), Uri Zwick (Tel Aviv University) \\
         \textbf{Date of Submission:} & September, 2020
    \end{tabular}
\end{table}

\clearpage
\pagebreak

\vspace*{3cm} 
\begin{center}
    \textit{To my wife, Johanna.}
\end{center}

\pagebreak

\pagebreak

\section*{Abstract}

In this thesis, we present new techniques to deal with fundamental algorithmic graph problems where graphs are directed and partially dynamic, i.e. undergo either a sequence of edge insertions \emph{or} deletions:
\begin{itemize}
    \item Single-Source Reachability (SSR): given a distinct source vertex $r$ in a graph, the objective is to maintain the set of vertices that $r$ can reach throughout the entire update sequence.
    \item Strongly-Connected Components (SCC): the goal is to maintain a partition of the vertex set $X_1, X_2, \dots, X_k$, such that every two vertices in the same partition set $X_i$ are on a common cycle, while no two vertices across different partition sets do.
    \item Single-Source Shortest Paths (SSSP): given a dedicated source vertex $s$, the objective is to maintain the distance from $s$ to every other vertex in the graph.
\end{itemize}

These problems have recently received an extraordinary amount of attention due to their role as subproblems in various more complex and notoriously hard graph problems, especially to compute flows, bipartite matchings and cuts.

Our techniques lead to the first near-optimal data structures for these problems in various different settings. Letting $n$ denote the number of vertices in the graph and by $m$ the maximum number of edges in any version of the graph, we obtain
\begin{itemize}
    \item  the first randomized data structure to maintain SSR and SCCs in near-optimal total update time $\tilde{O}(m)$ in a graph undergoing edge deletions.
    \item  the first randomized data structure to maintain SSSP in partially dynamic graphs in total update time $\tilde{O}(n^2)$ which is near-optimal in dense graphs.
    \item  the first deterministic data structures for SSR and SCC for graphs undergoing edge deletions, and for SSSP in partially dynamic graphs that improve upon the $O(mn)$ total update time by Even and Shiloach from 1981 that is often considered to be a fundamental barrier. 
\end{itemize}

\pagebreak

\section*{Abstrakt}

I denne afhandling præsenterer vi nye teknikker til at håndtere grundlæggende algoritmiske grafteoretiske problemer, hvor grafer er orienterede og delvist dynamiske, dvs. enten gennemgår en sekvens af kantindsættelser \emph{eller} sletninger:
\begin{itemize}
    \item Single-Source Reachability (SSR): Givet en særskilt kilde $r$ i en graf, er målet at opretholde et mængde af de knuder, $r$ kan nå gennem hele opdateringssekvensen.
    \item Strongly-Connected Components (SCC): Målet er at opretholde en partition af knudemængder $ X_1, X_2, \dots, X_k $, således at hvert par af knuder i samme partitionssæt $ X_i $ er på en fælles kreds, mens intet par af knuder i forskellige partitionsmængder har denne egenskab.
    \item Single-Source Shortest Paths (SSSP): Givet en dedikeret kilde $ s $, er målet at opretholde afstanden fra $ s $ til hvert enhver anden knude.
\end{itemize}

Disse problemer har for nylig fået ekstraordinær opmærksomhed på grund af deres rolle som delproblemer i forskellige mere komplekse og notorisk hårde grafproblemer, især til beregning af flows, bipartite matchings og cuts.

Vores teknikker fører til de første næsten-optimale datastrukturer til disse problemer i forskellige scenarier. For $n$, den antal af knuder i grafen og $m$, den maximale antal kanter i enhver version af grafen, vi præsenterer
\begin{itemize}
    \item den første randomiserede datastruktur, der opretholder SSR og SCC'er i næsten optimal total opdateringstid $ \tilde {O}(m) $ i en graf, der gennemgår kantsletninger.
    \item den første randomiserede datastruktur, der opretholder SSSP i delvist dynamiske grafer i den samlede opdateringstid $  \tilde{O}(n^2) $, hvilket er næsten optimalt i tætte grafer.
    \item de første deterministiske datastrukturer for SSR og SCC for grafer, der gennemgår kantsletninger, og for SSSP i delvist dynamiske grafer, der forbedrer den samlede opdateringstid $ O(mn)$ af Even og Shiloach fra 1981.
\end{itemize}

\pagebreak

\section*{Acknowledgments}

I would like to start my acknowledgment section with the person that clearly  had the most influence on me throughout my Ph.D. and who made it all possible in the first place: my principal advisor Christian Wulff-Nilsen. 

It was first through him that I got to know research in Theoretical Computer Science, and it was him who led me through my first projects as a master student which even culminated in a Best Student Paper at ESA 2018 \cite{probst2018complexity}. 

From the beginning of my Ph.D., Christian has always encouraged me to work on hard problems, but never left me alone with them even in the face of my many initial misconceptions and mistakes. During the first year, we often worked together in long sessions and nothing impresses me more - still nowadays - than Christian's incredible determination to solve problems even after severe setbacks. And indeed, after some initial setbacks, we developed alone in the first year new algorithms for worst-case update time all-pairs shortest paths \cite{probstWulffNilsenwcAPSP}, decremental SSSP in directed graphs \cite{GutenbergW20a} and deterministic decremental SSSP in undirected graphs \cite{DetDecrementalSSSP}. All of which, were accepted to SODA 2020. 

At the end of the first year, Christian introduced me to Aaron Bernstein who visited Christian for a week at the time. I had previously tried to improve an algorithm for decremental SSR and SCC, without much luck, but we decided to work on it and try some new approaches. From there, it only took us a few days to get through and solve the problem and we spend the last days going pedantically (but with great enthusiasm) through every proof while eating as much Sushi as Christian's Funding would cover. The paper appeared later in STOC 2019 \cite{bernstein2019decremental} and even got an invitation to the SICOMP special issue. 

In January 2019, I took to Boston, where I was more than fortunate to be hosted by Virginia Vassilevska Williams at MIT, much due to the help of my advisor Mikkel Thorup who initiated the contact and to STIBOFONDEN who financed the trip. The months I spent at MIT were among the most interesting of my entire Ph.D.; I met a whole range of open and interesting people to discuss research with, to cycle around Boston, play soccer or to go out for some dinner on Cambridge Street (special thanks to Anders Aamand, Yuval Dagan and Siddhartha Jayanti for making this time so pleasant, but also to all the others I met there!). I also met Nicole Wein, and we worked hard that spring together with Virginia to get some exciting new results for the incremental SSSP problem (and even harder that summer to write up the result). The resulting paper \cite{GutenbergWW20} was later published at STOC 2020.

Before leaving the U.S., I invited myself to be hosted by Aaron (and for the record, I am his first ever visitor!), who not only agreed to be visited for almost two weeks but also got Thatchaphol Saranurak to join us. It is hard to overestimate how many great papers were started during this time. When we left, we had discussed the algorithmic ideas for an exciting new cut sparsifier, an algorithm to compute a directed expander decomposition and a nice proof sketch for the first adaptive decremental SSR, SCC and SSSP data structures that improve upon the classic ES-tree. 

It took only a few weeks after coming back to Copenhagen to find various holes in our arguments, and in the end, we spend almost the entire next year on overcoming shortcomings in our ideas, reinventing components that did not work and pushing even further. The paper \cite{detDiSSSP} accepted to FOCS 2020 was one of the results of this hard (but incredibly fun) work, as well, as the paper \cite{bernstein2020fully} which is currently under submission after merging our cut sparsifier with a technique developed by Jan van den Brand, Danupon Nanongkai, Aaron Sidford and He Sun.

Further, upon returning to Copenhagen, Aaron, Christian and I decided to work on the decremental SSSP algorithm for directed graphs presented at SODA 2020 and throughout the next months, we developed a new perspective on the previous algorithm and finally were able to push through all the way to obtain a near-optimal algorithm for the problem in dense graphs. The resulting paper was accepted recently to FOCS 2020 \cite{nearOptDenseSSSP}. 

This concludes a sketch of the history behind the papers that were published during the Ph.D. project. I am very grateful to have met all the great people that have helped to achieve these exciting results, make this an amazing experience and who have taught me a lot along the way. Much in the same way I want to acknowledge all the great people that are co-authors on manuscripts and who I have worked together with over the years: Thiago Bergamaschi, Debarati Das, Thomas Dueholm Hansen, Jacob Evald, Viktor Fredslund-Hansen, Siddhartha Jayanti, Monika Henzinger, Nikos Parotsidis, Slobodan Mitrović. 

I also want to thank Mikkel Thorup who is one of the people who really inspired me in a way that made me want to go for this Ph.D. in the first place and for his general support throughout these years. I am much indebted to Uri Zwick, Virginia and Christian for writing me reference letters that helped to land a Post Doc position at ETH Zurich with Rasmus Kyng. I would further like to thank Uri and Valerie King and Jakob Nordström, for being on the Thesis Committee. I want to thank the people at BARC for creating an environment for high quality research that has helped to meet a lot of really interesting people. I also want to thank both Danupon Nanongkai and Fabian Kuhn for inviting me to give presentations in Stockholm and Freiburg, who were both beyond hospitable.

Finally, I would like to thank all the people who helped me to start the next chapter of my career, in particular, I would like to thank Rasmus Kyng, Sushant Sachdeva, Nikos Parotsidis, Vincent Cohen-Addad and Sebastian Forster. 

\pagebreak

\tableofcontents

\pagebreak
\pagenumbering{arabic}

\chapter{Introduction}

\input{intro}

\chapter{Preliminaries}

\input{prelim}

\chapter{A Randomized Algorithm for Decremental SSR and SCCs}

\input{rand_scc}

\chapter{A Randomized Algorithm for Decremental SSSP in Dense Digraphs}

\input{rand_decr_sssp}

\chapter{Deterministic Algorithms for Decremental SSR, SCC and SSSP}

\input{det_scc}

\chapter{A Deterministic Algorithm for Incremental SSSP}

\input{incr_sssp}

\chapter{Conclusion}
\label{chap:conclusion}

In this thesis, we make substantial progress on the partially dynamic SSR, SCC and SSSP problems in various settings. 

In particular, we gave the first near-optimal algorithms for decremental SSR and SCC in general graphs, and the first near-optimal algorithms for partially dynamic $(1+\epsilon)$-approximate SSSP in very dense graphs. 

Further, we also give the first \emph{deterministic} data structures to improve upon the $O(mn)$ total update barrier in the hardest of all settings: partially dynamic SSSP. For the incremental setting, our data structure is even near-optimal for very dense graphs. In the decremental setting, where a $o(mn)$ total update time data structure was not even known for SSR and SCC, we give total update time $mn^{2/3+o(1)}$ for the decremental SSR and SCC and $n^{2+2/3+o(1)}$ update time for decremental SSSP\footnote{Much like in \Cref{chap:intro}, we assume in this discussion that $W$ is polynomial and $\epsilon$ is constant}.

These significant improvements of the state-of-the-art data structures also motivates the following questions and open problems:

\paragraph{Near-Linear Time Algorithms for Sparse Graphs.} While we give the first algorithms that are near-optimal for very dense graphs for partially dynamic SSSP meaning that there total update time is $\tilde{O}(n^2)$, this gives no improvement on very sparse graphs, i.e. graphs where $m = O(n)$. Even though \cite{nearOptDenseSSSP} and \cite{henzinger2014sublinear} also present data structures for sparse graphs with moderate improvements over the $O(mn)$-barrier, both data structures leave much to desire and designing data structures that are near-optimal for any graph density is a major open problem. 

We point out that even beating the $O(m\sqrt{n})$ barrier is a major open problem since all partially dynamic SSSP data structures designed for sparse graphs rely on hopset techniques such that the input graph $G$ is augmented by weighted graph $H$ where in $G \cup H$, every pair of vertices $u,v\in V$, have a path consisting of a sublinear number $h$ of edges that $(1+\epsilon)$-approximates the shortest path from $u$ to $v$ in $G$. While in undirected graphs, a hopset is given and maintained that has $h = n^{o(1)}$ which is then exploited to derive a $m^{1+o(1)}$ total update time algorithm, such a bound on $h$ in directed graphs is not possible. Currently, the best lower bounds for directed graphs achieve $h = \Omega(n^{1/17})$ \cite{hesse2003directed}. On the other hand, the best existential upper bound on $h$ in directed graphs is $\tilde{\Theta}(\sqrt{n})$ achieved by a rather trivial hopset construction (for an efficient construction of such a hopset we refer the reader to \cite{fineman2018nearly, liu2019parallel, cao2020improved}). Thus, breaking the $O(m\sqrt{n})$ barrier for partially dynamic SSSP in directed graphs would either require a departure from the hopset approach or an improvement on the existential properties of hopsets in directed graphs.

We also point out that there exists no near-linear time algorithm for the incremental SCC problem. Especially after learning that such an algorithm exists for the decremental version of the problem, it seems likely that such an algorithm exists. However, considerable effort was spend by a large number of researchers on the problem and any progress over the recent $\Otil(\min\{m^{4/3}, m\sqrt{n}\})$ total update time data structures \cite{bernstein2018incremental, bhattacharya2018improved} would receive enormous attention in the field.

\paragraph{Deterministic/ Adaptive Algorithms.} The current state of the art leaves a fundamental gap between the current state-of-the-art non-adaptive randomized data structures for partially dynamic SSR, SCC and SSSP and their deterministic/adaptive counterparts. Resolving this gap most likely requires major new techniques and might serve as a test bed to developing more general approaches to derandomize dynamic (directed) graph algorithms.

We believe that the problem of obtain deterministic/ adaptive decremental SSR and SCC data structures with near-linear update time might be the most natural direction to pursue since the randomized setting is understood. However, we point out that the techniques given in \Cref{chap:rand_scc} and \Cref{chap_deterministic_scc} seem rather incompatible and most likely, a completely new approach is required to achieve this goal.

The second of these problems that we want to emphasize is the decremental SSSP problem. As shown recently in the decremental setting \cite{Chuzhoy:2019:NAD:3313276.3316320}, decremental SSSP might be a feasible approach to solve complex flow problems. Obtaining a deterministic/ adaptive data structure that matches the update time of the randomized decremental SSSP data structure given in \Cref{chap_rand_decr_sssp} would most likely give an $\tilde{O}(n^2)$ time algorithm to solve the static Exact Maximum Flow problem. This would be a major breakthrough. Even under the consideration that other approaches (for example \cite{bipartiteMathching}) might achieve this running time before, it might still be of considerable interest since the reductions given in  \cite{Chuzhoy:2019:NAD:3313276.3316320} might allow for various more generalizations of the maximum flow problem to be solved efficiently. 

\paragraph{All-Pairs Shortest Paths.} Another problem of major interest is the partially dynamic $(1+\epsilon)$-approximate All-Pairs Shortest Path problem. Here, the decremental APSP problem is solved in \cite{bernstein2016maintaining} with $\tilde{O}(mn)$ total update time which is almost optimal since it matches the best running time for static "combinatorial" APSP algorithms up to subpolynomial factors. The incremental version of the problem is still not fully understood but was recently considered in \cite{karczmarz2019reliable}. 

Both algorithms are however heavily randomized. In \cite{karczmarz2020simple}, a folklore result for decremental APSP was given with total update time $\tilde{O}(n^3)$. But for very sparse graphs, no improvement over an almost trivial $\tilde{O}(mn^2)$ total update time bound has been given (for example using \cite{demetrescu2004new}). Matching the $\tilde{O}(mn)$ near-optimal total update time bound in the deterministic setting thus remains a major open problem.

\paragraph{Deterministic/ Adaptive Data Structures for Undirected Graphs.} Finally, there has been a long line of research \cite{bernstein2016deterministic, bernstein2017deterministic, bernstein2017deterministicweighted, Chuzhoy:2019:NAD:3313276.3316320, DetDecrementalSSSP, bernstein2020fully} that aims at obtaining a deterministic/ adpative near-linear $m^{1+o(1)}$ total update time algorithm for the partially dynamic SSSP problem in undirected graphs. For very dense graphs, such a data structure that works against an adaptive adversary was first given by \cite{Chuzhoy:2019:NAD:3313276.3316320} and is currently used in state-of-the-art algorithms to compute various flow problems in undirected graphs to an $(1+\epsilon)$-approximation. However, as pointed out in \cite{bernstein2017deterministic, DetDecrementalSSSP} breaking the $O(m\sqrt{n})$ barrier for sparse graphs likely requires a new set of techniques.

\pagebreak

\printbibliography[heading=bibintoc] 

\end{document}

%% file: commands.tex
\allowdisplaybreaks[1]

\makeatletter
\g@addto@macro\bfseries{\boldmath}
\makeatother

\makeatletter
\g@addto@macro\mdseries{\unboldmath}
\g@addto@macro\normalfont{\unboldmath}
\g@addto@macro\rmfamily{\unboldmath}
\g@addto@macro\upshape{\unboldmath}
\makeatother

\DeclareCiteCommand{\citem}
    {}
    {\mkbibbrackets{\bibhyperref{\usebibmacro{postnote}}}}
    {\multicitedelim}
    {}

\renewcommand*{\multicitedelim}{\addcomma\space}

\AtEveryCitekey{\clearfield{note}}

\newcommand{\myhref}[1]{%
  \iffieldundef{doi}
    {\iffieldundef{url}
       {#1}
       {\href{\strfield{url}}{#1}}}
    {\href{http://dx.doi.org/\strfield{doi}}{#1}}%
}

\DeclareFieldFormat{title}{\myhref{\mkbibemph{#1}}}
\DeclareFieldFormat
  [article,inbook,incollection,inproceedings,patent,thesis,unpublished]
  {title}{\myhref{\mkbibquote{#1\isdot}}}

\makeatletter
\AtBeginDocument{%
    \newlength{\temp@x}%
    \newlength{\temp@y}%
    \newlength{\temp@w}%
    \newlength{\temp@h}%
    \def\my@coords#1#2#3#4{%
      \setlength{\temp@x}{#1}%
      \setlength{\temp@y}{#2}%
      \setlength{\temp@w}{#3}%
      \setlength{\temp@h}{#4}%
      \adjustlengths{}%
      \my@pdfliteral{\strip@pt\temp@x\space\strip@pt\temp@y\space\strip@pt\temp@w\space\strip@pt\temp@h\space re}}%
    \ifpdf
      \typeout{In PDF mode}%
      \def\my@pdfliteral#1{\pdfliteral page{#1}}
      \def\adjustlengths{}%
    \fi
    \ifxetex
      \def\my@pdfliteral #1{}
      \def\adjustlengths{\setlength{\temp@h}{-\temp@h}\addtolength{\temp@y}{1in}\addtolength{\temp@x}{-1in}}%
    \fi%
    \def\Hy@colorlink#1{%
      \begingroup
        \ifHy@ocgcolorlinks
          \def\Hy@ocgcolor{#1}%
          \my@pdfliteral{q}%
          \my@pdfliteral{7 Tr}
        \else
          \HyColor@UseColor#1%
        \fi
    }%
    \def\Hy@endcolorlink{%
      \ifHy@ocgcolorlinks%
        \my@pdfliteral{/OC/OCPrint BDC}%
        \my@coords{0pt}{0pt}{\pdfpagewidth}{\pdfpageheight}%
        \my@pdfliteral{F}
        \my@pdfliteral{EMC/OC/OCView BDC}%
        \begingroup%
          \expandafter\HyColor@UseColor\Hy@ocgcolor%
          \my@coords{0pt}{0pt}{\pdfpagewidth}{\pdfpageheight}%
          \my@pdfliteral{F}
        \endgroup%
        \my@pdfliteral{EMC}%
        \my@pdfliteral{0 Tr}
        \my@pdfliteral{Q}%
      \fi
      \endgroup
    }%
}
\makeatother

%% file: intro.tex
\label{chap:intro}
\section{Problem Statement}

In this thesis, we are concerned with three algorithmic graph problems where a weighted directed graph $G=(V,E,w)$ is inputted and a property of the graph is computed and returned by an algorithm. 

\paragraph{Single-Source Reachability:} In the Single-Source Reachability (SSR) problem, additionally to the weighted digraph $G=(V,E,w)$, a dedicated source vertex $r \in V$ is inputted, and the goal is to compute the set of reachable vertices $R$ such that for every vertex $v \in R$, $r$ can \emph{reach} vertex $v$, i.e. there exists a path from $r$ to $v$ in $G$, and for every vertex $v \not\in R$, $r$ cannot reach $v$. 

\paragraph{Strongly-Connected Components:} In the Strongly-Connected Components (SCC) problem, the goal is to output a partition of the vertex set $V$ into sets $X_1, X_2, \dots, X_k$ for some $k$ such that for every pair of vertices $u \in X_i$ and $v \in X_j$, if $i=j$, then $u$ can reach $v$ and vice versa, in which case we say that $u$ and $v$ are \emph{strongly-connected}, and if $i \neq j$, $u, v$ are not strongly-connected.

\paragraph{Single-Source Shortest Paths:}  In the Single-Source Shortest Paths (SSSP) problem, an additional vertex $r \in V$, is inputted and the goal is to return the distance $\dist_G(r,v)$ from $r$ to $v$ , for each $v \in V$.

These problems are among the most fundamental problems in Computer Science due to their various applications in navigation and networks and their role as subroutines in more complex graph algorithmic problems like Maximum Flow and Bipartite Matching \cite{tarjan1983data, cormen2009introduction}. Therefore, they are typically covered in undergraduate courses \cite{cormen2009introduction} and are subject to intensive research in theory and practice. Letting $n = |V|$ and $m = |E|$, the SSSP problem was solved already in 1959 by Dijkstra's algorithm \cite{dijkstra1959note} which only requires running time $O(m \log n)$. For SSR and SCC even faster algorithms are known due to Tarjan \cite{tarjan1972depth} since 1972 that achieve running time $O(m)$.

In this thesis, we are interested in the generalization of these graph algorithmic problems to consider a graph $G=(V,E,w)$ that is subject to edge insertions/ deletions and thus there are different \emph{versions} of $G$ where two consecutive versions only differ in a single edge. Such a graph $G$ is called a \emph{dynamic} graph in contrast to a \emph{static} graph $G$ that does not undergo changes. The corresponding dynamic problem versions transform the static algorithmic problems into data structure problems. We define the data structure problems as follows.

\paragraph{Dynamic SSR, SCC and SSSP.} In the dynamic SSR problem, the goal is to maintain a data structure $\mathcal{E}_r$ that supports the operations:
\begin{itemize}
    \item $\textsc{Preprocess}(G, r)$: The data structure is initialized by inputting an initial weighted digraph $G$, and a source vertex $r \in V$.
    \item $\textsc{Insert}(u,v)$: The data structure processes that the edge $(u,v)$ is inserted into $G$.
    \item $\textsc{Delete}(u,v)$: The data structure processes that the edge $(u,v)$ is deleted from $G$.
    \item $\textsc{Query}(v)$: The data structure returns whether $r$ can reach $v$ in the current version of $G$.
\end{itemize}
In the dynamic SCC problem, a similar data structure is maintained with the difference that an initial source vertex $r$ does not have to be inputted and the query operation takes two vertices $u,v \in V$ as input and returns whether $u,v$ are strongly-connected or not.

In the dynamic SSSP problem, a similar data structure is maintained with the key difference that on query, it returns the distance from $r$ to $v$ instead of outputting whether $r$ reaches $v$.

We point out that two trivial implementations for each kind of dynamic problem exist: we can rerun the static algorithm on the current version of graph $G$ after \emph{every} update, which leads to high update times but constant query times if implemented carefully. The other extreme is to use a lazy approach where updates are not processed at all but only stored and upon a query operation, the static algorithm is invoked on the current version of $G$. This leads to constant update times but high query times. The goal of the area of \emph{dynamic graph algorithms} is to improve this trivial trade-off by reusing the information computed for a graph $G$ to compute fundamental graph properties in $G$ after the insertion/ deletion in small time.\\

In this thesis, we restrict our attention to the setting where the dynamic graph $G$ only undergoes insertions, in which case we say that $G$ is \emph{incremental}, or exclusively undergoes edge deletions, in which case we say that $G$ is \emph{decremental}. If $G$ is either \emph{incremental} or \emph{decremental}, we also say that $G$ is \emph{partially dynamic} which is in contrast to the \emph{fully dynamic} setting where a mix of insertions and deletions is allowed. We also use these graph descriptors to further specify the data structure problem that we are concerned with, for example, the incremental SSSP problem, refers to the data structure problem of dynamic SSSP where the operation $\textsc{Delete}(u,v)$ does not have to be implemented since $G$ is guaranteed to be incremental. We also point out that we often say that there is an \emph{algorithm} for the incremental SSSP problem, which refers to the algorithm that maintains the underlying data structure.

\paragraph{Adversary Model.} We distinguish between two adversary models: we say that the adversary is \emph{non-adaptive} if the sequence of updates to a graph $G$ is already determined by the adversary before the data structure is initialized (however, this update sequence is still only revealed one-by-one). In contrast, we say the adversary is \emph{adaptive} if it creates the update sequence on-the-go, i.e. determines the $(i+1)^{th}$ update to $G$ only after the data structure has processed update $i$. We point out that a data structure that is deterministic, i.e. does not make any use of randomization, works in the model of an adaptive adversary. Informally, we say that a data structure is non-adaptive/adaptive if it works in the non-adaptive/adaptive adversary model, and we also sometimes say it works \emph{against} a non-adaptive/adaptive adversary.

\section{Motivation}

Partially dynamic SSR and SCC are of great interest for their direct applications especially for detecting cycles or unreachable vertices in dynamic dependency graphs. The incremental SCC problem is further a generalization of the famous incremental Cycle Detection problem which has practical applications in pointer analysis and circuit evaluation (see \cite{Haeupler12} for a discussion of these applications). The decremental SCC problem has further applications to compute Streett Objectives in Graphs and Markov Decision Processes \cite{chatterjee2019near}, and computing Weak Bisimilarity on Markov Chains \cite{jansen_et_al:LIPIcs:2020:12820}. 

\noindent The partially dynamic SSSP has wide-ranging applications:
\begin{itemize}
    \item Partially dynamic SSSP data structures are often used as internal data structures to solve fully dynamic SSSP and fully dynamic All-Pairs Shortest Paths (APSP), see for example \cite{king1999fully, roditty2004dynamic, henzinger2016dynamic}, which in turn can be used to maintain properties of real-world graphs undergoing changes.
    \item Partially dynamic SSSP is often employed as internal data structure for more complex dynamic algorithmic problems such as maintaining the diameter in partially dynamic graphs \cite{ancona2018algorithms, DBLP:journals/corr/abs-1812-01602} or matchings in incremental bipartite graphs \cite{bernstein2018online}.
    \item Many static algorithms use partially dynamic SSSP algorithms as a subroutine. In particular, a recent line of research shows that many flow problems can be efficiently reduced to decremental SSSP, and recent progress has already led to faster algorithms for multi-commodity flow \cite{madry2010faster}, vertex-capacitated flow, and sparsest vertex-cut \cite{Chuzhoy:2019:NAD:3313276.3316320, ChuzhoyGLNPS_det_cut} and min-cost flow problems \cite{ChuzhoyS20_apsp}. These reductions are based on algorithms by Garg and Könemann \cite{garg2007faster} and \cite{Fleischer00} who refined the multiplicative update framework to flow problems and use their duals, shortest path problems, to make progress. 
\end{itemize}

Further, we also point out that the theory for Dynamic SSR, SCC and SSSP problems has led to a plethora of new theoretical concepts and general methods. For example, the quest for deterministic Decremental SSR, SCC and SSSP data structures has recently driven theory  to develop a useful notion of directed expanders \cite{detDiSSSP} and for generalizing many of the tools from the undirected setting where expander decompositions (see \cite{SaranurakW19}) has recently fuelled many exciting new algorithms for a broad variety of problems \cite{chu2020graph, saranurak2020simple, forster2020computing}. 

Another example is the development of fast and robust dynamic sparsifiers \cite{bernstein2020fully} which were developed to remedy problems in the theory of Decremental SSSP algorithms for undirected graphs and have since been used in the recent breakthrough result \cite{bipartiteMathching} for static bipartite matchings, negative-weight shortest paths and the transshipment problem. 

We note that in order to be used in applications in a black-box fashion, a data structure has to work in the adaptive adversary model. While a data structure that only works in the non-adaptive adversary model might still be of use to applications, all of the above applications use adaptive data structures even though non-adaptive data structures often have significantly faster update times.

\section{Previous Work}

In our review of previous work, we focus on the dynamic versions of the problems in general directed graphs. For readers interested in the static settings, we refer the reader to the excellent article by Tarjan \cite{tarjan1972depth} and two recent surveys on SSSP \cite{sommer2014shortest, madkour2017survey}. We further point out that in adherence to convention in the field, we state results for partially dynamic algorithms as the cumulative running time over the total update sequence (excluding queries) and point out that all data structures of interest ensure constant query time. We let $n$ denotes the number of vertices, $m$ denote the maximum number of edges in any version of the graph (and we will assume that $m \geq n$), and $W$ be the weight ratio of the graph, that is, the ratio of largest edge weight by smallest edge weight and which for simplicity we assume to be bound by $n^c$ for some large constant $c > 0$ for the rest of this section. We also use $\tilde{O}$-notation and $\Omegatil(\cdot)$ to hide logarithmic factors. Similarly, we use $\Ohat(\cdot)$ and $\Omegahat(\cdot)$ to hide $n^{o(1)}$ factors.

\paragraph{Fully Dynamic SSR, SCC and SSSP.} While the fully dynamic setting, that allows for a mix of edge insertions and deletions is clearly a much more general setting than the restricted partially dynamic setting, a series of extremely strong conditional lower bounds for this model renders polynomial improvements over the trivial update/ query trade-off achieved by recomputing from scratch after every update or during every query as unlikely \cite{roditty2004dynamic, abboud2014popular,henzinger2015unifying}.

\paragraph{The ES-tree for partially dynamic SSR and SSSP.} The first improvement over the trivial total update time of $\tilde{O}(m^2)$ for partially dynamic SSR and SSSP was given by Even and Shiloach \cite{shiloach1981line} in 1981 which has total update time $O(mnW)$ and thereby improves over the trivial approach for dense graphs of small weight ratio. The key idea behind the data structure is a simple trick that allows to maintain, from a fixed root vertex $r$, a BFS tree of depth $\delta$ with total update time $O(m\delta)$ where the above running time follows since $\delta \leq nW$. Therefore the data structure is often referred to as the \emph{ES-tree}. We point out that originally, the data structure by Even and Shiloach was devised to work only for undirected, unweighted graphs, however, Henzinger and King \cite{henzinger1995fully, King99} later realized that small modifications suffice to extend the data structure to maintain a directed BFS out-tree from a fixed source $r$ in a directed graphs and thereby made clear that it can be used to certify reachability and to maintain shortest paths. Ever since, the ES-tree has been one of the most widely used and adapted data structures in dynamic graph algorithms and we therefore give a full introduction to the algorithm in \Cref{sec:esTree}.

\paragraph{Incremental SSR and SCC.} For incremental SSR, it is rather straight-forward to improve the total update time to $\tilde{O}(m)$ total update time by using a dynamic tree data structure (see for example \cite{tarjan1983data, alstrup2005maintaining}). 

For the incremental SCC problem, two classic results exists, the first by Haeupler et al. \cite{haeupler2008faster, Haeupler12} gives total update time $\tilde{O}(\min\{mn^{2/3}, m^{3/2}\})$ while the second approach by Bender et al. \cite{bender2009new, bender2016new} achieves total update time $\tilde{O}(n^2)$. While this still marks the state-of-the-art, we point out that various new approaches \cite{cohen2013labeling, bernstein2018incremental, bhattacharya2018improved} exist for the incremental cycle detection problem where rather than maintaining the SCCs in the incremental graph $G$, the data structure only has to report when the first non-trivial SCC forms. However, even for this simpler problem, the current state-of-the-art data structures cannot match the near optimal running time bound for incremental SSR.

\paragraph{Decremental SSR and SCC.} The first non-trivial data structure for Decremental SSR was the ES-tree with total update time $O(mn)$. Roditty and Zwick further observed in \cite{roditty2008improved} that placing a Decremental SSR data structure at a root $r$ and run it on $G$ and the reverse graph $G^{(rev)}$, the data allowed to maintain the set of all vertices that are strongly-connected to $r$. Further, by using a clever random root trick, they derived a reduction from Decremental SCC to Decremental SSR at the additional cost in total update time of only $O(\log n)$, thus, obtaining total update time $\tilde{O}(mn)$ for the SCC problem. In \cite{lkacki2013improved}, {\L}{\k{a}}cki gave a different approach that decomposed the problem of Decremental SCC into up to $n$ subproblems of Decremental SSR on directed \emph{acyclic} graphs (DAGs) which can be solved in total time $O(m)$ \cite{italiano1988finding}. This constituted the first deterministic algorithm for Decremental SCC with total update time $O(mn)$. This data structure was later further extended by Georgiadis et al. \cite{georgiadis2017decremental} to handle more advanced query operations.

The total running time of $O(mn)$, often seen as a fundamental barrier for partially dynamic SSR, SCC and SSSP problems, was finally broken by Henzinger et al. \cite{henzinger2014sublinear, henzinger2015improved} to $mn^{0.9+o(1)}$. The authors developed a new algorithm for the Decremental SSR based on clever hitting set techniques. Briefly after, Chechik et al. \cite{chechik2016decremental} showed that a clever combination of the algorithms in \cite{roditty2008improved} and \cite{lkacki2013improved} can be used to improve the total update time to $\tilde{O}(m \sqrt{n})$. 

\paragraph{Partially Dynamic SSSP.} While improvements on the upper bounds given by the ES-tree was made for all the partially dynamic SSR and SCC problems, Roditty and Zwick showed in \cite{roditty2004dynamic} that on undirected, weighted graphs, $\tilde{\Omega}(mn)$ total running time is required, or otherwise substantial progress on the static APSP problem was possible. This lower bound was latter strengthened by Abboud and Vassilevska Williams in \cite{abboud2014popular} and extended to hold even using algebraic techniques although with a weaker lower bound as proven by Henzinger et al. \cite{henzinger2015unifying}.

Since handling graphs with large weights is highly desirable for applications, research focused subsequently on the relaxed problem of only reporting distance to an $(1+\epsilon)$-approximation. A simple rounding trick (probably first stated in \cite{Madry10, bernstein2009fully}) can be used to transform the ES-tree into a data structure that achieves total running time $\tilde{O}(mn)$ for $(1+\epsilon)$-approximate SSSP. While Henzinger et al. \cite{henzinger2014sublinear, henzinger2015improved} showed that their approach could be extended to give $(1+\epsilon)$-approximate SSSP in partially dynamic graphs\footnote{While the data structure is only shown to work for decremental graphs, we believe that it can be adapted to the incremental setting.} in total update time $mn^{0.9+o(1)}$, this result remained state-of-the-art until this thesis.

\section{Related Work}

We also give a brief overview of related work. While we state only the best update bounds, we give citations of the current state-of-the art articles in various settings (for example, we sometimes cite the best result for the deterministic setting but only mention the bound achieved in the randomized setting for brevity).

\paragraph{Planar Directed Graphs.} For planar graphs, decremental algorithms are known to solve Single-Source Reachability deterministically in near-linear update time \cite{italiano2017decremental} and SSSP in directed graphs in total update time $\tilde{O}(n^{4/3})$ as shown by \cite{karczmarz2018decrementai}. The incremental versions of these problems have not received attention yet. The fully-dynamic SSSP problem was recently considered by Charalampopoulos and Karczmarz \cite{charalampopoulos_et_al:LIPIcs:2020:12897} who achieve worst-case update time $\tilde{O}(n^{4/5})$. 

\paragraph{Directed Dynamic All-Pairs Shortest Paths.} Decremental APSP was first considered by Baswana et al. \cite{baswana2007improved} where a data structure with $\tilde{O}(n^3)$ total update time is presented. The authors further present an algorithm for  $(1+\epsilon)$-approximate APSP which was subsequently improved to near-optimal total update time $\tilde{O}(mn)$ by Bernstein \cite{bernstein2016maintaining}. The techniques by Bernstein were further recently adapted to the incremental setting in \cite{karczmarz2019reliable}, however, only achieving total update time $\tilde{O}(mn^{4/3})$. In the fully dynamic setting, the state-of-the art data structure achieves amortized update time $\tilde{O}(n^2)$ \cite{demetrescu2004new, demetrescu2006fully, thorup2004fully} and the best worst-case update bound is currently $\tilde{O}(n^{2+2/3})$ \cite{thorup2005worst,abraham2017fully}. Again, allowing for a $(1+\epsilon)$-approximation in the distance estimates, a recent data structure \cite{brand2019dynamic} achieves worst-case update time $\tilde{O}(n^{2.045})$ ($\tilde{O}(n
^2)$ for undirected graphs). We also note that a substantial amount of research was dedicated to obtaining good dynamic APSP algorithms for planar digraphs where Fakcharoenphol and Rao \cite{fakcharoenphol2006planar} obtained an algorithm with update and query time $\tilde{O}(n^{2/3})$. Subsequently, Abboud and Dahlgaard \cite{abboud2016popular} obtained a conditional lower bound of $\tilde{\Omega}(\sqrt{n})$, however, this gap remains unresolved even after considerable effort \cite{kaplan2012submatrix, klein2005multiple, abraham2012fully, gawrychowski2018improved}. Finally, there also exists a wide literature on the topic of sensitivity oracles \cite{demetrescu2008oracles, bernstein2009nearly, baswana2015fault, chechik20171, choudhary2016optimal,baswana2019efficient, van2019sensitive, chechik2020distance} where a data structure is computed (often with large preprocessing time) to allow for few updates that can be processed very efficiently and then allows to answer all-pairs shortest paths queries.

\paragraph{Undirected Graphs.} In undirected graphs, the problem of connectivity, i.e. whether there exists an undirected path between two query vertices $u$ and $v$, is the pendant to the SSR and SCC problems. Fully dynamic connectivity is a well-studied problem where the fastest data structures take polylogarithmic (worst-case) update time \cite{holm2001poly, Thorup00, Wulff-Nilsen:2013:FDF:2627817.2627943, KapronKM13,nanongkai2017dynamicFirst, wulff2017fully, nanongkai2017dynamic}. For partially dynamic $(1+\epsilon)$-approximate SSSP, a breakthrough result by Henzinger et al. \cite{henzinger2014decremental} achieves total update time $m^{1+o(1)}$ and a recent goal in this setting has become to obtain deterministic or at least adaptive algorithms \cite{bernstein2016deterministic,bernstein2017deterministic, bernstein2017deterministicweighted, Chuzhoy:2019:NAD:3313276.3316320}. There further exist a plethora of work on the dynamic decremental APSP problem with varying approximation guarantees \cite{bernstein2011improved, henzinger2014decremental, henzinger2016dynamic, chechik2018near, karczmarz2020simple, ChuzhoyS20_apsp} and in the fully dynamic setting \cite{henzinger1995fully, King99, demetrescu2006fully, demetrescu2004new, roditty2004dynamic, thorup2005worst,  roditty2012dynamic, abraham2013dynamic, roditty2016fully, henzinger2016dynamic, abraham2017fully, brand2019dynamic} and we also point out that algebraic techniques were developed for both fully dynamic SSSP and APSP \cite{sankowski2005subquadratic, brand2019dynamic}.

\section{Contribution}
\label{sec:contribution}
In this thesis we present significant progress on the partially dynamic SSR, SCC and SSSP problems. In fact, we present the first data structures for all three problems that are near-optimal (although the data structures for SSSP are only near-optimal in very dense graphs).

\paragraph{A Near-Linear Update Time Randomized Algorithm for Decremental SSR and SCCs.} Our first contribution is an efficient randomized near-linear update time data structure for decremental SSR and SCC. This update time should be compared to the currently best total update time of $\tilde{O}(m\sqrt{n})$ \cite{chechik2016decremental}.

\begin{restatable}{theorem}{sccMain}[see \cite{bernstein2019decremental}]
\label{thm:ContributionSCCmain}
Given a decremental, directed graph $G=(V,E)$ and a dedicated source vertex $r \in V$, then we can maintain explicitly
\begin{itemize}
    \item the strongly-connected components of $G$, and 
    \item the set $R$ of vertices that are reachable from $r$.
\end{itemize}
The data structure runs in total expected update time $\tilde{O}(m)$ and works against an adaptive adversary, however, if path-queries are allowed it only works against a non-adaptive adversary.
\end{restatable}

\paragraph{Near-Optimal Partially Dynamic SSSP in Dense Digraphs.} Further, we present the first near-optimal data structure for partially dynamic SSSP in weighted directed graphs. This bound should be compared to the data structure by Henzinger et al. \cite{henzinger2014sublinear, henzinger2015improved}.

\begin{theorem}[see \cite{GutenbergW20a, nearOptDenseSSSP, GutenbergWW20}]
\label{thm:ContributionSSSPResult}
Given a partially dynamic input graph $G=(V,E,w)$, a dedicated source $r \in V$ and $\epsilon > 0$, there is a randomized algorithm that maintains a distance estimate $\widetilde{\mathbf{dist}}(r,x)$, for every $x \in V$, such that
\[
    {\mathbf{dist}}_G(r,x) \leq \widetilde{\mathbf{dist}}(r,x) \leq (1+\epsilon) {\mathbf{dist}}_G(r,x)
\]
at any stage w.h.p. The algorithm has total expected update time $\tilde{O}(n^2 \log W/\poly(\epsilon))$. Distance queries are answered in $O(1)$ time, and a corresponding path $P$ can be returned in $O(|P|)$ time. For the decremental version the data structure is non-adaptive, however, for the incremental version of the problem, our data structure is even deterministic.
\end{theorem}

We point out that in \cite{nearOptDenseSSSP}, a data structure for the decremental SSSP problem is given with total update time $\tilde{O}(mn^{2/3}\log W/\poly(\epsilon))$ which constitutes the best result on sparse digraphs. However, we only present the result on dense graphs in this thesis.

\paragraph{Deterministic Algorithms for Decremental SSR, SCC and SSSP.} Finally, we present the first \emph{deterministic} data structures for partially dynamic SSR, SCC and SSSP that improve upon the total update time of Even and Shiloach's data structure from 1981. In fact, this even is the first adaptive data structure that improves on the $O(mn)$ total running time achieved in \cite{shiloach1981line} when we want to support path-queries which are often crucial for applications. 

\begin{restatable}{theorem}{ContributionDetSCCResult}[see \cite{detDiSSSP}]
\label{thm:ContributionDetSCCResult}
Given a decremental, directed graph $G=(V,E)$, we can deterministically maintain the strongly-connected components of $G$ in total update time $mn^{2/3+o(1)}$. Further the data structure allows for constant time queries which upon inputting two vertices $u,v \in V$, outputs whether they are strongly-connected or not. 

The data structure can also, given a dedicated source vertex $r \in V$, maintain within the same update time the set of vertices reachable from $r$ and offers a constant time query which upon inputting a vertex $v \in V$, returns whether $r$ reaches $v$.
\end{restatable}

\begin{restatable}{theorem}{ContributionDetSSSPResult}[see \cite{detDiSSSP}]
\label{thm:ContributionDetSSSPResult}
Given a decremental dynamic input graph $G=(V,E,w)$, a dedicated source $r \in V$ and $\epsilon > 0$, there is a deterministic algorithm that maintains a distance estimate $\widetilde{\mathbf{dist}}(r,x)$, for every $x \in V$, such that
\[
    {\mathbf{dist}}_G(r,x) \leq \widetilde{\mathbf{dist}}(r,x) \leq (1+\epsilon) {\mathbf{dist}}_G(r,x)
\]
at any stage of the graph. The algorithm has total update time $n^{2+2/3+o(1)} \log W/\poly(\epsilon)$. Distance queries are answered in $O(1)$ time, and a corresponding path $P$ can be returned in $|P|n^{o(1)}$ time.
\end{restatable}

Both results for decremental SSSP (i.e. the results from \cite{nearOptDenseSSSP} and \cite{detDiSSSP}) build heavily upon the article \cite{GutenbergW20a} which was part of the Ph.D. project but is not featured in this thesis since its results were superseded by the results presented in this thesis. Further, we point out that \Cref{thm:ContributionDetSCCResult} and \Cref{thm:ContributionDetSSSPResult} indeed apply to partially dynamic graphs (instead of only decremental graphs) when considering the results from \cite{Haeupler12} and \Cref{thm:ContributionSSSPResult}.

\section{Contribution Beyond the Thesis} 

While the thesis focuses on the problems of partially dynamic SSR, SCC and SSSP in directed graphs, various results obtained during the Ph.D. project, have contributed to the current state of the art on dynamic graph algorithms. For completeness, we list here the contributions that are not covered by this thesis:
\begin{itemize}
    \item The article \cite{probstWulffNilsenwcAPSP} presents the first fully dynamic APSP algorithm with deterministic worst-case update time beyond the $\tilde{O}(n^{2+3/4})$ time bound given in 2015 by Thorup \cite{thorup2005worst}. The article further gives a randomized data structure that is more general than the current state-of-the-art data structure with best randomized worst-case update time $\tilde{O}(n^{2+2/3})$. Finally, it answers the question for a non-trivial dynamic APSP data structure with subcubic space as the latter algorithm only requires $\tilde{O}(n^2)$ space.
    \item The article \cite{DetDecrementalSSSP} presents an improvement for the decremental SSSP problem on undirected graphs. The achieved total update time of $mn^{0.5+o(1)}$ should be compared to the results by Chechik and Bernstein \cite{bernstein2016deterministic, bernstein2017deterministic} that achieve running time $\tilde{O}(\min\{mn^{3/4}, n^2\})$. Thus, the new data structure improves over the state-of-the-art whenever the graph $G$ is sufficiently sparse.
    \item The article \cite{bernstein2020fully} presents the first non-trivial algorithm to maintain graph sparsifiers against an adaptive adversary. In particular, on a fully dynamic graph, the data structure can maintain a $O(\polylog n)$-approximate spanner, cut-sparsifier or spectral-sparsifier with polylogarithmic update time.
    \item The article \cite{probst2018complexity} investigates a generalization of the classic Thorup-Zwick Distance Oracles \cite{thorup2005approximate, chechik2015approximate} where a static data structure is computed on a graph where each vertex is given a color (possibly many vertices receive the same color) and then has to answer queries given a vertex $v$ and a color $c$ for the (approximately) nearest $c$-colored vertex from $v$. The article demonstrates that the latter problem is strictly harder by establishing hardness in the cell-probe model, and gives an oracle with query times that are optimal up to a constant factor based on previous techniques from \cite{chechik2012improved, wulff2013approximate}.
    \item The article \cite{GutenbergWW20} of which excerpts are used in this thesis to derive \Cref{thm:ContributionSSSPResult}, also presents new conditional lower bounds for the exact weighted partially dynamic SSSP. The strongest among these lower bounds states that even in undirected graphs, for any sparsity $m$, there cannot be an exact algorithm, with preprocessing time $m^{2-\delta}$ and update and query time $m^{1-\delta}$ for any arbitrarily small constant $\delta > 0$, unless the $k$-Cycle Hypothesis is falsified. This lower bound also applies to algebraic algorithms which was previously not possible since lower bound reductions were based on the "combinatorial" static APSP conjecture.
    \item The article \cite{detDiSSSP} from which we present the proof of \Cref{thm:ContributionDetSCCResult} and \Cref{thm:ContributionDetSSSPResult}, presents a new framework for directed expanders and a new technique called congestion balancing. While we touch on these new developments partially in this thesis, we do not cover them in detail. We also do not cover a simple algorithm that is derived in \Cref{thm:ContributionDetSCCResult} by using the technique of congestion balancing in conjunction with a recent result by Wajc \cite{wajc2020rounding} and which gives the first $\tilde{O}(m)$ total update time algorithm for decremental $(1-\epsilon)$-approximate bipartite matchings.
\end{itemize}

\section{Organisation}

The rest of the thesis is concerned with presenting data structures as described in \Cref{sec:contribution}. We also give full proofs for all stated theorems, however, we will omit or only sketch several intermediate results. 

Let us now provide an overview of the thesis, chapter by chapter:
 
\paragraph{\Cref{chap:prelim}.} In this chapter, basic notation and preliminaries for the rest of the thesis are introduced where we try to adhere to conventions from the field of dynamic graph algorithms. Further, we give a brief introduction to the data structure by Even and Shiloach \cite{shiloach1981line} (the so called ES-tree) in \Cref{sec:esTree} which we recommend to readers not overly familiar with the data structure since many of the presented techniques rely on modifications and new insights to the classic ES-tree. 
    
We further give a short reduction in \Cref{sec:reductionWeights} which allows us to assume that the weight ratio of partially dynamic graphs $G$ is bound over all stages by a small polynomial in $n$, at the expense of introducing a $\log nW$ factor in the running time. While the proof is neither insightful nor novel, we state it here and use it implicitly in the rest of the thesis.

\paragraph{\Cref{chap:rand_scc}.} We then present a randomized data structure for the decremental SCC problem with expected near-linear update time $\tilde{O}(m)$. This data structure is straight-forwardly extended to also maintain SSR in decremental graphs. This gives the result stated in \Cref{thm:ContributionSCCmain}. The chapter is based on the STOC'2019 publication \cite{bernstein2019decremental}.

\paragraph{\Cref{chap_rand_decr_sssp}.} Building upon some of the techniques introduced in \Cref{chap:rand_scc}, we then present a randomized data structure for the decremental SSSP problems with expected update time $\tilde{O}(n^2)$. We point out however that the chapter is entirely self-contained and does not require the reader to have read \Cref{chap:rand_scc}, however, it does omit some proofs which are repetitive. 
This chapter is based on the SODA'2020 publication \cite{GutenbergW20a} and the FOCS'2020 publication \cite{nearOptDenseSSSP} which supersedes the results obtained in the former publication (the chapter introduces both approaches and highlights the improvements in the framework obtained in \cite{nearOptDenseSSSP}).

\paragraph{\Cref{chap_deterministic_scc}.} We present the first deterministic data structures that supersede the result by Even and Shiloach for decremental SSR, SCC and SSSP. We point out that these results require a new range of techniques for directed expanders. Since this thesis is focused on partially dynamic SSR, SCC and SSSP data structures, we do not cover proofs for these techniques, but rather give an intuitive approach to expanders and illustrate how to obtain the data structures described in \Cref{thm:ContributionDetSCCResult} and \Cref{thm:ContributionDetSSSPResult} from directed expander techniques. We then give a thorough proof of \Cref{thm:ContributionDetSSSPResult} which is based on the techniques from \Cref{chap_rand_decr_sssp} and \cite{GutenbergW20a}. The chapter is however self-contained as we give a brief (re)-introduction to the framework from \Cref{chap_rand_decr_sssp}.

\paragraph{\Cref{chap:incr_sssp}.} In this chapter, we present a deterministic incremental SSSP data structure with total update time $\tilde{O}(n^2)$. The data structure is inspired by a data structure for undirected graphs \cite{bernstein2017deterministic} however it requires many additional insights to make the same approach work in directed graphs. The chapter is based on the STOC'2020 publication \cite{GutenbergWW20}.

\paragraph{\Cref{chap:conclusion}.} Finally, we give a conclusion where we reflect on the contributions of this thesis and list a range of related open problems motivated by this work.

%% file: prelim.tex
\label{chap:prelim}
\section{Notation and Basic Definitions}

We let a graph $H$ refer to a weighted, directed graph with vertex set denoted by $V(H)$ of size $n_H$, edge set $E(H)$ of size $m_H$ and weight function $w_H : E(H) \rightarrow [1,W] \cup \{\infty\}$. We say that $H$ is a \emph{decremental} graph if it is undergoing a sequence of edge deletions and edge weight increases (also referred to as updates), and refer to \emph{version} $t$ of $H$, or $H$ at \emph{stage} $t$ as the graph $H$ obtained after the first $t$ updates have been applied. In this article, we denote the (decremental) input graph by $G=(V,E,w)$ with $n = |V|$ and $m = |E|$ (where $m$ refers to the number of edges of $G$ at stage $0$). In all subsequent definitions, we often use a subscript to indicate which graph we refer to, however, when we refer to $G$, we often omit the subscript.

\paragraph{Neighborhoods and Degree.}  We let $\mathcal{N}^{in}_H(v) = \{ u \;|\; (u,v) \in E(H) \}$ and $\mathcal{N}^{out}_H(v)  = \{ u \;|\; v \in \mathcal{N}^{in}_H(u)\}$ denote the in-neighborhood and out-neighborhoods of $v \in V$. We let the neighborhood $\mathcal{N}_H(v)$ of $v$ be defined by $\mathcal{N}_H(v) = \mathcal{N}^{in}_H(v) \cup \mathcal{N}^{out}_H(v)$. The weighted in-degree and out-degree
of a vertex $u$ are $\deg^{in}(u)=w(E(V,u))$ and $\deg^{out}(u)=w(E(u,V))$,
respectively. The weighted degree of $u$ is $\deg(u)=\deg^{in}(u)+\deg^{out}(u)$.
The volume of a set $S$ is $\vol(S)=\sum_{u\in S}\deg(u)$. 

\paragraph{Subgraphs, Cuts and Vertex Subsets.} We let $H[X]$ refer to the subgraph of $H$ induced by $X$, i.e. $H[X] = (X, E_H(X,X), w_H)$. We use $H \subseteq G$ to denote that $V(H) = V(G)$ and $E(H) \subseteq E(G)$. For graph $H$, and any two disjoint subsets $X,Y \subseteq V(H)$, we let $E_H(X)$ be the set of edges in $E(H)$ with an endpoint in $X$, and $E_H(X,Y)$ denote the set of edges in $E(H)$ with tail in $X$ and head in $Y$. We say that for $S \subseteq V$, $(S, V \setminus S)$ is a \emph{vertex cut} (sometimes simply \emph{cut}) and sometimes denote $V \setminus S$ by $\overline{S}$. For any $S$, let $\delta^{out}(S)=w(E(S,V\setminus S))$ and $\delta^{in}(S)=w(E(V\setminus S,S))$ denote the total weight of edges going out and coming in to $S$, respectively. 

\paragraph{Balanced Cuts.} We say that cut $(S,V\setminus S)$ is \emph{$\epsilon$-balanced} if $\vol(S)\ge\epsilon\vol(V)$, and it
is $\phi$-sparse if $\min\{\delta^{in}(S),\delta^{out}(S)\}<\phi\vol(S)$.
We say that $(L,S,R)$ is a \emph{vertex-cut} of $G$ if $L$,$S$, and $R$ partition the vertex set $V$,
and either $E(L,R)=\emptyset$ or $E(R,L)=\emptyset$. 
Assuming that $|L|\le|R|$,
$(L,S,R)$ is \emph{$\epsilon$-vertex-balanced} if $|L|\ge\epsilon|V|$,
and it is \emph{$\phi$-vertex-sparse} if $|S|<\phi|L|$. 

\paragraph{Contractions.} We define the graph $H / X$ to be the graph obtained from $H$ by contracting all vertices in $X$ into a single node (we use the word \emph{node} instead of vertex if it was obtained by contractions). Similarly, for a set of pairwise disjoint vertex sets $X_1, X_2, \dots, X_k$, we let $H / \{X_1, X_2, \dots, X_k\}$ denote the graph $((((H / X_1) / X_2) \dots ) / X_k)$. If $\mathcal{V}$ forms a partition of $V$, we use the convention to denote by $X^v$ the node in $G/ \mathcal{V}$ that contains $v \in V$, i.e. $v \in X^v$. To undo contractions, we define the function $\textsc{Flatten}(X)$ for a family of sets $X$ by $\textsc{Flatten}(X) = \bigcup_{x \in X} x$.

\paragraph{Reachability and Strong-Connectivity.} For graph $H$ and any two vertices $u,v \in V(H)$, we let $u \leadsto_H v$ denote that $u$ can reach $v$ in $H$, and $u \rightleftarrows_H v$ that $u$ can reach $v$ and vice versa (in the latter case, we also say $u$ and $v$ are strongly-connected). For any sets $X,Y \subseteq V(H)$, we say that $X \leadsto_H Y$ if there exists some $x \in X$, $y \in Y$ such that $x \leadsto_H y$; we define $X \rightleftarrows_H Y$ analogously. We say that the partition of $V(H)$ induced by the equivalence relation $\rightleftarrows_H$ is the set of \emph{strongly-connected components} (SCCs). We denote by $\textsc{Condensation}(H)$ the \textit{condensation} of $H$, that is the graph obtained by contracting each SCC in $H$ into a node.

\paragraph{Distances, Diameter and Balls.} We let $\mathbf{dist}_H(u,v)$ denote the distance from vertex $u$ to vertex $v$ in graph $H$ and denote by $\pi_{u,v, H}$ the corresponding shortest path (we assume uniqueness by implicitly referring to the lexicographically shortest path). We define the weak diameter of $X \subseteq V(H)$ in $H$ by $\mathbf{diam}(X, H) = \max_{x,y \in X} \mathbf{dist}_H(x,y)$. We define a ball $B_{out}(r, \delta)$ for some source $r \in V$ and positive real $\delta \in \mathbb{R}_{>0}$ to be the set of vertices at distance at most $\delta$ from $r$, i.e. $B_{out}(r, \delta) = \{ v \in V \; | \; \mathbf{dist}_G(r, v) \leq \delta \}$. We similarly define $B_{in}(r, \delta)$ to be $B_{out}(r, \delta)$ in the graph $G$ where edge directions are reversed.

\paragraph{$S$-Distances.} We define the notion of $S$-distances for any $S \subseteq V(H)$ where for any pair of vertices $u,v \in V(H)$, the $S$-distance $\mathbf{dist}_H(u,v, S)$ denotes the minimum number of vertices in $S \setminus \{v\}$ encountered on any path from $u$ to $v$. Alternatively, the $S$-distance corresponds to $\mathbf{dist}_{H'}(u,v)$ where $H'$ is a graph with edges $E_{out}(S)$ of weight $1$ and edges $E \setminus E_{out}(S)$ of weight $0$. It therefore follows that for any $u,v \in V(H)$, $\mathbf{dist}_H(u,v) = \mathbf{dist}_H(u,v,V)$. 

\paragraph{Partitions.} For two partitions $P$ and $P'$ of a set $U$, we say that partition $P$ is a \textit{melding} for a partition $P'$ if for every set $X \in P'$, there exists a set $Y \in P$ with $X \subseteq Y$. We also observe that \textit{melding} is transitive, thus if $P$ is a melding for $P'$ and $P'$ a melding for $P''$ then $P$ is a melding for $P''$.

\section{The ES-tree}
\label{sec:esTree}

The starting point for all data structures derived in the rest of this thesis is the ES-tree. We first state the formal result obtained by Even and Shiloach \cite{shiloach1981line} (in a slightly extended version)

\begin{theorem}[ES-tree, Extended Version of \cite{shiloach1981line}.]\label{thm:EStree}
Given a directed, unweighted, partially dynamic graph $G=(V,E)$, a fixed source vertex $r \in V$ and a depth $\delta > 0$, there exists a data structure $\mathcal{E}_r$ called the ES-tree that explicitly maintains for every $v \in B_{out}(r, \delta)$, the distance from $r$ to $v$ in any version of $G$. The algorithm runs in total update time $O(m\delta)$ and has query time $O(1)$. A corresponding shortest $r$-to-$v$ path $P$ can be returned in time $O(|P|)$.
\end{theorem}

We point out that $B_{out}(r, \delta)$ in the above theorem refers to the ball at $r$ in the current version of the graph $G$, i.e. the set $B_{out}(r, \delta)$ is decremental itself and for any vertex $v$ at distance at most $\delta$ at some stage $i$ of $G$, we can query at stage $i$, the exact distance from $r$ to $v$. We also point out that for vertices not in $B_{out}(r, \delta)$, the data structure can be implemented to return $\infty$ (or any other sentinel). We present the data structure for a decremental graph, however, it is straight-forward to obtain the same result in incremental graphs.

\paragraph{The Algorithm.} We maintain distances from a fixed source $r \in V$ to each vertex $v$ in $V$ up to distance $\delta > 0$ by storing a distance estimate $\widetilde{\mathbf{dist}}(r,v)$ that is initialized to the distance between $r$ and $v$ in $G$ along with a shortest-path tree $T$ rooted at $r$. On update $(u,v)$, we delete the edge from $G$ and possibly from $T$. Then, if possible, we extract $w_{min} \in V \setminus \{r\}$, the vertex with smallest distance estimate among vertices without incoming edge in $T$. (For the first extraction, we always have $w_{min} = v$.) For this vertex $w_{min}$, we then try to find a vertex $x \in \mathcal{N}^{in}(w_{min})$ such that adding $(x, w_{min})$ to $T$ implies $\mathbf{dist}_T(r,w_{min}) \leq \widetilde{\mathbf{dist}}(r,w_{min})$. We therefore search in $\mathcal{N}^{in}(w_{min})$ for an $x$ that satisfies \begin{equation}
\label{eq:enforce}
    \widetilde{\mathbf{dist}}(r,x)+ w(x,w_{min}) \leq \widetilde{\mathbf{dist}}(r,w_{min}).
\end{equation}
If no such $x$ exists, then $\widetilde{\mathbf{dist}}(r,w_{min})$ has to be incremented, and we set $T$ to $T \setminus N^{out}(w_{min})$. If $\widetilde{\mathbf{dist}}(r,w_{min}) > \delta$, we set it to $\infty$ and remove $w_{min}$ from the tree. We iterate the process until $T$ is spanning for vertices with distance estimate $< \infty$. 

\paragraph{Total Update Time.} We claim that for each distance estimate value $\widetilde{\mathbf{dist}}(r,v)$, $\mathcal{N}^{in}(v)$ has to be scanned only once. This follows since if Equation \ref{eq:enforce} is satisfied for a neighbor then the edge is taken into the tree $T$ until deleted, but if it does not satisfy the equation, it never has to be considered at the current distance estimate value of $v$ since the left-hand-side of the equation is monotonically increasing over time and thus the equation cannot be satisfied before the right-hand-side is increased. Since estimates increase monotonically, the total update time can be bound by $O(\sum_{v \in V} |\mathcal{N}^{in}(v)| \delta) = O(m \delta)$. 

\paragraph{Correctness of the Algorithm.} It remains to show that after the algorithm to reconstruct the tree $T$ terminates, we have for all $x \in B_{out}(r, \delta)$, that $\widetilde{\mathbf{dist}}(r, x) = {\mathbf{dist}}(r,x)$.  

While a formal argument would require nested induction, we simply observe that a vertex $x$ at distance $d \leq \delta$, always has the vertex $y$ that precedes it on the shortest path from $r$-to-$x$ in its in-neighborhood. An inductive argument over the distances from $r$ thus shows that we can always add an in-edge for $x$ to the tree $T$ that ensures that the distance estimate is at most the real distance. On the other hand, since $T$ is a subset of $G$ after each stage, each path in the tree is at most of weight equal to the shortest path distances.

\section{Reducing the Weight Dependency in Partially Dynamic SSSP}
\label{sec:reductionWeights}

Here, we also give a small reduction that extends $(1+\epsilon)$-approximate partially dynamic SSSP for weight ratio $n^4$, to any weight ratio $W$ at the expense of an additional $\log W$ factor. Throughout this thesis, we are therefore only concerned with obtaining a polylogarithmic dependency on $W$ since by the theorem below, we can then reduce the dependency on $W$ to $\log nW$. We will not state explicitly in the sections that we use the below theorem, however, claim the update times in 
\Cref{sec:contribution} by making use of it.

\begin{theorem}
For any $1/2 > \epsilon > 1/n$, given a data structure $\mathcal{E}_r$ that maintains $(1 + \epsilon)$-approximate partially dynamic SSSP on a graph $H$ in time $\mathcal{T}_{SSSP}(m_H, n_H, \epsilon, W_H)$ (where we assume that distance estimates are maintained explicitly). Then, there exists a data structure, that maintains  $(1 + 6\epsilon)$-approximate partially dynamic SSSP on a graph $G$ in time 
\[
\left( (m+n) \log n / \epsilon + \mathcal{T}_{SSSP}(m, n, \epsilon, n^4)\right) \cdot O(\log(nW) ).
\]
\end{theorem}
\begin{proof}
Given the graph $G$ with weight ratio $W$. Assume that the smallest weight is $1$ and the largest weight $W$ (this is without loss of generality). Throughout the update sequence, for $i = 0, 1, \dots, \lceil \lg Wn \rceil + 1$, maintain a graph $G_i$ derived from $G$ by removing all edges of weight larger $2^i$ and by rounding all remaining edges up to the nearest multiple of $2^{i}/ n^2$. Finally, run a data structure $\mathcal{E}^i_r$ on each such graph $G_i$ and maintain for each for each $v \in V$, a min-priority heap over the distance estimates $\widetilde{\mathbf{dist}}^i(r,v)$ in all data structures $\mathcal{E}^i_r$.

We first observe that the resulting graphs $G_i$ clearly have weight ratio $n^2$ by definition. For the running time, we first observe that we have $O(\log nW)$ data structures $\mathcal{E}^i_r$. Further, each of the $m$ updates to $G$ results in at most $O(\log nW)$ updates to graphs $G_i$. The data structures $\mathcal{E}^i_r$ update the distance estimates at most $O(\mathcal{T}_{SSSP}(m, n, \epsilon, n^2))$ times. Updating the heap data structure only when distance estimates change by $(1+\epsilon)$ which occurs at most $O(\log(nW)/\epsilon)$ times for each vertex, and using that each update to the heap is at cost $O(\log \log Wn) = O(\log n)$, we derive our final bound on the running time.

For correctness, observe that during stage $t$, for a vertex $v$ at distance $2^{i-1} < \mathbf{dist}(r,v) \leq 2^i$ for some $i$, in the graph $G_i$, the shortest path from $r$ to $v$ is at most $(1+2/n) \mathbf{dist}(r,v)$ since each edge on the shortest $r$-to-$v$ path in $G$ has weight at most $2^i$ and is thus in $G_i$ and further each edge incurs additive error at most $2^i/n^2$ by the rounding of edges combined with the fact that the distance is at least $2^{i-1}$ and that there are at most $n$ edges on any simple path. Further, the distance in each $G_i$ between any pair of vertices is at least as large as in $G$. The final approximation guarantee stems from the fact that $(1+\epsilon) \cdot (1+2/n) \leq (1+\epsilon) \cdot (1+2\epsilon) \leq (1+6\epsilon)$ since for $0 \leq x \leq 1$, $1+x \leq e^x \leq 1+2x$.
\end{proof}

%% file: rand_scc.tex
\label{chap:rand_scc}

In this section, we prove the following result.

\sccMain*

We point out that it suffices to maintain the strongly-connected components of $G$ since we can then run the decremental SCC data structure on the graph $G \cup (V \times \{r\}$, that is the graph $G$ with an additional edge from every vertex $v \in V$ to $r$. Observe that a vertex $v$ is reachable from $r$ in $G$ if and only if $r$ and $v$ are strongly-connected in $G \cup (V \times \{r\})$.

We next give an overview on how to obtain the Decremental SCC data structure presented in \Cref{thm:ContributionSCCmain}. We then present our one of our key insights, a generalization of the ES-tree, and in the final sections show how to use this insight to obtain a fast data structure. 

\section{Overview}
\label{subsec:overview}

We start by introducing the graph hierarchy maintained by our algorithm, followed by a high-level overview of our algorithm.

\paragraph{High-level overview of the Hierarchy.}  Our hierarchy has levels $0$ to $\lfloor \lg n \rfloor + 1$ and we associate with each level $i$ a subset $E_i$ of the edges $E$. The sets $E_i$ form a partition of $E$; we define the edges that go into each $E_i$ later in the overview but point out that we maintain $E_{\lfloor \lg n \rfloor + 1} = \emptyset$. We define a graph hierarchy $\hat{G} = \{\hat{G}_0, \hat{G}_1, .. ,\hat{G}_{\lfloor \lg n \rfloor + 1}\}$ such that each graph $\hat{G}_i$ is defined as 
\[
\hat{G}_i = \textsc{Condensation}((V, \bigcup_{j < i} E_j)) \;\cup\; E_{i}
\]
That is, each $\hat{G}_i$ is the condensation of a subgraph of $G$ with some additional edges. As mentioned in the preliminary section, we refer to the elements of the set $\hat{V}_i = V(\hat{G}_i)$ as \textit{nodes} to distinguish them from \textit{vertices} in $V$. We use capital letters to denote nodes and small letters to denote vertices. We let $X_i^v$ denote the node in $\hat{V}_i$ with $v \in X_i^v$. Observe that each node $X$ corresponds to a subset of vertices in $V$ and that for any $i$, $\hat{V}_i$ can in fact be seen as a partition of $V$. For $\hat{G}_0 = \textsc{Condensation}((V, \emptyset)) \;\cup\; E_{0}$, the set $\hat{V}_0$ is a partition of singletons, i.e. $\hat{V}_0 = \{ \{v\} | v \in V\}$, and $X_0^v = \{ v\}$ for each $v \in V$. 

Observe that because the sets $E_i$ form a partition of $E$ and $E_{\lfloor \lg n \rfloor + 1} = \emptyset$, the top graph $\hat{G}_{\lfloor \lg n \rfloor + 1}$ is simply defined as
\begin{align*}
    \hat{G}_{\lfloor \lg n \rfloor + 1} &= \textsc{Condensation}((V, \bigcup_{j < \lfloor \lg n \rfloor + 1} E_j)) \;\cup\; E_{\lfloor \lg n \rfloor + 1} \\ &= \textsc{Condensation}((V,E)).
\end{align*}
Therefore, if we can maintain $\hat{G}_{\lfloor \lg n \rfloor + 1} $ efficiently, we can answer queries on whether two vertices $u,v \in V$ are in the same SCC in $G$ by checking if $X_{\lfloor \lg n \rfloor + 1}^u$ is equal to $X_{\lfloor \lg n \rfloor + 1}^v$.

Let us offer some intuition for the hierarchy. The graph $\hat{G}_0$ contains all the vertices of $G$, and all the edges of $E_0 \subseteq E$. By definition of $\textsc{Condensation}(\cdot)$, the nodes of $\hat{G}_1$ precisely correspond to the SCCs of $\hat{G}_0$. $\hat{G}_1$ also includes the edges $E_0$ (though some of them are contracted into self-loops in $\textsc{Condensation}((V, E_0))$), as well as 
the additional edges in $E_1$. These additional edges might lead to $\hat{G}_1$ having larger SCCs than those of $\hat{G}_0$; each SCC in $\hat{G}_1$ then corresponds to a node in $\hat{G_2}$. More generally, the nodes of $\hat{G}_{i+1}$ are the SCCs of $\hat{G}_{i}$. 

As we move up the hierarchy, we add more and more edges to the graph, so the SCCs get larger and larger. Thus, each set $\hat{V}_i$ is a \textit{melding} for any $\hat{V}_j$ for $j \leq i$; that is for each node $Y \in \hat{V}_j$ there exists a set $X \in \hat{V}_i$ such that $Y \subseteq X$. We sometimes say we \textit{meld} nodes $Y, Y' \in \hat{V}_j$ to $X \in \hat{V}_{i}$ if $Y, Y' \subseteq X$ and $j < i$. Additionally, we observe that for any SCC $Y \subseteq \hat{V}_i$ in $\hat{G}_i$, we meld the nodes in SCC $Y$ to a node in $X \in \hat{V}_{i+1}$, and $X$ consists exactly of the vertices contained in the nodes of $Y$. More formally, $X = \textsc{Flatten}(Y)$.


To maintain the SCCs in each graph $\hat{G}_i$, our algorithm employs a bottom-up approach. At level $i+1$ we want to maintain SCCs in the graph with all the edges in $\bigcup_{j \leq i+1} E_j$, but instead of doing so from scratch, we use the SCCs maintained at level $\hat{G}_i$ as a starting point. The SCCs in $\hat{G}_i$ are precisely the SCCs in the graph with edge set $\bigcup_{j \leq i} E_j$; so to maintain the SCCs at level $i+1$, we only need to consider how the sliver of edges in $E_{i+1}$ cause the SCCs in $\hat{G}_{i}$ to be melded into larger SCCs (which then become the nodes of $\hat{G}_{i+2}$).

If the adversary deletes an edge in $E_i$, all the graphs $\hat{G}_{i-1}$ and below remain unchanged, as do the nodes of $\hat{G}_{i}$. But the deletion might split apart an SCC in $G_i$, which will in turn cause a node of $\hat{G}_{i+1}$ to split into multiple nodes. This split might then cause an SCC of $\hat{G}_{i+1}$ to split, which will further propagate up the hierarchy.

In addition to edge deletions caused by the adversary, our algorithm will sometimes move edges from $E_i$ to $E_{i+1}$. Because the algorithm only moves edges \emph{up} the hierarchy, each graph $\hat{G}_i$ is only losing edges, so the update sequence remains decremental from the perspective of each $\hat{G}_i$. We now give an overview of how our algorithm maintains the hierarchy efficiently.

\paragraph{Maintaining SCCs with ES-trees.} Consider again graph $\hat{G}_0$ and let $X \subseteq \hat{V}_0$ be some SCC in $\hat{G}_0$ that we want to maintain. Let some node $X'$ in $X$ be chosen to be the \textit{center} node of the SCC (In the case of $\hat{G}_0$, the node $X'$ is just a single-vertex set $\{ v \}$). We then maintain an ES in-tree and an ES out-tree from $X'$ that spans the nodes in $X$ in the induced graph $\hat{G}_0[X]$. We must maintain the trees up to distance $\mathbf{diam}(\hat{G}_0[X])$, so the total update time is $O(|E(\hat{G}_0[X])|* \mathbf{diam}(\hat{G}_0[X]))$ (recall \Cref{thm:EStree}). 

Now, consider an edge deletion to $\hat{G}_0$ such that the ES in-tree or ES out-tree at $X'$ is no longer a spanning tree. Then, we detected that the SCC $X$ has to be split into at least two SCCs $X_1, X_2, .., X_k$ that are node-disjoint with $X = \bigcup_i X_i$. Then in each new SCC $X_i$ we choose a new center and initialize a new ES in-tree and ES out-tree.

\paragraph{Exploiting small diameter.} The above scheme clearly is quite efficient if $\mathbf{diam}(\hat{G}_0)$ is very small. Our goal is therefore to choose the edge set $E_0$ in such a way that $\hat{G}_0$ contains only SCCs of small diameter. We therefore turn to some insights from \cite{chechik2016decremental} and extract information from the ES in-tree and out-tree to maintain small diameter. Their scheme fixes some $\delta > 0$ and if a set of nodes $Y \subseteq X$ for some SCC $X$ is at distance $\Omega(\delta)$ from/to $\textsc{Center}(X)$ due to an edge deletion in $\hat{G}_0$, they find a node separator $S$ of size $O(\min\{|Y|, |X \setminus Y|\}\log n / \delta)$; removing $S$ from  $\hat{G}_0$ causes $Y$ and $X \setminus Y$ to no longer be in the same SCC. We use this technique and remove edges incident to the node separator $S$ from $E_0$ and therefore from $\hat{G}_0$. One subtle observation we want to stress at this point is that each node in the separator set appears also as a single-vertex node in the graph $\hat{G}_1$; this is because each separator node $\{ s\}$ for some $s \in V$ is not \textit{melded} with any other node in $\hat{V}_0$, as it has no edges in $\hat{G}_0$ to or from any other node.

For some carefully chosen $\delta = \Theta(\log^2 n)$, we can maintain $\hat{G}_0$ such that at most half the nodes in $\hat{V}_0$ become separator nodes at any point of the algorithm.  This follows since each separator set is small in comparison to the smaller side of the cut and since each node in $\hat{V}_0$ can only be $O(\log n)$ times on the smaller side of a cut. 

\paragraph{Reusing ES-trees.} Let us now refine our approach to maintain the ES in-trees and ES out-trees and introduce a crucial ingredient devised by Roditty and Zwick \cite{roditty2008improved}. Instead of picking an arbitrary center node $X'$ from an SCC $X$ with $X' \in X$, we are going to pick a vertex $r \in \textsc{Flatten}(X) \subseteq V$ uniformly at random and run our ES in-tree and out-tree $\mathcal{E}_{r}$ from the node $X_0^r$ on the graph $\hat{G}_0$. For each SCC $X$ we denote the randomly chosen root $r$ by $\textsc{Center}(X)$. In order to improve the running time, we \textit{reuse} ES-trees when the SCC $X$ is split into SCCs $X_1, X_2, .. , X_k$, where we assume wlog that $r \in \textsc{Flatten}(X_1)$, by removing the nodes in $X_2, .., X_k$ from $\mathcal{E}_{r}$ and setting $\textsc{Center}(X_1) = r$. Thus, we only need to initialize a new ES-tree for the SCC $X_2, .. , X_k$. Using this technique, we can show that each node is expected to participate in $O(\log n)$ ES-trees over the entire course of the algorithm, since we expect that if a SCC $X$ breaks into SCCs $X_1, X_2, .. , X_k$ then we either have that every SCC $X_i$ is of at most half the size of $X$, or with probability at least $1/2$ that $X_1$ is the new SCC that contains at least half the vertices, i.e. that the random root is contained in the largest part of the graph. Since the ES-trees work on induced graphs with disjoint node sets, we can therefore conclude that the total update time for all ES-trees is $O(m \log n* \mathbf{diam}(\hat{G}_0))$. 

We point out that using the ES in-trees and out-trees to detect node separators as described above complicates the analysis of the technique by Roditty and Zwick \cite{roditty2008improved} but a clever proof presented in \cite{chechik2016decremental} shows that the technique can still be applied. In our paper, we present a proof that can even deal with some additional complications and that is slightly simpler.

\paragraph{A contrast to the algorithm of Chechik et al \cite{chechik2016decremental}.}

Other than our hierarchy, the overview we have given so far largely comes from the algorithm of Chechik et al \cite{chechik2016decremental}. However, their algorithm does not use a hierarchy of graphs. Instead, they show that for any graph $G$, one can find (and maintain) a node separator $S$ of size $\tilde{O}(n/\delta)$ such that all SCCs in $G$ have diameter at most $\delta$. They can then use ES-trees with random sources to maintain the SCCs in $G \setminus S$ in total update time $\tilde{O}(m\delta)$. This leaves them with the task of computing how the vertices in $S$ might meld some of the SCCs in $G \setminus S$. They are able to do this in total update time $\tilde{O}(m|S|) = \tilde{O}(mn/\delta)$ by using an entirely different technique of \cite{lkacki2013improved}. Setting $\delta = \tilde{O}(\sqrt{n})$, they achieve the optimal trade-off between the two techniques: total update time $\tilde{O}(m\sqrt{n})$ in expectation.

We achieve our $\tilde{O}(m)$ total update time by entirely avoiding the technique of \cite{lkacki2013improved} for separately handling a small set of separator nodes, and instead using the graph hierarchy described above, where at each level we set $\delta$ to be polylog rather than $\tilde{O}(\sqrt{n})$. 

We note that while our starting point is the same as \cite{chechik2016decremental}, using a hierarchy of separators forces us to take a different perspective on the function of a separator set. The reason is that it is simply not possible to ensure that at each level of the hierarchy, all SCCs have small diameter. To overcome this, we instead aim for separator sets that decompose the graph into SCCs that are small with respect to a different notion of distance. The rest of the overview briefly sketches this new perspective, while sweeping many additional technical challenges under the rug. 

\paragraph{Refining the Hierarchy.} So far, we only discussed how to maintain $\hat{G}_0$ efficiently by deleting many edges from $E_0$ and hence ensuring that SCCs in $\hat{G}_0$ have small diameter. To discuss our bottom-up approach, let us define our graphs $\hat{G}_i$ more precisely.

We maintain a separator hierarchy $\mathcal{S} = \{S_0, S_1, .. , S_{\lfloor \lg{n} \rfloor + 2}\}$ where $\hat{V}_0 = S_0 \supseteq S_1 \supseteq .. \supseteq S_{\lfloor \lg{n} \rfloor + 1} = S_{\lfloor \lg{n} \rfloor + 2} = \emptyset$, with $|S_i| \leq n/2^i$, for all $i \in [0, \lfloor \lg{n} \rfloor + 2]$ (see below that for technical reasons we need to define $S_{\lfloor \lg{n} \rfloor + 2}$ to define $\hat{G}_{\lfloor \lg{n} \rfloor + 1}$). Each set $S_i$ is a set of single-vertex nodes -- i.e. nodes of the form $\{ v \}$ -- that is monotonically increasing over time. 

We can now more precisely define each edge set 
\[
E_i = E(\textsc{Flatten}(S_{i})) \setminus E(\textsc{Flatten}(S_{i+1})).
\] 
To avoid clutter, we abuse notation slightly referring henceforth to $\textsc{Flatten}(X)$ simply as $X$ if $X$ is a set of singleton sets and the context is clear. We therefore obtain
\begin{align*}
    \hat{G}_i &=  \textsc{Condensation}((V, \bigcup_{j < i} E_j)) \;\cup\; E_{i} \\ &= \textsc{Condensation}(G \setminus E(S_{i})) \;\cup\; \left(E(S_i) \setminus E(S_{i+1})\right).
\end{align*}
In particular, note that $\hat{G}_i$ contains all the edges of $G$ except those in $E(S_{i+1})$; as we move up to level $\hat{G}_{i+1}$, we add the edges $E(S_{i+1}) \setminus E(S_{i+2})$. Note that if $s \in S_{i} \setminus S_{i+1}$, and our algorithm then adds $s$ to $S_{i+1}$, this will remove all edges incident to $s$ from $E_{i}$ and add them to $E_{i+1}$. Thus the fact that the sets $S_i$ used by the algorithm are monotonically increasing implies the desired property that edges only move up the hierarchy (remember that we add more vertices to $S_i$ due to new separators found on level $i-1$).

At a high-level, the idea of the hierarchy is as follows. Focusing on a level $i$, when the ``distances'' in some SCC of $\hat{G}_i$ get too large (for a notion of distance defined below), the algorithm will add a carefully chosen set of separator nodes $s_1, s_2, ..$ in $S_i$ to $S_{i+1}$. By definition of our hierarchy, this will remove the edges incident to the $s_i$ from $\hat{G}_i$, thus causing the SCCs of $\hat{G}_i$ to decompose into smaller SCCs with more manageable ``distances''. We note that our algorithm always maintains the invariant that nodes added to $S_{i+1}$ were previously in $S_i$, which from the definition of our hierarchy, ensures that at all times the separator nodes in $S_{i+1}$ are single-vertex nodes in $\hat{V}_{i+1}$; this is because the nodes of $\hat{V}_{i+1}$ are the SCCs of $\hat{G}_i$, and $\hat{G}_i$ contains no edges incident to $S_{i+1}$.


\paragraph{Exploiting $S$-distances.} For our algorithm, classic ES-trees are only useful to maintain SCCs in $\hat{G}_0$; in order to handle levels $i > 0$ we develop a new generalization of ES-trees that use a different notion of distance. This enables us to detect when SCCs are split in graphs $\hat{G}_i$ and to find separator nodes in $\hat{G}_i$ as discussed above more efficiently. 

Our generalized ES-tree (GES-tree) can be seen as a combination of the classic ES-trees\cite{shiloach1981line} and a data structure by Italiano\cite{italiano1988finding} that maintains reachability from a distinguished source in a directed acyclic graph (DAG), and which can be implemented in total update time $O(m)$.

Let $S$ be some feedback vertex set in a graph $G = (V,E)$; that is, every cycle in $G$ contains a vertex in $S$. Then our GES-tree can maintain
$S$-distances and a corresponding shortest-path tree up to $S$-distance $\delta > 0$ from a distinguished source $X_i^r$ for some $r \in V$ in the graph $G$. (See \Cref{chap:prelim} for the definition of $S$-distances.) This data structure can be implemented to take $O(m\delta)$ total update time.


\paragraph{Maintaining the SCCs in $\hat{G}_i$.} Let us focus on maintaining SCCs in 
\[
\hat{G}_i = \textsc{Condensation}(G \setminus E(S_{i})) \;\cup\; (E(S_{i}) \setminus E(S_{i+1})).\]
Since the condensation of any graph forms a DAG, every cycle in $\hat{G}_i$ contains at least one edge from the set $E(S_{i}) \setminus E(S_{i+1})$. Since $E(S_{i}) \setminus E(S_{i+1})$ is a set of edges that is incident to $S_i$, we have that $S_i$ forms a feedback node set of $\hat{G}_i$. Now consider the scheme described in the paragraphs above, but instead of running an ES in-tree and out-tree from each center $\textsc{Center}(X)$ for some SCC $X$, we run a GES in-tree and out-tree on $\hat{G}_i[X]$ that maintains the $S_i$-distances to depth $\delta$.  Using this GES, whenever a set $Y \subseteq X$ of nodes has $S_i$-distance $\Omega(\delta)$, we show that we can find a separator $S$ of size $O(\min\{|S_i \cap Y|, |S_i \cap (X \setminus Y)|\}\log n / \delta)$ that only consists of nodes that are in $\{\{s\} | s \in S_i\}$; we then add the elements of set $S$ to the set $S_{i+1}$, and we also remove the nodes $Y$ from the GES-tree, analogously to our discussion of regular ES-trees above. Note that adding $S$ to $S_{i+1}$ removes the edges $E(S)$ from $\hat{G}_i$; since we chose $S$ to be a separator, this causes $Y$ and $X \setminus Y$ to no longer be part of the same SCC in $\hat{G}_i$. Thus, to maintain the hierarchy, we must then split nodes in $\hat{G}_{i+1}$ into multiple nodes corresponding to the new SCCs in $\hat{G}_i$: $X \setminus (Y \;\cup\; S)$, $Y \setminus S$ and every single-vertex set in $S$ ($Y$ might not form a SCC but we then further decompose it after we handled the node split). This might cause some self-loops in $\hat{G}_{i+1}$ to become edges between the newly inserted nodes (resulting from the split) and needs to be handled carefully to embed the new nodes in the GES-trees maintained upon the SCC in $\hat{G}_{i+1}$ that $X$ is part of. Observe that this does not result in edge insertions but only remaps the endpoints of edges. Further observe that splitting nodes can only increase $S_{i+1}$-distance since when they were still contracted their distance from the center was equivalent. Since $S_{i+1}$-distance still might increase, the update to might trigger further changes in the graph $\hat{G}_{i+1}$. 

Thus, overall, we ensure that all SCCs in $\hat{G}_i$ have $S_i$-diameter at most $O(\delta)$, and can hence be efficiently maintained by GES-trees. In particular, we show that whenever an SCC exceeds diameter $\delta$, we can, by moving a carefully chosen set of nodes in $S_i$ to $S_{i+1}$, remove a corresponding set of edges in $\hat{G}_i$, which breaks the large-$S_i$-diameter SCC into SCCs of smaller $S_i$-diameter.

\paragraph{Bounding the total update time.} Finally, let us sketch how to obtain the total expected running time $O(m \log^4 n)$. We already discussed how by using random sources in GES-trees (analogously to the same strategy for ES-trees), we ensure that each node is expected to be in $O(\log n)$ GES-trees maintained to depth $\delta = O(\log^2 n)$. 
Each such GES-tree is maintained in total update time $O(m\delta) = O(m\log^2n)$, so we have $O(m \log^3 n)$ total expected update time for each level, and since we have $O(\log n)$ levels, we obtain total expected update time $O(m\log^4 n)$. We point out that we have not included the time to compute the separators in our running time analysis; indeed, computing separators efficiently is one of the major challenges to building our hierarchy. Since implementing these subprocedures efficiently is rather technical and cumbersome, we omit their description from the overview but refer to section \ref{subsec:separators} for a detailed discussion.

\section{Generalized ES-trees}
\label{subsec:EStree}

For our algorithm, we devise a new version of the ES-trees that maintains the shortest-path tree with regard to $S$-distances. We show that if $S$ is a \textit{feedback vertex set} for $G$, that is a set such that every cycle in $G$ contains at least one vertex in $S$, then the data structure requires only $O(m\delta)$ total update time. Our fundamental idea is to combine classic ES-trees with techniques to maintain Single-Source Reachability in DAGs which can be implemented in linear time in the number of edges \cite{italiano1988finding}. Since $\mathbf{dist}_G(r, v) = \mathbf{dist}_G(r, v, V)$  and $V$ is a trivial feedback vertex set, we have that our data structure generalizes the classic ES-tree. Since the empty set is a feedback vertex set for DAGs, our data structure also matches the time complexity of Italiano's data structure. We define the interface formally below.

\begin{definition}
\label{def:GES}
Let $G=(V,E)$ be a graph and $S$ a feedback vertex set for $G$, $r \in V$ and $\delta > 0$. We define a generalized ES-tree $\mathcal{E}_r$ (GES) to be a data structure that supports the following operations:
\begin{itemize}
    \item $\textsc{InitGES}(r, G, S, \delta)$: Sets the parameters for our data structure. We initialize the data structure and return the GES.
    \item $\textsc{Distance}(r ,v )$: $\forall v \in V$, if $\mathbf{dist}_G(r ,v, S) \leq \delta$, $\mathcal{E}_r$ reports $\mathbf{dist}_G(r ,v, S)$, otherwise $\infty$.
    \item $\textsc{Distance}(v ,r )$: $\forall v \in V$, if $\mathbf{dist}_G(v ,r, S) \leq \delta$, $\mathcal{E}_r$ reports $\mathbf{dist}_G(v ,r, S)$, otherwise $\infty$.
    \item $\textsc{Delete}(u,v)$: Sets $E \gets E \setminus \{(u,v)\}$.
    \item $\textsc{Delete}(V')$: For $V' \subseteq V$, sets $V \gets V \setminus V'$, i.e. removes the vertices in $V'$ and all incident edges from the graph $G$.
    \item $\textsc{GetUnreachableVertex}()$: Returns a vertex $v \in V$ with 
    \[
    \max\{ \mathbf{dist}_G(r ,v, S), \mathbf{dist}_G(v, r, S)\} > \delta
    \]
    or $\bot$ if no such vertex exists.
\end{itemize}
\end{definition}

\begin{lemma}
\label{lma:SimpleGES}
The GES $\mathcal{E}_r$ as described in definition \ref{def:GES} can be implemented with total initialization and update time $O(m \delta)$ and requires worst-case time $O(1)$ for each operation $\textsc{Distance}(\cdot)$ and $\textsc{GetUnreachableVertex}()$.
\end{lemma}

Here, we only sketch the proof idea and refer the reader to \cite{bernstein2019decremental} for a full proof.

\begin{proof} (sketch)
Consider a classic ES-tree with each edge weight $w(u,v)$ of an edge $(u,v)$ in $E_{out}(S)$ set to $1$ and all other edges of weight $0$. Then, the classic ES-tree analysis maintains with each vertex $v \in V$ the distance level $l(v)$ that expresses the current distance from $s$ to $v$. We also have a shortest-path tree $T$, where the path in the tree from $s$ to $v$ is of weight $l(v)$. Since $T$ is a shortest-path tree, we also have that for every edge $(u,v) \in E$, $l(v) \leq l(u) + w(u,v)$. Now, consider the deletion of an edge $(u,v)$ from $G$ that removes an edge that was in $T$. To certify that the level $l(v)$ does not have to be increased, we scan the in-going edges at $v$ and try to find an edge $(u', v) \in E$ such that $l(v) = l(u') + w(u', v)$. On finding this edge, $(u',v)$ is added to $T$. The problem is that if we allow $0$-weight cycles, the edge $(u',v)$ that we use to reconnect $v$ might come from a $u'$ that was a descendant of $v$ in $T$. This will break the algorithm, as it disconnects $v$ from $s$ in $T$. But we show that this bad case cannot occur because $S$ is assumed to be a feedback vertex set, so at least one of the vertices on the cycle must be in $S$ and therefore the out-going edge of this vertex must have weight $1$ contradicting that there exists any $0$-weight cycle. The rest of the analysis follows closely the classic ES-tree analysis.
\end{proof}

To ease the description of our SCC algorithm, we tweak our GES implementation to work on the multi-graphs $\hat{G}_i$. We still root the GES at a vertex $r \in V$, but maintain the tree in $\hat{G}_i$ at $X_i^r$. The additional operations and their running time are described in the following Lemma whose proof is straight-forward and therefore omitted here (but can be found in \cite{bernstein2019decremental}). Note that we now deal with nodes rather than vertices which makes the definition of $S$-distances ambiguous (consider for example a node containing two vertices in $S$). For this reason, we require $S \subseteq \{\{v\} | v \in V\} \cap \hat{V}$ in the Lemma below, i.e. that the nodes containing vertices in $S$ are single-vertex nodes. But as discussed in the paragraph ``Refining the hierarchy" in the overview, our hierarchy ensures that every separator node $X \in S_i$ is always just a single-vertex node in $\hat{G}_i$. Thus this constraint can be satisfied by our hierarchy.

\begin{lemma}
\label{lma:AugmentedGES}
Say we are given a partition $\hat{V}$ of a universe $V$ and the graph $\hat{G}=(\hat{V}, E)$, a feedback node set $S \subseteq \{\{v\} | v \in V\} \cap \hat{V}$, a distinguished vertex $r \in V$, and a positive integer $\delta$. Then, we can run a GES $\mathcal{E}_{r}$ as in definition \ref{def:GES} on $\hat{G}$ in time $O(m\delta + \sum_{X \in \hat{V}} E(X) \log X)$ supporting the additional operations:
\begin{itemize}
    \item $\textsc{SplitNode}(X)$: the input is a set of vertices $X$ contained in node $Y \in \hat{V}$, such that either $E \cap (X \times Y \setminus X)$ or $E \cap (Y \setminus X \times X)$ is an empty set, which implies $X \not\rightleftarrows_{\hat{G} \setminus S} Y \setminus X$. We remove the node $Y$ in $\hat{V}$ and add node $X$ and $Y \setminus X$ to $\hat{V}$.
    \item $\textsc{Augment}(S')$: This procedure adds the nodes in $S'$ to the feedback vertex set $S$. Formally,
    the input is a set of single-vertex sets $S' \subseteq \{\{v\} | v \in V\} \cap \hat{V}$. $\textsc{Augment}(S')$ then adds every $s \in S'$ to $S$.
\end{itemize}
\end{lemma}

We point out that we enforce the properties on the set $X$ in $\textsc{SplitNode}(X)$ in order to ensure that the set $S$ remains a feedback node set at all times. 

\section{Initializing the \texorpdfstring{the graph hierarchy $\hat{\mathcal{G}}$}{the graph hierarchy}} 
\label{subsec:Preprocessing}

We assume henceforth that the graph $G$ initially is strongly-connected. If the graph is not strongly-connected, we can run Tarjan's algorithm \cite{tarjan1972depth} in $O(m+n)$ time to find the SCCs of $G$ and run our algorithm on each SCC separately. 

\begin{algorithm}
\caption{$\textsc{Preprocessing}(G, \delta)$}
\label{alg:preprocessing}
\KwIn{A strongly-connected graph $G=(V,E)$ and a parameter $\delta > 0$.}
\KwOut{A hierarchy of sets $\mathcal{S} = \{ S_0, S_1, .., S_{\lfloor \lg{n} \rfloor + 2}\}$ and graphs $\hat{\mathcal{G}} = \{\hat{G}_0, \hat{G}_1, .. , \hat{G}_{\lfloor \lg{n} \rfloor + 1}\}$ as described in section \ref{subsec:overview}. Further, each SCC $X$ in $\hat{G}_i$ for $i \leq \lfloor \lg n \rfloor + 2$, has a center $\textsc{Center}(X)$ such that for any $y \in X$, $\mathbf{dist}_{\hat{G}_{i}}(\textsc{Center}(X),y, S_{i}) \leq \delta/2$ and $\mathbf{dist}_{\hat{G}_{i}}(y,\textsc{Center}(X), S_{i}) \leq \delta/2$. }
\BlankLine

$ S_0 \gets V$\;
$ \hat{V}_0 \gets \{\{v\} | v \in V\}$\;
$ \hat{G}_0 \gets (\hat{V}_0, E)$\label{lne:G0Init}\;
\For{ $i = 0 $ \KwTo $ \lfloor \lg{n} \rfloor$}{
    \tcc{Find separator $S_{Sep}$ such that no two vertices in the same SCC in $\hat{G}_{i}$ have $S_i$-distance $\geq \delta/2$. $P$ is the collection of these SCCs.}
	$ (S_{Sep}, P) \gets \textsc{Split}(\hat{G}_i, S_{i}, \delta/2)$ \;
	$ S_{i+1} \gets S_{Sep} $\;
	$\textsc{InitNewPartition}(P, i, \delta)$\;
    
    \tcc{Initialize the graph $\hat{G}_{i+1}$}
	$ \hat{V}_{i+1} \gets P$\;
    $ \hat{G}_{i+1} \gets (\hat{V}_{i+1}, E)$\label{lne:GIInit}\;
}        
$S_{\lfloor \lg{n} \rfloor + 2} \gets \emptyset$\;
\end{algorithm}

Our procedure to initialize our data structure is presented in pseudo-code in  \Cref{alg:preprocessing}. We first initialize the level $0$ where $\hat{G}_0$ is simply $G$ with the vertex set $V$ mapped to the set of singletons of elements in $V$. 

Let us now focus on an iteration $i$. Observe that the graph $\hat{G}_i$ initially has all edges in $E$ (by initializing each $\hat{G}_i$ in  \Cref{lne:G0Init} or \Cref{lne:GIInit}). Our goal is then to ensure that all SCCs in $\hat{G}_i$ are of small $S_i$-diameter at the cost of removing some of the edges from $\hat{G}_i$. Invoking the procedure $\textsc{Split}(\hat{G}_i, S_i, \delta/2)$ provides us with a set of separator nodes $S_{Sep}$ whose removal from $\hat{G}_i$ ensure that the $S_i$ diameter of all remaining SCCs is at most $\delta$. The set $P$ is the collection of all these SCCs, i.e. the collection of the SCCs in $\hat{G}_i \setminus E(S_{Sep})$. 

\Cref{lma:split} below describes in detail the properties satisfied by $\textsc{Split}(\hat{G}_i, S_i, \delta/2)$. In particular, besides the properties ensuring small $S_i$-diameter in the graph $\hat{G}_i \setminus E(S_{Sep})$ (properties \ref{prop:1} and \ref{prop:2}), the procedure also gives an upper bound on the number of separator vertices (property \ref{prop:3}). Setting $\delta = 64 \lg^2 n$, clearly implies that $|S_{Sep}| \leq |S_i|/2$ and ensures running time $O(m \log^3 n)$.

\begin{lemma}
\label{lma:split}
$\textsc{Split}(\hat{G}_i, S_i, \delta/2)$ returns a tuple $(S_{Sep}, P)$ where $P$ is a partition of the node set $\hat{V}_i$ such that 
\begin{enumerate}
    \item for $X \in P$, and nodes $u,v \in X$ we have  $\mathbf{dist}_{G \setminus E(S_{Sep})}(u,v,S) \leq \delta/2$, and \label{prop:1}
    \item for distinct $X, Y \in P$, with nodes $u \in X$ and $v \in Y$,   $u \not\rightleftarrows_{G \setminus E(S_{Sep})} v$, and \label{prop:2}
    \item 
    $|S_{Sep}| \leq  \frac{32 \lg^2 n}{\delta} |S_i|$. \label{prop:3}
\end{enumerate}
The algorithm runs in time $O\left(\delta m \lg n\right)$.
\end{lemma}

We then set $S_{i+1} = S_{Sep}$ which implicitly removes the edges $E(S_{i+1})$ from the graph $\hat{G}_i$. We then invoke the procedure $\textsc{InitNewPartition}(P, i, \delta)$, that is presented in algorithm \ref{alg:newPart}. The procedure initializes for each $X \in P$ that corresponds to an SCC in $\hat{G}_i$ the GES-tree from a vertex $r \in \textsc{Flatten}(X)$ chosen uniformly at random on the induced graph $\hat{G}_i[X]$. Observe that we are not explicitly keeping track of the edge set $E_i$ but further remove edges implicitly by only maintaining the induced subgraphs of $\hat{G}_i$ that form SCCs. A small detail we want to point out is that each separator node $X \in S_{Sep}$ also forms its own single-node set in the partition $P$.

\begin{algorithm}
\caption{$\textsc{InitNewPartition}(P, i, \delta)$\;}
\label{alg:newPart}
\KwIn{A partition of a subset of the nodes $V$, and the level $i$ in the hierarchy.}
\KwResult{Initializes a new ES-tree for each set in the partition on the induced subgraph $\hat{G}_i$.}
\BlankLine
\ForEach{$ X \in P $}{
    Let $r$ be a vertex picked from $\textsc{Flatten}(X)$ uniformly at random. \;
    $\textsc{Center}(X) \gets r$\; 
    \tcc{Init a generalized ES-tree from $\textsc{Center}(X)$ to depth $\delta$.}
    $\mathcal{E}_r^i \gets \textsc{InitGES}(\textsc{Center}(X), \hat{G}_i[X], S_i, \delta)$\;
}
\end{algorithm}

On returning to algorithm \ref{alg:preprocessing}, we are left with initializing the graph $\hat{G}_{i+1}$. Therefore, we simply set $\hat{V}_{i+1}$ to $P$ and use again all edges $E$. Finally, we initialize $S_{\lfloor \lg n \rfloor + 2}$ to the empty set which remains unchanged throughout the entire course of the algorithm.

Let us briefly sketch the analysis of the algorithm which is more carefully analyzed in subsequent sections. Using again $\delta = 64 \lg^2 n$ and Lemma \ref{lma:split}, we ensure that $|S_{i+1}| \leq |S_i|/2$, thus $|S_i| \leq n/2^i$ for all levels $i$. The running time of executing the $\textsc{Split}(\cdot)$ procedure $\lfloor \lg n \rfloor + 1$ times incurs running time $O(m \log^4 n)$ and initializing the GES-trees takes at most $O(m \delta)$ time on each level therefore incurring running time $O(m \log^3 n)$. 

\section{Finding Separators} 
\label{subsec:separators}

Before we describe how to update the data structure after an edge deletion, we want to explain how to find good separators since it is crucial for our update procedure. We then show how to obtain an efficient implementation of the procedure $\textsc{Split}(\cdot)$ that is the core procedure in the initialization.

Indeed, the separator properties that we want to show are essentially reflected in the properties of Lemma \ref{lma:split}. For simplicity, we describe the separator procedures on simple graphs instead of our graphs $\hat{G}_i$; it is easy to translate these procedures to our multi-graphs
$\hat{G}_i$ because the separator procedures are not dynamic; they are only ever invoked on a fixed graph, and so we do not have to worry about node splitting and the like.

To gain some intuition for the technical statement of our separator properties stated in Lemma \ref{lma:sep}, consider that we are given a graph $G=(V,E)$, a subset $S$ of the vertices $V$, a vertex $r \in V$ and a depth $d$. Our goal is to find a separator $S_{Sep} \subseteq S$, such that every vertex in the graph $G \setminus S_{Sep}$ is either at $S$-distance at most $d$ from $r$ \emph{or} cannot be reached from $r$, i.e. is separated from $r$. 

We let henceforth $V_{Sep} \subseteq V$ denote the set of vertices that are still reachable from $r$ in $G \setminus S_{Sep}$ (in particular there is no vertex $S_{Sep}$ contained in $V_{Sep}$ and $r \in V_{Sep}$). Then, a natural side condition for separators is to require the set $S_{Sep}$ to be small in comparison to the smaller side of the cut, i.e. small in comparison to $\min\{|V_{Sep}|, |V \setminus (V_{Sep} \;\cup\; S_{Sep})|\}$.

Since we are concerned with $S$-distances, we aim for a more general guarantee: we want the set $S_{Sep}$ to be small in comparison to the number of $S$ vertices on any side of the cut, i.e. small in comparison to $\min\{|V_{Sep} \cap S|, |(V \setminus (V_{Sep} \;\cup\; S_{Sep})) \cap S|\}$. This is expressed in property \ref{prop:balanceS} of the Lemma.

\begin{lemma}[Balanced Separator]
\label{lma:sep}
There exists a procedure $\textsc{OutSep}(r, G, S, d)$ (analogously $\textsc{InSep}(r, G, S, d)$) where $G=(V,E)$ is a graph, $r \in V$ a root vertex, $S \subseteq V$ and $d$ a positive integer. The procedure computes a tuple $(S_{Sep}, V_{Sep})$ such that
\begin{enumerate}
    \item $S_{Sep} \subseteq S$, $V_{Sep} \subseteq V$, $S_{Sep} \cap V_{Sep} = \emptyset$, $r \in V_{Sep}$,
    \item \label{step:separator-distance} $\forall v \in V_{Sep} \;\cup\; S_{sep}$, we have $\mathbf{dist}_G(r,v,S) \leq d$ (analogously $\mathbf{dist}_G(v,r,S) \leq d$ for $\textsc{InSep}(r, G, S, d)$), 
    \item \[
    |S_{Sep}| \leq \frac{\min\{|V_{Sep}\cap S|, | (V \setminus (S_{Sep} \;\cup\; V_{Sep})) \cap S|\} 2\log{n}}{d},\] \label{prop:balanceS}
    and
    \item for any $x \in V_{Sep}$ and $y \in V \setminus (S_{Sep} \;\cup\; V_{Sep})$, we have $u \not\leadsto_{G \setminus E(S_{Sep})} v$ (analogously $v \not\leadsto_{G \setminus E(S_{Sep})} u$ for $\textsc{InSep}(r, G, S, d)$).
\end{enumerate}
The running time of both $\textsc{OutSep}(\cdot)$ and $\textsc{InSep}(\cdot)$ can be bounded by $O(E(V_{Sep}))$.
\end{lemma}

Again, we only sketch the proof idea and refer the reader to \cite{bernstein2019decremental} for a full proof.

\begin{proof}
To implement procedure $\textsc{OutSep}(r, G, S, d)$, we start by computing a BFS at $r$. Here, we assign edges in $E_{out}(S)$ again weight $1$ and all other edges weight $0$ and say a layer consists of all vertices that are at same distance from $r$. To find the first layer, we can use the graph $G \setminus E_{out}(S)$ and run a normal BFS from $r$ and all vertices reached form the first layer $L_0$. We can then add for each edge $(u,v) \in E_{out}(S)$ with $u \in L_0$ the vertex $v$ to $L_1$ if it is not already in $L_0$. We can then contract all vertices visited so far into a single vertex $r'$ and repeat the procedure described for the initial root $r$. It is straight-forward to see that the vertices of a layer that are also in $S$ form a separator of the graph. To obtain a separator that is small in comparison to $|V_{Sep} \cap S|$, we add each of the layers $0$ to $d/2$ one after another to our set $V_{Sep}$, and output the index $i$ of the first layer that grows the set of $S$-vertices in $V_{Sep}$ by factor less than $(1+\frac{2 \log n}{d})$. We then set $S_{Sep}$ to be the vertices in $S$ that are in layer $i$. If the separator is not small in comparison to $| (V \setminus (S_{Sep} \;\cup\; V_{Sep})) \cap S|$, we grow more layers and output the first index of a layer such that the separator is small in comparison to $| (V \setminus (S_{Sep} \;\cup\; V_{Sep})) \cap S|$. This layer must exist and is also small in comparison to $|V_{Sep} \cap S|$. Because we find our separator vertices $S_{Sep}$ using a BFS from $r$, a useful property of our separator is that all the vertices in $S_{Sep}$ and $V_{Sep}$ are within bounded distance from $r$.

Finally, we can ensure that the running time of the procedure is linear in the size of the set $E(V_{Sep})$, since these are the edges that were explored by the BFS from root $r$.
\end{proof}

Let us now discuss the procedure $\textsc{Split}(G,S,d)$ that we already encountered in section \ref{subsec:Preprocessing} and whose pseudo-code is given in algorithm \ref{alg:split}. Recall that the procedure computes a tuple $(S_{Split}, P)$ such that the graph $G \setminus E(S_{Split})$ contains no SCC with $S$-diameter larger $d$ and where $P$ is the collection of all SCCs in the graph $G \setminus E(S_{Split})$.

\begin{algorithm}
\caption{$\textsc{Split}(G, S, d)$}
\label{alg:split}
\KwIn{A graph $G=(V,E)$, a set $S \subseteq V$ and a positive integer $d$.}
\KwOut{Returns a tuple $(S_{Split}, P)$, where $S_{Split} \subseteq S$ is a separator such that no two vertices in the same SCC in $G \setminus E(S)$ have $S$-distance greater than $d$. $P$ is the collection of these SCCs.}
\BlankLine

$S_{Split} \gets \emptyset; P \gets \emptyset; G' \gets G;$\;

\While{$G' \neq \emptyset$}{
    Pick an arbitrary vertex $r$ in $V$.\;
    Run in parallel $\textsc{OutSep}(r, G', S, d/16)$ and $\textsc{InSep}(r, G', S, d/16)$ and let $(S_{Sep}, V_{Sep})$ be the tuple returned by the first subprocedure that finishes.\label{lne:sepTwoWay}\;

    \lIf(\label{lne:sepTwoWayIf}){$|V_{Sep}| \leq \frac{2}{3}|V|$}{
        $(S'_{Sep}, V'_{Sep}) \gets (S_{Sep}, V_{Sep})$
    }\Else(\label{lne:sepTwoWayElse}){
        Run the separator procedure that was aborted in line \ref{lne:sepTwoWay} until it finishes and let the tuple returned by this procedure be $(S'_{Sep}, V'_{Sep})$.
    }
    
    \If(\label{line:split-if-case}){$|V'_{Sep}| \leq \frac{2}{3}|V|$} {
        $(S_{Small}, P_{Small}) \gets \textsc{Split}(G'[V'_{Sep}], V'_{Sep} \cap S, d)$\label{lne:splitRecurseIf} \;
        $S_{Split} \gets S_{Split} \;\cup\; S_{Small} \;\cup\; S'_{Sep}$\label{lne:addtoS1}\;
        $P \gets P \;\cup\; P_{Small} \;\cup\; \{ \{s\} | s \in S'_{Sep}\})$\label{lne:addToPSep}\;
        $G' \gets G'[V \setminus (V'_{Sep} \;\cup\; S'_{Sep})]$
    }\Else(\label{line:split-else-case}){ 
        \tcc{Init a generalized ES-tree from $r$ to depth $d/2$.}
        $\mathcal{E}_r^i \gets \textsc{InitGES}(r, G', S, d/2)$\;
        \tcc{Find a good separator for every vertex that is far from $r$.}
         \While(\label{line:split-else-while}) {  $(v \gets \mathcal{E}_r^i.\textsc{GetUnreachableVertex}()) \neq \bot$}  {
            \If{$\mathcal{E}_r^i.\textsc{Distance}(r,v) > d/2$}{
                $(S''_{Sep}, V''_{Sep}) \gets \textsc{InSep}(v,G', S, d/4)$\;
            }\Else(\tcp*[h]{If $\mathcal{E}_r^i.\textsc{Distance}(v,r) > d/2$}){
                $(S''_{Sep}, V''_{Sep}) \gets \textsc{OutSep}(v,G', S, d/4)$\;
            }
            $\mathcal{E}_r.\textsc{Delete}(S''_{Sep} \;\cup\; V''_{Sep})$\;
        
            $(S'''_{Sep}, P''') \gets \textsc{Split}( G[V''_{Sep}] , V''_{Sep} \cap  S , d)$\label{lne:splitRecurse}\;
            $S_{Split} \gets S_{Split} \;\cup\; S''_{Sep} \;\cup\; S'''_{Sep}$\label{lne:addtoS2}\;
            $P \gets P \;\cup\; P''' \;\cup\; \{ \{s\} | s \in S''_{Sep}\})$\label{lne:addToP1}\;
        }
        $P \gets P \;\cup\; \{\mathcal{E}_r.\textsc{GetAllVertices}()\}$\label{lne:addToP2}\;
        $G' \gets \emptyset$\;
    }
}
\Return $(S_{Split}, P)$\;
\end{algorithm}

Let us sketch the implementation of the procedure $\textsc{Split}(G,S,d)$. We first pick a vertex and invoke the procedures $\textsc{OutSep}(r, G', S, d/4)$ and $\textsc{InSep}(r, G', S, d/4)$ to run in parallel, that is the operations of the two procedures are interleaved during the execution. If one of these subprocedures returns and presents a separator tuple $(S_{Sep}, V_{Sep})$, the other procedure is aborted and the tuple $(S_{Sep}, V_{Sep})$ is returned. If $|V_{Sep}| \leq \frac{2}{3}|V|$, then we conclude that the separator function only visited a small part of the graph. Therefore, we use the separator subsequently, but denote the tuple henceforth as $(S'_{Sep},V'_{Sep})$. Otherwise, we decide the separator is not useful for our purposes. We therefore return to the subprocedure we previously aborted and continue its execution. We then continue with the returned tuple $(S'_{Sep}, V'_{Sep})$.

From there on, there are two possible scenarios. The first scenario is that the subprocedure producing $(S'_{Sep}, V'_{Sep})$ has visited a rather small fraction of the vertices in $V$ (line \ref{line:split-if-case}); in this case, we have pruned away a small number of vertices $V'_{Sep}$ while only spending time proportional to the smaller side of the cut, so we can simply recurse on $V'_{Sep}$. We also have to continue pruning away vertices from the original set $V$, until we have either removed all vertices from $G$ by finding these separators and recursing, or until we enter the else-case (line \ref{line:split-else-case}). 

The else-case in line \ref{line:split-else-case} is the second possible scenario:  note that in this case we must have entered the else-case in line \ref{lne:sepTwoWayElse} and had both the $\textsc{InSep}(\cdot)$ \emph{and} $\textsc{OutSep}(\cdot)$ explore the large side of the cut. Thus we cannot afford to simply recurse on the smaller side of the cut $V \setminus V'_{sep}$, as we have already spent time $|V'_{sep}| > |V \setminus V'_{sep}|$. Thus, for this case we use a different approach. We observe that because we entered the else-case in line \ref{lne:sepTwoWayElse} and since we entered the else-case \ref{line:split-else-case}, we must have had that $|V_{Sep}| \geq \frac{2}{3}|V|$ \textit{and} that $|V'_{Sep}| \geq \frac{2}{3}|V|$. We will show that in this case, the root $r$ must have small $S$-distance to and at least $\frac{1}{3}|V|$ vertices. We then show that this allows us to efficiently prune away at most $\frac{2}{3}|V|$ vertices from $V$ at large $S$-distance to or from $r$. We recursively invoke $\textsc{Split}(\cdot)$ on the induced subgraphs of vertex sets that we pruned away.

We analyze the procedure in detail in multiple steps, and summarize the result in Lemma \ref{lma:splitFull} that is the main result of this section. Let us first prove that if the algorithm enters the else-case in line \ref{line:split-else-case} then we add an SCC of size at least $\frac{1}{3}|V|$ to $P$. 

\begin{claim}
\label{clm:largeSCCifEStree}
If the algorithm enters line \ref{line:split-else-case} then the vertex set returned by the procedure $\mathcal{E}_r.\textsc{GetAllVertices}()$ in line \ref{lne:addToP2} is of size at least $\frac{1}{3}|V|$.
\end{claim}
\begin{proof}
Observe first that since we did no enter the if-case in line \ref{lne:splitRecurseIf}, that $|V_{Sep}| > \frac{2}{3}|V|$ and $|V'_{Sep}| > \frac{2}{3}|V|$ (since we also cannot have entered the if case in line \ref{lne:sepTwoWayIf}). 

Since we could not find a sufficiently good separator in either direction, we certified that the $S$-out-ball from $r$ defined 
\[
B_{out}(r) = \{ v \in V | \mathbf{dist}_G(r , v, S) \leq d/16\}
\]
has size greater than $\frac{2}{3}|V|$, and that similarly, the $S$-in-ball $B_{in}(r)$ of $r$ has size greater than $\frac{2}{3}|V|$. This implies that 
\[
|B_{out}(r) \cap B_{in}(r)| > \frac{1}{3}|V|.
\]
Further, we have that every vertex on a shortest-path between $r$ and a vertex $v \in B_{out}(r) \cap B_{in}(r)$ has a shortest-path from and to $r$ of length at most $d/16$. Thus the $S$-distance between any pair of vertices in $B_{out}(r) \cap B_{in}(r)$ is at most $d/8$. Now, let $SP$ be the set of all vertices that are on a shortest-path w.r.t. $S$-distance between two vertices in $B_{out}(r) \cap B_{in}(r)$. Clearly, $B_{out}(r) \cap B_{in}(r) \subseteq SP$, so $|SP| \geq |V|/3$. It is also easy to see that $G[SP]$ has $S$-diameter at most $d/4$.

At this point, the algorithm repeatedly finds a vertex $v$ that is far from $r$ and finds a separator from $v$. We will now show that the part of the cut containing $v$ is always disjoint from $SP$; since $|SP| > |V|/3$, this implies that at least $|V|/3$ vertices remain in $\mathcal{E}_r$.

Finally, consider some vertex $v$ chosen in line \ref{line:split-else-while}. Let us say that we now run $\textsc{InSep}(v,G',S,d/4)$; the case where we run $\textsc{OutSep}(v,G',S,d/4)$ is analogous. Now, by property \ref{step:separator-distance} in Lemma \ref{lma:sep}, every $s \in S_{Sep}$ has $\mathbf{dist}(s,v,S) \leq d/4$. Thus, since we only run the \textsc{InSep} if we have $\mathbf{dist}(r,v,S) > d/2$, we must have $\mathbf{dist}(r,s,S) > d/4$. 
\end{proof}

We point out that claim \ref{clm:largeSCCifEStree} implies that $\textsc{Split}(\cdot)$ only recurses on disjoint subgraphs containing at most a $2/3$ fraction of the vertices of the given graph. To see this, observe that we either recurse in line \ref{lne:splitRecurseIf} on $G'[V'_{Sep}]$ after we explicitly checked whether $|V'_{Sep}| \leq \frac{2}{3}|V|$ in the if-condition, or we recurse in line \ref{lne:splitRecurse} on the subgraph pruned from the set of vertices that $\mathcal{E}_r$ was initialized on. But since by claim \ref{clm:largeSCCifEStree} the remaining vertex set in $\mathcal{E}_r$ is of size at least $|V|/3$, the subgraphs pruned away can contain at most $\frac{2}{3}|V|$ vertices.

We can use this observation to establish correctness of the $\textsc{Split}(\cdot)$ procedure.

\begin{claim}
\label{clm:SplitCorrectness}
$\textsc{Split}(G, S, d)$ returns a tuple $(S_{Sep}, P)$ where $P$ is a partition of the vertex set $V$ such that 
\begin{enumerate}
    \item for $X \in P$, and vertices $u,v \in X$ we have  $\mathbf{dist}_{G \setminus E(S_{Sep})}(u,v,S) \leq d$, and
    \item for distinct $X, Y \in P$, with vertices $u \in X$ and $v \in Y$,   $u \not\rightleftarrows_{G \setminus E(S_{Sep})} v$, and
    \item 
    $|S_{Split}| \leq  \frac{32 \log n}{d} \sum_{X \in P}  \lg (n  / |X \cap S|) |X \cap S|$. \label{prop:splitcorrect3}
\end{enumerate}
\end{claim}
\begin{proof}
Let us start with the first two properties which we prove by induction on the size of $|V|$ where the base case $|V|=1$ is easily checked. For the inductive step, observe that each SCC $X$ in the final collection $P$ was added to $P$ in line \ref{lne:addToPSep}, \ref{lne:addToP1} or \ref{lne:addToP2}. We distinguish by 3 cases: 
\begin{enumerate}
    \item a vertex $s$ was added as singleton set after appearing in a separator $S_{Sep}$ but then $\{s\}$ is strongly-connected and $s$ cannot reach any other vertex in $G \setminus E(S_{Sep})$ since it has no out-going edges, or
    \item an SCC $X$ was added as part of a collection $P'''$ in line  \ref{lne:addToP1}. But then we have that the collection $P'''$ satisfies the properties in $G[V''_{Sep}]$ by the induction hypothesis and since $V''_{Sep}$ was a cut side and $S''_{Sep}$ added to $S_{out}$, we have that there cannot be a path to \emph{and} from any vertex in $G \setminus E(S_{out})$, or
    \item we added the non-trivial SCC $X$ to $P$ after constructing an GES-tree from some vertex $r \in X$ and after pruning each vertex at $S$-distance to/from $r$ larger than $d/2$ (see the while loop on line \ref{line:split-else-while}). But then each vertex that remains in $X$ can reach $r$ within $S$-distance $d/2$ and is reached from $r$ within distance $d/2$ implying that any two vertices $u,v \in X$ have a path from $u$ to $v$ of $S$-distance at most $d$.
\end{enumerate}

Finally, let us upper bound the number of vertices in $S_{Split}$. We use a classic charging argument and argue that each time we add a separator $S_{Sep}$ to $S_{Split}$ with sides $V_{Sep}$ and $V \setminus (V_{Sep} \;\cup\; S_{Sep})$ at least one of these sides contains at most half the $S$-vertices in $V \cap S$. Let $X$ be the smaller side of the cut (in term of $S$-vertices) then by property \ref{prop:balanceS} from Lemma \ref{lma:sep}, we can charge each $S$ vertex in $X$ for $\frac{32 \log{n}}{d}$ separator vertices (since we invoke $\textsc{OutSep}(\cdot)$ and $\textsc{InSep}(\cdot)$ with parameter at least $d/16$). 

Observe that once we determined that a separator $S_{Sep}$ that is about to be added to $S_{Split}$ in line \ref{lne:addtoS1} or \ref{lne:addtoS2}, we only recurse on the induced subgraph $G'[V_{Sep}]$ and let the graph in the next iteration be $G'[V \setminus (V_{Sep} \;\cup\; S_{Sep})$. 

Let $X$ be an SCC in the final collection $P$. Then each vertex $v \in X$ can only have been charged at most $\lg (n  - |X \cap S|)$ times. The Lemma follows.
\end{proof}

It remains to bound the running time. Before we bound the overall running time, let us prove the following claim on the running time of invoking the separator procedures in parallel.

\begin{claim}
\label{clm:twowaysep}
We spend $O(E(V'_{Sep} \;\cup\; S'_{Sep}))$ time to find a separator in line \ref{lne:sepTwoWayIf} or \ref{lne:sepTwoWayElse}.
\end{claim}
\begin{proof}
Observe that we run $\textsc{OutSep}(r, G', S, d/16)$ and $\textsc{InSep}(r, G', S, d/16)$ in line \ref{lne:sepTwoWay} in parallel. Therefore, when we run them, we interleave their machine operations, computing one operation from $\textsc{OutSep}(r, G', S, d/16)$ and then one operation from $\textsc{InSep}(r, G', S, d/16)$ in turns. Let us assume that $\textsc{OutSep}(r, G', S, d/16)$ is the first subprocedure that terminates and returns tuple $(S_{Sep}, V_{Sep})$. Then, by Lemma \ref{lma:sep}, the subprocedure used $O(E(V_{Sep} \;\cup\; S_{Sep}))$ time. Since the subprocedure $\textsc{InSep}(r, G', S, d/16)$ ran at most one more operation than $\textsc{OutSep}(r, G', S, d/16)$, it also used $O(E(V_{Sep} \;\cup\; S_{Sep}))$ operations. If $\textsc{InSep}(r, G', S, d/16)$ finishes first, a symmetric argument establishes the same bounds. The overhead induced by running two procedures in parallel can be made constant. 

Since assignments take constant time, the claim is vacuously true by our discussion if the if-case in line \ref{lne:splitRecurseIf} is true. Otherwise, we compute a new separator tuple by continuing the execution of the formerly aborted separator subprocedure. But by the same argument as above, this subprocedure's running time now clearly dominates the running time of the subprocedure that finished first in line \ref{lne:sepTwoWay}. The time to compute $(S'_{Sep}, V'_{Sep})$ is thus again upper bounded by $O(E(V'_{Sep}))$ by Lemma \ref{lma:sep}, as required.
\end{proof}

Finally, we have established enough claims to prove Lemma \ref{lma:splitFull}. 

\begin{lemma}[Strengthening of Lemma \ref{lma:split}]
\label{lma:splitFull}
The procedure $\textsc{Split}(G, S, d)$ returns a tuple $(S_{Split}, P)$ where $P$ is a partition of the vertex set $V$ such that 
\begin{enumerate}
    \item for $X \in P$, and vertices $u,v \in X$ we have  $\mathbf{dist}_{G \setminus E(S_{Split})}(u,v,S) \leq d$, and
    \item for distinct $X, Y \in P$, with vertices $u \in X$ and $v \in Y$,   $u \not\rightleftarrows_{G \setminus E(S_{Split})} v$, and
    \item 
    $|S_{Split}| \leq  \frac{32 \log n}{d} \sum_{X \in P}  \lg (n  - |X \cap S|) |X \cap S|$
\end{enumerate}
The algorithm runs in time $O\left(d \sum_{X \in P} (1 + \lg( n  - |X|)) E(X) \right)$.
\end{lemma}
\begin{proof}
Since correctness was established in Lemma \ref{clm:SplitCorrectness}, it only remains to bound the running time of the procedure. Let us first bound the running time without recursive calls to procedure $\textsc{Split}(G, S,d)$. To see that we only spend $O(|E(G)| d)$ time in $\textsc{Split}(G, S,d)$ excluding recursive calls, observe first that we can find each separator tuple $(S'_{Sep}, V'_{Sep})$ in time $O(E(V'_{Sep}))$ by claim \ref{clm:twowaysep}. We then, either recurse on $G'[V'_{Sep}])$ and remove the vertices $V'_{Sep} \;\cup\; S'_{Sep}$ with their incident edges from $G'$ or we enter the else-case (line \ref{line:split-else-case}). Clearly, if our algorithm never visits the else-case, we only spend time $O(|E(G)|)$ excluding the recursive calls since we immediately remove the edge set that we found in the separator from the graph. 

We further observe that the running time for the GES-tree can be bounded by $O(|E(G)| d)$. The time to compute the separators to prune vertices away from the GES-tree is again combined at most $O(|E(G)|)$ by Lemma \ref{lma:sep} and the observation that we remove edges from the graph $G$ after they were scanned by one such separator procedure.

We already discussed that claim \ref{clm:largeSCCifEStree} implies that we only recurse on disjoint subgraphs with at most $\frac{2}{3}|V|$ vertices. We obtain that each vertex in a final SCC $X$ in $P$ participated in at most $O(\log( n - |X|))$ levels of recursion and so did its incident edges hence we can bound the total running time by $O\left(d \sum_{X \in P} (1 + \log( n  - |X|)) E(X) \right)$.
\end{proof}

\section{Handling deletions}
\label{subsec:delete}

Let us now consider how to process the deletion of an edge $(u,v)$ which we describe in pseudo code in algorithm \ref{alg:delete}. We fix our data structure in a bottom-up procedure where we first remove the edge $(u,v)$ if it is contained in any induced subgraph $\hat{G}_i[X]$ from the GES $\mathcal{E}_{\textsc{Center}(X)}$. 

\begin{algorithm}
\caption{$\textsc{Delete}(u,v)$}
\label{alg:delete}
\KwIn{An edge $(u,v) \in E$.}
\KwResult{Updates the data structure such that queries for the graph $G \setminus \{ (u,v)\}$ can be answered in constant time.}
\BlankLine

\For{ $i = 0 $ \KwTo $ \lfloor \log{n} \rfloor$}{
    \If{If there exists an $X \in \hat{V}_{i+1}$ with $u,v \in X$}{
        $\mathcal{E}_{\textsc{Center}(X)}.\textsc{Delete}(u,v)$\;
    } 
    \While{there exists an $X \in \hat{V}_{i+1}$ with $\mathcal{E}_{\textsc{Center}(X)}.\textsc{GetUnreachable}() \neq \bot$}{
        $X' \gets \mathcal{E}_{\textsc{Center}(X)}.\textsc{GetUnreachable}()$\;
        
        \tcc{Find a separator from $X'$ depending on whether $X'$ is far to reach from $r$ or the other way around.}
        \If{$\mathcal{E}_{\textsc{Center}(X)}.\textsc{Distance}(\textsc{Center}(X),X') > \delta$}{
            $(S_{Sep}, V_{Sep}) \gets \textsc{InSep}(X', \hat{G}_i[X], X \cap S_i, \delta/2)$ \label{lne:DelSepIn}
        }
        \Else(\tcp*[h]{$\mathcal{E}_{\textsc{Center}(X)}.\textsc{Distance}(X', \textsc{Center}(X)) > \delta$}){
            $(S_{Sep} , V_{Sep}) \gets \textsc{OutSep}(X' , \hat{G}_i[X] , X \cap S_i , \delta/2)$\label{lne:DelSepOut}
        }

        \tcc{If the separator is chosen such that $V_{Sep}$ is small, we have a good separator, therefore we remove $V_{Sep}$ from $\mathcal{E}_r$ and maintain the SCCs in $\hat{G}_{i}[V_{Sep}]$ separately. Otherwise, we delete the entire GES $\mathcal{E}_{\textsc{Center}(X)}$ and partition the graph with a good separator.}
        \If{$|\textsc{Flatten}(V_{Sep})| \leq \frac{2}{3}|\textsc{Flatten}(X)|$}{
            $\mathcal{E}_{\textsc{Center}(X)}.\textsc{Delete}(V_{Sep}\;\cup\; S_{Sep})$\;
            $(S'_{Sep}, P') \gets \textsc{Split}(\hat{G}_i[V_{Sep}], V_{Sep} \cap S_i, \delta/2)$\label{lne:DelSplit1}\;
            $S''_{Sep} \gets S_{Sep} \;\cup\; S'_{Sep}$\;
            $P'' \gets P' \;\cup\;  S_{Sep}$\;
        }
        \Else{
            $\mathcal{E}_{\textsc{Center}(X)}.\textsc{Delete}()$\label{lne:ESdelete}\;
            $(S''_{Sep}, P'') \gets \textsc{Split}(\hat{G}_i[X], X \cap S_i, \delta/2)$\label{lne:DelSplit2}\;
        }

        \tcc{After finding the new partitions, we init them, execute the vertex splits on the next level and add the separator vertices.}
        $\textsc{InitNewPartition}(P'', i, \delta)$\;
        
        \ForEach{$Y \in P''$}{
            $\mathcal{E}_{\textsc{Center}(X)}.\textsc{SplitNode}(Y)$\;
        }
        $\mathcal{E}_{\textsc{Center}(X)}.\textsc{Augment}(S''_{Sep})$\label{lne:augmentInDelete}\;
        $S_{i+1} \gets S_{i+1} \;\cup\; S''_{Sep}$\;
    }
}
\end{algorithm}

Then, we check if any GES $\mathcal{E}_{\textsc{Center}(X)}$ on a subgraph $\hat{G}_i[X]$ contains a node that became unreachable due to the edge deletion or the fixing procedure on a level below. Whilst there is such a GES $\mathcal{E}_{\textsc{Center}(X)}$, we first find a separator $S_{Sep}$ from $X'$ in lines \ref{lne:DelSepIn} or \ref{lne:DelSepOut}. We now consider two cases based on the size of the set $\textsc{Flatten}(V_{Sep})$. Whilst focusing on the size of $\textsc{Flatten}(V_{Sep})$ instead of the size of $V_{Sep}$ seems like a minor detail, it is essential to consider the underlying vertex set instead of the node set, since the node set can be further split by node split updates from lower levels.

Now, let us consider the first case, when the set $V_{Sep}$ separated by $S_{Sep}$ is small (with regard to $\textsc{Flatten}(V_{Sep})$); in this case, we simply prune $V_{Sep}$ from our tree by adding $S_{Sep}$ to $S_{i+1}$, and then invoke $\textsc{Split}(\hat{G}_i[V_{Sep}], V_{Sep} \cap S_i, \delta/2)$ to get a collection of subgraphs $P'$ where each subgraph $Y \in P'$ has every pair of nodes $A, B \in Y$ at $S_{i}$-distance $\delta/2$. (We can afford to invoke $\textsc{Split}$ on the vertex set $V_{Sep}$ because we can afford to recurse to on the smaller side of a cut.)

The second case is when $V_{Sep}$ is large compared to the number of vertices in node set of the GES-tree. In this case we do not add $S_{Sep}$ to $S_{i+1}$. Instead we we declare the GES-tree $\mathcal{E}_{\textsc{Center}(X)}$ invalid, and delete the entire tree. We then partition the set $X$ that we are working with by invoking the $\textsc{Split}$ procedure on all of $X$. (Intuitively, this step is expensive, but we will show that whenever it occurs, there is a constant probability that the graph has decomposed into smaller SCCs, and we have thus made progress.)


Finally, we use the new partition and construct on each induced subgraph a new GES-tree at a randomly chosen center. This is done by invoking $\textsc{InitNewPartition}(P', i, \delta)$ that was presented in subsection \ref{subsec:Preprocessing}. We then apply the updates to the graph $\hat{G}_{i+1}$ using the GES-tree operations defined in Lemma \ref{lma:AugmentedGES}. Note, that we include the separator vertices as singleton sets in the partition and therefore invoke $\mathcal{E}_X.\textsc{SplitNode}(\cdot)$ on each singleton before invoking $\mathcal{E}_X.\textsc{Augment}(S''_{Sep})$ which ensures that the assumption from Lemma \ref{lma:AugmentedGES} is satisfied. As in the last section, let us prove the following two Lemmas whose proofs will further justify some of the details of the algorithm. 

We start by showing that because we root the GES-tree for SCC $X$ at a \emph{random} root $r$, if the GES-tree ends up being deleted in \ref{lne:ESdelete} in algorithm \ref{alg:delete}, this means that with constant probability $X$ has decomposed into smaller SCCs, and so progress has been made.

\begin{restatable}{lemma}{participation}[c.f. also \cite{chechik2016decremental}, Lemma 13]
\label{lma:EStreeprob}
Consider an GES $\mathcal{E}_r = \mathcal{E}_{\textsc{Center}(X)}$ that was initialized on the induced graph of some node set $X_{Init}$, with $X \subseteq X_{Init}$, and that is deleted in line \ref{lne:ESdelete} in algorithm \ref{alg:delete}. Then with probability at least $\frac{2}{3}$, the partition $P''$ computed in line \ref{lne:DelSplit2} satisfies that each $X' \in P''$ has $|\textsc{Flatten}(X')| \leq \frac{2}{3}|\textsc{Flatten}(X_{Init})|$.
\end{restatable}
\begin{proof}
Let $i$ be the level of our hierarchy on which $\mathcal{E}_{r}$ was initialized, i.e. $\mathcal{E}_{r}$ was initialized on graph $\hat{G}_i[X_{Init}]$, and went up to depth $\delta$ with respect to $S_i$-distances (see Algorithm \ref{alg:newPart}). 

Let $u_1, u_2, ..$ be the sequence of updates since the GES-tree $\mathcal{E}_r$ was initialized that were either adversarial edge deletions, nodes added to $S_i$ or node splits in the graph $\hat{G}_i[X_{Init}]$. Observe that this sequence is independent of how we choose our random root $r$, since they occur at a lower level, and so do not take any GES-trees at level $i$ into account. Recall, also, that the adversary cannot learn anything about $r$ from our answers to queries because the SCCs of the graph are objective, and so do not reveal any information about our algorithm. We refer to the remaining updates on $\hat{G}_i[X_{Init}]$ as \textit{separator} updates, which are the updates adding nodes to $S_{i+1}$ and removing edges incident to $S_{i+1}$ or between nodes that due to such edge deletions are no longer strongly-connected. We point out that the separator updates are heavily dependent on how we chose our random source. The update sequence that the GES-tree undergoes up to its deletion in line \ref{lne:ESdelete} is a mixture of the former updates that are independent of our chosen root $r$ and the separator updates.

Let $G^j$ be the graph $\hat{G}_i$ after the update sequence $u_1, u_2, ..., u_j$ is applied. Let $X_{max}^j$ be the component of $S_i$-diameter at most $\delta/2$ that maximizes the cardinality of $\textsc{Flatten}(X_{max}^j)$ in $G^j$. We choose $X_{max}^j$ in this way because we want to establish an upper bound on the largest SCC of $S_i$-diameter at most $\delta/2$ in $G^j$. We then show that that if a randomly chosen source deletes a GES-tree (see line \ref{lne:ESdelete}) after $j$ updates, then there is a good probability that $X_{max}^j$ is small. Then by the guarantees of Lemma \ref{lma:split}, the $\textsc{Split}(\cdot)$ procedure in line \ref{lne:DelSplit2} partitions the vertices into SCCs $X'$ of $S_i$-diameter at most $\delta/2$, which all have small $|\textsc{Flatten}(X')|$ because $X_{max}^j$ is small. 

More precisely, let $G^j_r$, be the graph is obtained by applying \emph{all} updates up to update $u_j$ to $\hat{G}_i[X_{Init}]$; here we include the updates $u_1, ..., u_j$, as well as all separator updates up to the time when $u_j$ takes place. (Observe that $G^j$ is independent from the choice of $r$, but $G^j_r$ is not.) Let $X_{max, r}^j$ be the component of $S_i$-diameter at most $\delta/2$ that maximizes the cardinality of $\textsc{Flatten}(X^j_{max, r})$ in this graph $G^j_r$. It is straight-forward to see that since $S_i$-distances can only increase due to separator updates, we have $|\textsc{Flatten}(X_{max, r}^j)| \leq |\textsc{Flatten}(X_{max}^j)|$ for any $r$. Further $|\textsc{Flatten}(X_{max, r}^j)|$ upper bounds the size of any component $X' \in P''$, i.e. $|\textsc{Flatten}(X')| \leq |\textsc{Flatten}(X_{max, r}^j)|$ if the tree $\mathcal{E}_r$ is deleted in line \ref{lne:ESdelete} while  handling update $u_j$; the same bound holds if $\mathcal{E}_r$ is deleted after update $u_j$, because the cardinality of $\textsc{Flatten}(X_{max, r}^j)$ monotonically decreases in $j$, i.e. $|\textsc{Flatten}(X_{max, r}^{j})| \leq |\textsc{Flatten}(X_{max, r}^{j-1})|$ since updates can only increase $S_i$-distances. 

Now, let $k$ be the index, such that 
\[
|\textsc{Flatten}(X_{max}^k)| \leq \frac{2}{3}|\textsc{Flatten}(X_{Init})| < |\textsc{Flatten}(X_{max}^{k-1})|.
\]
i.e. $k$ is chosen such that after the update sequence $u_1, u_2,..., u_{k}$ were applied to $\hat{G}_i[X_{Init}]$, there exists no SCC $X$ in $G^k$ of diameter at most $\delta/2$ with $|\textsc{Flatten}(X)| > \frac{2}{3}|\textsc{Flatten}(X_{Init})|$. 

In the remainder of the proof, we establish the following claim: if we chose some vertex $r \in \textsc{Flatten}(X^{k-1}_{max})$, then the GES-tree would not be been deleted before update $u_k$ took place. Before we prove this claim, let us point out that this implies the Lemma: observe that by the independence of how we choose $r$ and the update sequence $u_1, u_2, ..$, we have that 
\[
Pr[r \in X_{max}^{k-1} | u_1, u_2, ..] = Pr[r \in X_{max}^{k-1}] = \frac{|\textsc{Flatten}(X_{max}^{k-1})|}{|\textsc{Flatten}(X_{Init})} > \frac{2}{3}
\]
where the before-last equality follows from the fact that we choose the root uniformly at random among the vertices in $\textsc{Flatten}(X_{Init})$. Thus, with probability at least $\frac{2}{3}$, we chose a root whose GES-tree is deleted during or after the update $u_k$ and therefore the invoked procedure $\textsc{Split}(\cdot)$ ensures that every SCC $X' \in P''$ satisfies $|\textsc{Flatten}(X')| \leq  |\textsc{Flatten}(X_{max}^k)| \leq \frac{2}{3}|\textsc{Flatten}(X_{Init})|$, as required. 

Now, let us prove the final claim. We want to show that if $r \in X^{k-1}_{max}$, then the GES-tree would not have been deleted before update $u_k$. To do so, we need to show that even if we include the separator updates, the SCC containing $r$ continues to have size at least $\frac{2}{3}|\textsc{Flatten}(X_{Init})|$ before update $u_k$. In particular, we argue that before update $u_k$, none of the separator updates decrease the size of $X^{k-1}_{max}$. The reason is that the InSep computed in Line
\ref{lne:DelSepIn} of Algorithm \ref{alg:delete} is always run from a node $X$ whose $S_i$-distance from $r$ is at least $\delta$. (The argument for an OutSep in Line \ref{lne:DelSepOut} is analogous.)
Now, the InSep from $X$ is computed up to $S_i$-distance $\delta/2$, so by Property \ref{step:separator-distance} of Lemma \ref{lma:sep}, we have that all nodes pruned away from the component have $S_i$-distance at most $\delta/2$ to $X$; this implies that these nodes have $S_i$-distance more than $\delta/2$ from $r$, and so cannot be in $X^{k-1}_{max}$, because $X^{k-1}_{max}$ was defined to have $S_i$-diameter at most $\delta/2$. Thus none of the separator updates affect $X^{k-1}_{max}$ before update $u_k$, which concludes the proof of the Lemma.


\end{proof}

Next, let us analyze the size of the sets $S_i$. We analyze $S_i$ using the inequality below in order to ease the proof of the Lemma. We point out that the term $\lg(n -|X \;\cup\; S_i|)$ approaches $\lg n$ as the SCC $X$ splits further into smaller pieces. Our Lemma can therefore be stated more easily, see therefore Corollary \ref{cor:SisSmall}. 

\begin{lemma}
\label{lma:setS}
During the entire course of deletions our algorithm maintains 
\begin{align}
|S_0| &= n                    & \\
|S_{i+1}| &\leq  \frac{32 \log n}{\delta} \sum_{X \in \hat{V}_{i}}  \lg (n  - |X \cap S_{i}|) |X \cap S_{i}| &  \text{for }i \geq 0
\end{align}
\end{lemma}
\begin{proof}
We prove by induction on $i$. It is easy to see that $S_0$ has cardinality $n$ since we initialize it to the node set in procedure \ref{alg:preprocessing}, and since each set $S_i$ is an increasing set over time.

Let us therefore focus on $i > 0$. Let us first assume that the separator nodes were added by the procedure $\textsc{OutSep}(\cdot)$ (analogously $\textsc{InSep}(\cdot)$). Since the procedure is invoked on an induced subgraph $\hat{G}_i[X]$ that was formerly strongly-connected, we have that either $V_{Sep}$ or $X \setminus (V_{Sep} \;\cup\; S_{Sep})$ (or both) contain at most half the $S_i$-nodes originally in $X$. Let $Y$ be such a side. Since adding $S_{Sep}$ to $S_i$ separates the two sides, we have that RHS of the equation is increased by at least $\frac{32 \log n}{\delta} |Y \cap S_i|$ since $\lg( n - |Y \cap S_i|) |Y \cap S_i| - \lg( n - |X \cap S_i|) |Y \cap S_i| \geq |Y \cap S_i|$. Since we increase the LHS by at most $\frac{4 \log n}{\delta} |Y \cap S_i|$ by the guarantees in Lemma \ref{lma:sep}, the inequality is still holds.

Otherwise, separator nodes were added due to procedure $\textsc{Split}(\cdot)$. But then we can straight-forwardly apply Lemma \ref{lma:splitFull} which immediately implies that the inequality still holds.

Finally, the hierarchy might augment the set $S_i$ in line \ref{lne:augmentInDelete}, but we observe that $f(s) = \lg(n - s) * s$ is a function increasing in $s$ for $s \leq \frac{1}{2} n$ which can be proven by finding the derivative. Thus adding nodes to the set $S_i$ can only increase the RHS whilst the LHS remains unchanged.
\end{proof}

\begin{corollary}
\label{cor:SisSmall}
During the entire course of deletions, we have, for any $i \geq 0$, \[
|S_{i+1}| \leq \frac{32 \lg^2 n}{\delta} |S_{i}|.
\]
 
\end{corollary}

\section{Putting it all together}
\label{sec:alltogether}

By Corollary \ref{cor:SisSmall}, using $\delta = 64 \lg^2 n$, we enforce that each $|S_i| \leq n/2^i$, so $\hat{G}_{\lfloor \lg{n} \rfloor + 1}$ is indeed the condensation of $G$. Thus, we can return on queries asking whether $u$ and $v$ are in the same SCC of $G$, simply by checking whether they are represented by the same node in $\hat{G}_{\lfloor \lg{n} \rfloor + 1}$ which can be done in constant time. 

By Lemma \ref{lma:EStreeprob}, we have that with probability $\frac{2}{3}$, that every time a node leaves a GES, its induced subgraph contains at most a fraction of $\frac{2}{3}$ of the underlying vertices of the initial graph. Thus, in expectation each vertex in $V$ participates on each level in $O(\log n)$ GES-trees. Each time it contributes to the GES-trees running time by its degree times the depth of the GES-tree which we fixed to be $\delta$. Thus we have expected time $O(\sum_{v \in V} \mathbf{deg}(v) \delta \log n) = O(m \log^3 n)$ to maintain all the GES-trees on a single level by Lemma \ref{lma:AugmentedGES}. There are $O(\log n)$ levels in the hierarchy, so the total expected running time is bounded by $O(m \log^4 n)$.

By Lemmas \ref{lma:splitFull}, the running time for invoking $\textsc{Split}(G[X], S, \delta/2)$ can be bounded by $O(E(X) \delta \log n) = O(E(X) \log^3 n)$. After we invoke  $\textsc{Split}(G[X], S, \delta/2)$ in algorithm \ref{alg:delete}, we expect with constant probability again by Lemma \ref{lma:EStreeprob}, that each vertex is at most $O(\log n)$ times in an SCC on which the $\textsc{Split}(\cdot)$ procedure is invoked upon. We therefore conclude that total expected running time per level is $O(m \log^4 n)$, and the overall total is $O(m \log^5 n)$. 

Finally, we can bound the total running time incurred by all invocations of $\textsc{InSep}(\cdot)$ and $\textsc{OutSep}(\cdot)$ outside of $\textsc{Split}(\cdot)$ by the same argument and obtain total running time $O(m \log^2 n)$ since each invocation takes time $O(E(G))$ on a graph $G$.

This completes the running time analysis, establishing the total expected running time $O(m \log^5 n)$ which concludes our proof of Theorem \ref{thm:ContributionSCCmain}.

%% file: rand_decr_sssp.tex
\label{chap_rand_decr_sssp}

In this chapter, we prove the first part of \Cref{thm:ContributionSSSPResult}. We state the precise Theorem that we prove below. 

\begin{theorem}[Decremental Part of \Cref{thm:ContributionSSSPResult}]
\label{thm:ContributionDecrSSSPResult}
Given a decremental input graph $G=(V,E,w)$ with $n = |V|, m=|E|$ and aspect ratio $W$, a dedicated source $r \in V$ and $\epsilon > 0$, there is a randomized algorithm that maintains a distance estimate $\widetilde{\mathbf{dist}}(r,x)$, for every $x \in V$, such that
\[
    {\mathbf{dist}}_G(r,x) \leq \widetilde{\mathbf{dist}}(r,x) \leq (1+\epsilon) {\mathbf{dist}}_G(r,x)
\]
at any stage w.h.p. The algorithm has total expected update time $\tilde{O}(n^2 \log^4 W/\poly(\epsilon))$. Distance queries are answered in $O(1)$ time, and a corresponding path $P$ can be returned in $O(|P|)$ time. The algorithm works against a non-adaptive adversary.
\end{theorem}

We start this chapter by introducing some additional preliminaries that are specific to this chapter. We then give an extended overview of our data structure, and finally provide formal proofs (however, we only sketch some of the more technical proofs).  

\section{Additional Preliminaries}
\label{sec:prelim}

\paragraph{Exponential Distribution.} Finally, we make use of the exponential distribution, that is we use random variables $X$ with cumulative distribution function $F_X(x,\lambda) = 1 - e^{-\lambda x}$ for all $x \geq 0, \lambda > 0$, which we denote by the shorthand $X \sim \textsc{Exp}(\lambda)$. If $X \sim \textsc{Exp}(\lambda)$ is clear, we also use $F_X(x)$ in place of $F_X(x,\lambda)$. The exponential distribution has the special property of being \emph{memoryless}, that is if $X \sim \textsc{Exp}(\lambda)$, then
\[
    \mathbb{P}[X > s + t \;|\; X > t] = \mathbb{P}[X > s].
\]

\paragraph{Generalized Topological Order.} We define a \emph{generalized topological order} \\ $\textsc{GeneralizedTopOrder}(H)$ to be a tuple $(\mathcal{V}, \tau)$ where $\mathcal{V}$ is the set of SCCs of $H$ and $\tau : \mathcal{V} \rightarrow [0,n)$ is a function that maps any sets $X, Y \in \mathcal{V}$ such that $\tau(X) < \tau(Y)$, if $X \leadsto_H Y$ and such that $[\tau(X), \tau(X) + |X|) \cap [\tau(Y), \tau(Y) + |Y|) = \emptyset$. Thus, $\tau$ effectively establishes a one-to-one correspondence between $|X|$-sized intervals and SCCs $X$ in $H$. We point out that a $\textsc{GeneralizedTopOrder}(H)$ can always be computed in $O(|E(H)|)$ time \cite{tarjan1972depth}. In fact, a generalized topological order can also be maintained efficiently in a decremental graph $H$. Here, we say that $(\mathcal{V},\tau)$ is a dynamic tuple that forms a generalized topological order of $H$ if it is a topological order for all versions of $H$. Further, we say that $(\mathcal{V}, \tau)$ has the \emph{nesting} property, if for any set $X \in \mathcal{V}$ and a set $Y \supseteq X$ that was in $\mathcal{V}$ at an earlier stage, we have $\tau(X) \in [\tau(Y), \tau(Y) + |Y| - |X|]$; in other words, the interval $[\tau(X), \tau(X)+|X|)$ is entirely contained in the interval $[\tau(Y), \tau(Y) + |Y|)$. Thus, the associated interval with $X$ is contained in the interval associated with $Y$. We refer to the following result that can be obtained straight-forwardly by combining the data structure given in \Cref{chap:rand_scc} and the static procedure by Tarjan \cite{tarjan1972depth} as described in \cite{GutenbergW20a}.

\begin{theorem}[see \Cref{chap:rand_scc} and \cite{tarjan1972depth, GutenbergW20a}]
\label{thm:SCCinDecrGraph}
Given a decremental digraph $H$, there exists an algorithm that can maintain the generalized topological order $(\mathcal{V}, \tau)$ of $H$ where $\tau$ has the nesting property. The algorithm runs in expected total update time $O(m \log^4 n)$, is randomized and works against an adaptive adversary. 
\end{theorem}

\section{Overview}
\label{sec:overview}

We now give an overview of our algorithm. In order to illustrate the main concepts, we start by giving a simple algorithm to obtain total update time $O(n^2 \log n/\epsilon)$ in directed acyclic graphs (DAGs). While this algorithm was previously not explicitly mentioned, it follows rather directly from the techniques developed in \cite{bernstein2017deterministicweighted, GutenbergW20a}. We present this algorithm to provide intuition for our approach and motivates our novel notion of approximate topological orders. In the light of approximate topological orders, we then shed light on limitations of the previous approach in \cite{GutenbergW20a} and present techniques to surpass these limitations to obtain \Cref{thm:ContributionSSSPResult}. In this overview, we focus on obtaining an SSSP algorithm that runs in $\tilde{O}(n^2 \log^4 W / \poly( \epsilon))$ expected update time.

\subsection{A Fast Algorithm for DAGs} 
\label{subsec:DAGalgo}

\paragraph{The Topological Order Difference.} Let $G = (V,E,w)$ be a DAG and let $\tau$ be the function returned by  $\textsc{GeneralizedTopOrder}(G)$ computed on the initial version of $G$ (since $G$ is a DAG, this is just a standard topological order). Let us now make an almost trivial observation: for any shortest $s$-to-$t$ path $\pi_{s,t}$, in any version of $G$, the sum of topological order differences is bounded by $n$. More formally:
\begin{equation}
\label{eq:defT}
        \mathcal{T}(\pi_{s,t}, \tau) \stackrel{\text{def}}{=} \sum_{(u,v) \in \pi_{s,t}} \tau(v) - \tau(u) = \tau(t) - \tau(s) \leq n.
\end{equation}
Observe that every path in $G$ can only contain few edges $(u,v)$ with large topological order difference, i.e. with $\tau(v) - \tau(u)$ large, by the pigeonhole principle.

\paragraph{Reviewing the ES-tree.} To understand how this fact can be exploited, recall that in the classic ES-tree algorithm (see \Cref{sec:esTree}), a vertex $v$ that currently does not have an in-edge in the tree $T$, searches its in-neighborhood $\mathcal{N}_{in}(v$ for a vertex $x$ such that 
\begin{equation}
\label{eq:enforce2}
    \widetilde{\mathbf{dist}}(r,x)+ w(x,v) \leq \widetilde{\mathbf{dist}}(r,v).
\end{equation}
Only if no such $x$ exists, then $\widetilde{\mathbf{dist}}(r,v)$ has to be incremented. However, observe that if a vertex $v$ was only allowed to scan a certain vertex $x \in \mathcal{N}^{in}(v)$ after every $i$ distance estimate increments (for example whenever $\widetilde{\mathbf{dist}}(r,v) $ is divisible by $i$) then this corresponds to enforcing that at all times $\widetilde{\mathbf{dist}}(r,x)+ w(x,v) + i - 1 \leq \widetilde{\mathbf{dist}}(r,v)$ since we check equation \ref{eq:enforce2} every $i$ steps (in particular, for $i=1$, we get an exact algorithm). Consequently, we get at most $i-1$ additive error in the distance estimate for any $t$ whose shortest $r$-to-$t$ path $\pi_{r,t}$ contains $(x,v)$. On the other hand, we only need to scan and check the edge $(x,v)$, by the classic runtime analysis argument, $\delta / i$ times instead of $\delta$ times\footnote{This trade-off was first observed in \cite{bernstein2017deterministicweighted}.}.

\paragraph{Improving the Running Time.} Let us now exploit Inequality \ref{eq:defT}. We therefore define 
\[
B_j(v) = \{ u \in \mathcal{N}^{in}(v) \textit{ with }  2^{j} \leq \tau(v) - \tau(u) < 2^{j+1}\}
\]
for every $v \in V$ and $0 \leq j \leq \lg n$; the $B_j(v)$ partition the in-neighborhood of $v$ according to topological order difference to $v$. Observe that $|B_j(v)| \leq 2^j$. Now, consider the algorithm as above where for every vertex $v$, instead of checking all $\mathcal{N}^{in}(v)$, we only check edge $(x,v)$ for $x \in B_j(v)$ if 
$\widetilde{\mathbf{dist}}(r,v)$ is divisible by $\lceil2^j\frac{\epsilon \delta}{n}\rceil$.
By the arguments above the total running time now sums to 
\[O\left(\sum_{v \in V} \sum_{0 \leq j \leq \lg n} |B_j(v)| \frac{\delta}{2^j\frac{\epsilon \delta}{n}}\right) = O\left(\sum_{v \in V} \sum_{0 \leq j \leq \lg n} 2^{j+2} \frac{\delta}{2^j\frac{\epsilon \delta}{n}}\right) = \tilde{O}(n^2/\epsilon).
\]


\paragraph{Bounding the Error.} Fix a shortest $r$-to-$t$ path $\pi_{r,t}$, and consider any edge $(u,v) \in \pi_{r,t}$ with $u \in B_{j+1}(v)$. We observe that the edge $(u,v)$ contributes at most an additive error of $2^{j+1}\frac{\epsilon \delta}{n}$ to $\widetilde{\mathbf{dist}}(r,t)$ since it is scanned every $\lceil 2^{j}\frac{\epsilon \delta}{n} \rceil$ distance values and if $\lceil 2^{j}\frac{\epsilon \delta}{n} \rceil$ is equal to $1$ it does not induce any error.

On the other hand, since $u \in B_{j+1}(v)$ we also have $\tau(v) - \tau(u) \geq 2^j$. We can thus charge $n/(2 \epsilon \delta)$ units from $\mathcal{T}(\pi_{r,t}, \tau)$ for each additive error unit; we know from Equation \ref{eq:defT} that $\mathcal{T}(\pi_{r,t}, \tau) \leq n$, so the total additive error is at most 
 $\frac{n}{n/(2\epsilon \delta)} = 2\epsilon\delta$. Thus, for all distances $\approx \delta$ (say in $[\delta/2, \delta)$), we obtain a $(1+4\epsilon)$-multiplicative distance estimate\footnote{Technically, we run to depth $(1+4\epsilon)\delta$ to ensure that vertices' distance estimates are not set to $\infty$ too early.}.

\paragraph{Working with Multiple Distance Scales.} Observe that the data structure above has no running time dependency on $\delta$. Thus, to obtain a data structure that maintains a $(1+2\epsilon)$-approximate distance estimate from $r$ to any vertex $x$, we can simply use $\lg(nW)$ data structures in parallel where the $i^{th}$ data structure has $\delta = 2^i$. A query can then be answered by returning the smallest distance estimate from any data structure, using a min-heap data structure to obtain this smallest estimate in constant time\footnote{Here, we exploit that all distance estimates are overestimates, and at least one of them is $(1+2\epsilon)$-approximate.}. The running time for all data structures is then bounded by $\tilde{O}(n^2 \log W /\epsilon)$.

\subsection{Extending the Result to General Graphs}
\label{subsec:ExtendSSSPToGenGraphs}

We now encourage the reader to verify that in the data structure for DAGs, we used at no point that the graph was acyclic, but rather only used that $\mathcal{T}(\pi_{s,t}, \tau)$ is bounded by $n$ for any path $\pi_{s,t}$. In this light, it might be quite natural to ask whether such a function $\tau$ might exist for general decremental graphs. Surprisingly, it turns out that after carrying out some contractions in $G$ that only distort distances slightly, we can find such a function $\tau$ that comes close in terms of guarantees. We call such a function $\tau$ an \emph{approximate topological order} (this function will no longer encode guarantees about reachability, but it helps for intuition to think of $\tau$ as being similar to a topological order). 

\paragraph{The Approximate Topological Order} We start with the formal definition:

\begin{definition}
\label{def:ATO}
Given a decremental weighted digraph $G = (V,E,w)$ and parameter $\eta_{diam} \geq 0$. We call a dynamic tuple $(\mathcal{V}, \tau)$ an \emph{approximate topological order of $G$ of quality $q > 1$} (abbreviated $\mathcal{ATO}(G, \eta_{diam})$ of quality $q$), if at any stage
\begin{enumerate}
    \item $\mathcal{V} = \{ X_1, X_2, .., X_k\}$ forms a partition of $V$ and a refinement of all earlier versions of $\mathcal{V}$, \label{prop:VUpdate} and

    \item $\tau : \mathcal{V} \rightarrow [0, n)$ is a function that maps each $X \in \mathcal{V}$ to a value $\tau(X)$. If some set $X \in \mathcal{V}$ is split at some stage into disjoint subsets $X_1, X_2, .., X_k$, then we let $\tau(X_{\pi(1)}) = \tau(X)$ and $\tau(X_{\pi(j+1)}) = \tau(X_{\pi(j)}) + |X_{\pi(j)}|$ for each $j < k$ and some permutation $\pi$ of $[1, k]$, and \label{prop:tauUpdate}
    
    \item each $X \in \mathcal{V}$ has weak diameter $\mathbf{diam}(X, G) \leq \frac{|X| \eta_{diam}}{n}$, and \label{prop:ContractLittle}
    
    \item \label{prop:TauTotal} for any two vertices $s, t \in V$, the shortest path $\pi_{s,t}$ in $G$ satisfies $\mathcal{T}(\pi_{s,t}, \tau) \leq q \cdot w(\pi_{s,t}) + n$ where we define $\mathcal{T}(\pi_{s,t}, \tau) \stackrel{\text{def}}{=} \sum_{(u,v) \in \pi_{s,t}} |\tau(X^u) - \tau(X^v)|$.
\end{enumerate}
We say that $(\mathcal{V}, \tau)$ is an $\mathcal{ATO}(G, \eta_{diam})$ of \emph{expected} quality $q$, if $(\mathcal{V}, \tau)$ satisfies properties \ref{prop:VUpdate}-\ref{prop:ContractLittle}, and at any stage, for every $s,t \in V$, $\mathbb{E}[\mathcal{T}(\pi_{s,t}, \tau)] \leq q \cdot w(\pi_{s,t}) + n$.
\end{definition}

Let us expound the ideas captured by this definition. We remind the reader that such a function $\tau$ is required by a data structure that only considers distances in $[\delta/2, \delta)$. Let us consider a tuple $(\mathcal{V}, \tau)$ that forms an $\mathcal{ATO}(G, \epsilon \delta)$ of quality $q$. Then, for any $s$-to-$t$ shortest path $\pi_{s,t} = \langle s = v_1, v_2, \dots v_{\ell} = t\rangle$ in $G$, let $s_i$ and $t_i$ be the first and last vertex on the path in $X_i \in \mathcal{V}$ (see property \ref{prop:VUpdate}) if there are any. Observe that by property \ref{prop:ContractLittle}, the vertices $s_i$ and $t_i$ are at distance at most $\frac{|X_i|\epsilon \delta}{n}$ in $G$. It follows that if we contract the SCC $X_i$, the distance $\mathbf{dist}_{G / X_i}(s,t)$ is at least the distance from $s$ to $t$ in $G$ minus an additive error of at most $\frac{|X_i|\epsilon \delta}{n}$. It follows straight-forwardly, that after contracting all sets in $\mathcal{V}$, we have that distances in $G / \mathcal{V}$ correspond to distances in $G$ up to a negative additive error of at most $\sum_{X_i \in \mathcal{V}} \frac{|X_i| \epsilon\delta}{n} = \frac{n \cdot \delta\epsilon}{n} = \epsilon \delta$. Thus, maintaining the distances in $G / \mathcal{V}$ $(1+2\epsilon)$-approximately is still sufficient for getting a $(1 \pm 2\epsilon)$-approximate distance estimate.

Property \ref{prop:VUpdate} simply ensures that the vertex sets forming the elements of $\mathcal{V}$ decompose over time. Property \ref{prop:tauUpdate} states that $\tau$ assigns each node in $G / \mathcal{V}$ a distinct number in $[0,n)$. It also ensures that if a set $X \in \mathcal{V}$ receives $\tau(X)$ that every later subset of $X$ will obtain a number in the interval $[\tau(X), \tau(X) + |X|)$. Moreover, $\tau$ effectively establishes a one-to-one correspondence between nodes in a version of $G / \mathcal{V}$ and intervals in $[0, n)$ of size equal to their underlying vertex set. Once a set $X\in\mathcal{V}$ decomposes into sets $X_1, X_2, \dots$, property \ref{prop:tauUpdate} stipulates that the intervals that $\tau$ maps $X_1, X_2, \dots$ to are disjoint subintervals of $[\tau(X), \tau(X) + |X|)$. We point out that once $\mathcal{V}$ consists of singletons, each vertex is essentially assigned a single number.

Finally, property \ref{prop:TauTotal} gives an upper bound on the topological order difference. Observe that we redefine $\mathcal{T}(\pi_{s,t}, \tau)$ in a way that is consistent with Definition \ref{eq:defT}. 
In our algorithm, for $\eta_{diam} \approx \epsilon\delta$, we obtain a quality of $\tilde{O}(n/\epsilon\delta)$. Thus, any path $\pi$ of weight $\approx \delta$ has $\mathcal{T}(\pi / \mathcal{V}, \tau) \leq \tilde{O}(n / \epsilon)$ which is very close to the upper bound obtained by the topological order function in DAGs. We summarize this result in the theorem below which is one of our main technical contributions.

\begin{restatable}{theorem}{ATOResultIntro}[see \Cref{sec:ATORealImpl} and \Cref{subsec:bootstrapDense}.]
\label{thm:TopologicalOrderMaintenanceOverview}
For any $0 \leq i \leq \lg(Wn)$, given a decremental digraph $G=(V,E,w)$, we can maintain an $\mathcal{ATO}(G, 2^i)$ of expected quality $\tilde{O}(n/2^i)$. The algorithm runs in total expected update time $\tilde{O}(n^2)$ against a non-adaptive adversary with high probability. 
\end{restatable}

Combining the theorem above and the theorem below which is obtained by generalizing the above decremental SSSP algorithm for DAGs, we obtain our main result \Cref{thm:ContributionSSSPResult}.

\begin{restatable}{theorem}{ssspSuperSimple}[see \Cref{subsec:alphaDeltaSSSPReduction}.]
\label{thm:SSSPEfficientIntro}
Given $G=(V,E,w)$ and $(\mathcal{V}, \tau)$ an  $\mathcal{ATO}(G, \eta_{diam} \approx \epsilon\delta)$, for some depth parameter $\delta > 0$, of quality $q$, a dedicated source $r$ in $V$, and an approximation parameter $\epsilon > 0$. Then, there exists a deterministic data structure $\mathcal{E}_r$ that maintains a distance estimate $\widetilde{\mathbf{dist}}(r,v)$ for each $v \in V$, which is guaranteed to be $(1+\epsilon)$-approximate if ${\mathbf{dist}}(r,v) \in [\delta, 2\delta)$. Distance queries are answered in $O(1)$ time and a corresponding path $P$ can be returned in $O(|P|)$ time. The total update time is $\tilde{O}(n \delta q/\epsilon + n^2)$.
\end{restatable}

\subsection{The Framework by \texorpdfstring{ \cite{GutenbergW20a}}{Gutenberg and Wulff-Nilsen}}

Before we describe our new result, we review the framework in \cite{GutenbergW20a} to construct and maintain an $\mathcal{ATO}(G, \eta_{diam})$. We point out that while the abstraction of an approximate topological order is new to our paper, analyzing the technique in \cite{GutenbergW20a} through the $\mathcal{ATO}$-lens is straight-forward and gives a first non-trivial result. Throughout this review section, we assume that the graph $G$ is unweighted to simplify presentation. This allows us to make use of the following result which states that for vertex sets that are far apart, one can find deterministically a vertex separator that is small compared to the smaller side of the induced partition (to obtain an algorithm for weighted graphs a simple edge rounding trick is sufficient to generalize the ideas presented below).

\begin{definition}
\label{def:balancedVertexSep}
Given graph $G = (V,E,w)$, then we say a partition of $V$ into sets $A, S_{Sep}, B$ is a one-way vertex separator if $A \not\leadsto_{G \setminus S_{Sep}} B$ and $A$ and $B$ are non-empty. 
\end{definition}

\begin{lemma}[see Definition 5 and Lemma 6 in \cite{chechik2016decremental}]
\label{lma:balancedVertexSep}
Given an unweighted graph $G$ of diameter $\mathbf{diam}(G)$. Then we can find sets $A, S_{Sep}, B$ that form a one-way vertex separator such that $|S_{Sep}| \leq \tilde{O}(\frac{\min\{|A|, |B|\}}{\mathbf{diam}(G)})$ in time $O(m)$. 
\end{lemma}

\paragraph{High-level Framework.} The main idea of \cite{GutenbergW20a} is to maintain a tuple $(\mathcal{V}, \tau)$ which is an $\mathcal{ATO}(G, \eta_{diam})$ by setting $(\mathcal{V}, \tau)$ to be $\textsc{GeneralizedTopOrder}(G')$ of some decremental graph $G' \subseteq G$ (over the same vertex set, i.e. $V(G') = V(G)$). It is straight-forward to see that $(\mathcal{V}, \tau)$ satisfies property \ref{prop:VUpdate} in \Cref{def:ATO}, since SCCs in the decremental graph $G'$ decompose. Further, it is not hard to extend the existing algorithm for maintaining SCCs in a decremental graph $G'$ given in \cite{bernstein2019decremental} to also maintain function $\tau$ that obeys property \ref{prop:tauUpdate} in \Cref{def:ATO}. The algorithm to maintain $(\mathcal{V}, \tau)$ given $G'$ runs in total update time $\tilde{O}(m)$ (same as in \cite{bernstein2019decremental}).

So far, we have not given any reason why $G'$ needs to be a subgraph of $G$. To see why we cannot use the above strategy on $G$ directly, recall property \ref{prop:ContractLittle} in the $\mathcal{ATO}$-definition \ref{def:ATO}, which demands that each SCC $X$ has weak diameter at most $\frac{|X|\eta_{diam}}{n}$. This property might not hold in the main graph $G$. In order to resolve this issue, $G'$ is initialized to $G$ and then the diameter of SCCs in $G'$ is monitored. Whenever an SCC $X$ violates property \ref{prop:ContractLittle}, a vertex separator $S_{Sep}$ is found in the graph $G'[X]$ as described in \Cref{lma:balancedVertexSep} and all edges incident to $S_{Sep}$ are removed from $G'$. Letting $S$ denote the union of all such separators $S_{Sep}$, we can now write $G' = G \setminus E(S)$. 

\paragraph{Establishing the Quality Guarantee.} To establish a quality $q$ of the $\mathcal{ATO}$ as described in property \ref{prop:TauTotal} in \Cref{def:ATO}, let us first partition the set $S$ into sets $S_0, S_1, \dots, S_{\lg n}$ where each $S_i$ contains all separator vertices found on a graph of size $[n/2^i, n/2^{i+1})$, thus it was found when the procedure from \Cref{lma:balancedVertexSep} was invoked on a graph with diameter at least $\frac{(n/2
^{i+1})\eta_{diam}}{n} =  \frac{\eta_{diam}}{2
^{i+1}}$. Since separators are further balanced, i.e. there size is controlled by the smaller side of the induced cut, we can further use induction and \Cref{lma:balancedVertexSep} to establish that there are at most $\tilde{O}(\frac{2^i n}{\eta_{diam}})$ vertices in $S_i$. Next, observe that since any separator that was added to $S_i$ was found in a graph $G'[X]$ with $|X| \in [n/2^i, n/2^{i+1})$, we have by property \ref{prop:tauUpdate} that nodes $X' \subseteq X$ that are in the current version of $\mathcal{V}$ are assigned a $\tau(X')$ from the interval $[\tau(X), \tau(X) + |X|)$. Thus, any edge $(x,s)$ or $(s,x)$ with $x \in X, s \in S_i \cap X$ has topological order difference $|\tau(X^x) - \tau(X^s)| \leq |X| \leq n/2^i$.

Finally, let us define
\begin{equation}
\label{eq:redefT}
        \mathcal{T}'(\pi_{s,t}, \tau) \stackrel{\text{def}}{=} \sum_{(u,v) \in \pi_{s,t}} \min\{0, \tau(X^v) - \tau(X^u)\}
\end{equation}
the function similar to $\mathcal{T}(\pi_{s,t}, \tau)$ that only captures negative terms, i.e. sums only over edges that go "backwards" in $\tau$. Observe that $\mathcal{T}(\pi_{s,t}, \tau) \leq 2 \mathcal{T}'(\pi_{s,t}, \tau) + n$. Now, since $(\mathcal{V}, \tau)$ is a $\textsc{GeneralizedTopOrder}(G')$, we have that $(u,v)$ occurs in the sum of $\mathcal{T}'(\pi_{s,t}, \tau)$ if and only if $(u,v) \in G \setminus G'$, so one endpoint is in a set $S_i$ and therefore $\tau(X^v) -\tau(X^u) \leq n/2^i$. Since we only have two edges on any shortest path incident to the same vertex, we can establish that
\[
    \mathcal{T}'(\pi_{s,t}, \tau) \leq \sum_{i} 2|S_i| n/2^i = \tilde{O}(n^2/\eta_{diam}).
\]
We obtain that $(\mathcal{V}, \tau)$ is a $\mathcal{ATO}(G, \eta_{diam})$ of quality $\tilde{O}(\frac{n^2}{2^i\eta_{diam}})$ for all paths of weight at least $2^i$. Thus, when the distance scale $\delta \geq \sqrt{n}/\epsilon$, \Cref{thm:SSSPEfficientIntro} requires total update time $\approx n^{2.5}$ to maintain $(1+\epsilon)$-approximate SSSP. For distance scales where $\delta < \sqrt{n}/\epsilon$, a classic ES-tree has total update time $\approx m\sqrt{n} \leq n^{2.5}$.

\paragraph{Limitations of the Framework.} 
 Say that the goal is to  maintain shortest paths of length around $\sqrt{n}$.
The first step in the framework of \cite{GutenbergW20a} is to find separator $S$ such that all SCCs of $G' = G \setminus E(S)$ have diameter at most $\epsilon \sqrt{n}$ and then maintain $(\mathcal{V},\tau) =$ \textsc{GeneralizedTopOrder}($G'$).  Every edge $(u,v) \notin E(S)$ will only go forward in $\tau$, but each edge $(u,s)$, for $s \in S$, can go ``backwards" in $\tau$. By the nesting property of generalized topological orders, the amount that $(u,s)$ goes backwards -- i.e. the quantity $|\tau(X^u) - \tau(X^s)|$ -- is upper bounded by the size of the SCC in $G'$ from which $s$ was chosen: the original SCC has size $n$, but as we add vertices to $S$, the SCCs of $G' = G \setminus E(S)$ decompose and new vertices added to $S$ may belong to smaller SCCs. Define $S^* \subseteq S$ to contain all vertices $s \in S$ that were chosen in an SCC of size $\Omega(n)$. Intuitively, $S^*$ is the top-level separator chosen in $G$, before SCCs decompose into significantly smaller pieces. Every edge in $E(S^*)$ may go backwards by as much as $n$ in $\tau$, so for any path $\pi_{x,y}$ in $G$, the best we can guarantee is that $\mathcal{T}(\pi_{x,y}) \sim n \cdot |\pi_{x,y} \cap S^*|$. 

The framework of \cite{GutenbergW20a} tries to find a \emph{small} separator $S^*$ and then uses the trivial upper bound $|\pi_{x,y} \cap S^*| \leq |S^*|$. In fact, one can show that given \emph{any} deterministic separator procedure, the adversary can pick a sequence of updates where $|\pi_{x,y} \cap S^*| \sim |S^*|$. But now, say that $G$ is a $\sqrt{n} \times \sqrt{n}$-grid graph with bidirectional edges. It is not hard to check that  $|S^*| = \Omega(\sqrt{n})$, because every balanced separator of a grid has $\Omega(\sqrt{n})$ vertices. The framework of \cite{GutenbergW20a} can thus at best guarantee  $\mathcal{T}(\pi_{x,y}) \sim n \cdot |\pi_{x,y} \cap S^*| \sim n|S^*| = \Omega(n^{1.5})$, which is a $\sqrt{n}$ factor higher than it would be in a DAG, and thus leads to running time $\tilde{O}(n^{2.5})$ instead of $\tilde{O}(n^2)$.

Our algorithm uses an entirely different random separator procedure. We allow $S^*$ to be arbitrarily large, but use randomness to ensure that $|\pi_{x,y} \cap S^*|$ is nonetheless small. 

\subsection{Our Improved Framework}

We now introduce our new separator procedure and then show how it can be used in a recursive algorithm that uses ATOs of worse quality (large $q$) to compute ATOs of better quality (small $q$). (By contrast, the framework of \cite{GutenbergW20a} could not benefit from a multi-layered algorithm because it would still hit upon the fundamental limitation outlined above.)

\paragraph{A New Separator Procedure.} Before we describe the separator procedure, let us formally define the guarantees that we obtain. In the lemma below, think of $\zeta = \Theta(\log(n))$.

\begin{restatable}{lemma}{randomLayerSep}
\label{lma:sepIntro}
There exists a procedure $\textsc{OutSeparator}(r, G, d, \zeta)$ where $G$ is a weighted graph, $r \in V$ a root vertex, and integers $d, \zeta > 0$. Then, with probability at least $1- e^{-\zeta}$, the procedure computes a tuple $(E_{Sep}, V_{Sep})$ where edges $E_{Sep} \subseteq E$, and vertices $V_{Sep} = \{ v \in V | r \leadsto_{G \setminus E_{Sep}} v\}$ such that
\begin{enumerate}
    \item for every vertex $v \in V_{Sep}$, $\mathbf{dist}_{G \setminus E_{Sep}}(r,v) \leq d$, and
    \item for every $e \in E$, we have $\mathbb{P}[e \in E_{Sep} | r \leadsto_{G \setminus E_{Sep}} \mathbf{tail}(e)] \leq \frac{\zeta}{d}w(e)$\label{prop:lowProb}.
\end{enumerate}
Otherwise, it reports $\mathbf{Fail}$. The running time of $\textsc{OutSeparator}(\cdot)$ can be bounded by $O(|E(V_{Sep})|\log n)$.
\end{restatable}

In fact, \Cref{alg:OutseparatorIntro} gives a simple implementation of procedure $\textsc{OutSeparator}(\cdot)$. Here, we pick a ball $B = B^{out}(r, X)$ in the graph $G$ from $r$ to random depth $X$, and then simply return the tuple $(E_{Sep}, V_{Sep}) = (E(B, \overline{B}), B)$ where $E(B, \overline{B})$ are the edges $(u,v)$ with $u \in B$ but $v \not\in B$. The procedure thus only differs from a standard edge separator procedure in that we choose $X$ according to the exponential distribution $\textsc{Exp}(\frac{\zeta}{d})$.

\begin{algorithm}
\caption{$\textsc{OutSeparator}(r, G, d, \zeta)$}
\label{alg:OutseparatorIntro}
Choose $X \sim \textsc{Exp}(\frac{\zeta}{d})$.\;
\lIf{$X \geq d$}{\Return \textbf{Fail}\label{lne:failSeparatorIntro}}
Compute the Ball $B = B^{out}(r, X) = \{ v \in V | \mathbf{dist}(r,v) \leq X\}$\;
\Return $(E(B, \overline{B}), B)$
\end{algorithm}

A proof of Lemma \ref{lma:sepIntro} is now straightforward. We return \textbf{Fail} in \Cref{lne:failSeparatorIntro} with probability $\mathbb{P}[X \geq d] = 1 - F_X(d) = 1 - (1 - e^{- \frac{\zeta}{d} \cdot d}) = e^{-\zeta}$ (recall from \cref{sec:prelim} that $F_X(d)$ is shorthand for $F(x, \frac{\zeta}{d})$, the cumulative distribution function of an exponential distribution with parameter $ \frac{\zeta}{d}$). Assuming no failure, we have $E_{Sep} = E(B, \overline{B})$, and it is easy to see that $V_{Sep} = \{ v \in V | r \leadsto_{G \setminus E_{Sep}} v \} = B$, so Property 1 of Lemma \ref{lma:sepIntro} holds by definition of $B$. Moreover, we can compute $B$ in the desired $O(|E(B)|\log n)$ time by using Dijkstra's algorithm by only extracting a vertex from the heap if it is at distance at most $X$. Finally, for property \ref{prop:lowProb}, note that $e \in E_{sep}$ iff $\dist(r, \mathbf{tail}(e)) \leq X < \dist(r, \mathbf{tail}(e)) + w(e)$. Thus,
\begin{align*}
    \mathbb{P}[e \in E_{Sep} | r \leadsto_{G \setminus E_{Sep}} \mathbf{tail}(e)]
    &= \mathbb{P}[X < \mathbf{dist}(r,\mathbf{tail}(e)) + w(e) | X \geq \mathbf{dist}(r,\mathbf{tail}(e))]\\
    &= \mathbb{P}[X < w(e)] = F_X(w(e)) = 1 - e^{- \frac{\zeta}{d} w(e)} \\
    &\leq 1 - \left(1 - \frac{\zeta}{d} w(e)\right) = \frac{\zeta}{d} w(e)
\end{align*}
where the second equality follows from the memory-less property of the exponential distribution, and the inequality holds because $1 + x \leq e^x$ for all $x \in \mathbb{R}$. We point out that the technique of random ball growing using the exponential distribution is not a novel contribution in itself and has be previously used in the context of low-diameter decompositions \cite{linial1993low, bartal1996probabilistic, miller2013parallel, pachocki2018approximating} which have recently also been adapted to dynamic algorithms \cite{forster2019dynamic, chechik2020dynamic}.  

\paragraph{A New Framework.} Let us now outline how to use \Cref{lma:sepIntro} to derive 
\Cref{thm:TopologicalOrderMaintenanceOverview} which is stated below again. The full details and a rigorous proof are provided in \Cref{sec:ATORealImpl}.

\ATOResultIntro*

As in \cite{GutenbergW20a}, we maintain a graph $G' \subseteq G$ and its generalized topological order $(\mathcal{V}, \tau)$. Whenever the diameter of an SCC $X$ in $G'$ is larger than $\frac{|X|\eta_{diam}}{n}$, we now use the separator procedure described in \Cref{lma:sepIntro} with $d = \frac{|X|\eta_{diam}}{2n}$ from some vertex $r$ in $X$ with $|B^{out}(r, d = \frac{|X|\eta_{diam}}{2n})| \leq |X|/2$. Such a vertex exists since by definition of diameter, as we can find two vertices with disjoint balls. We obtain an edge separator $E_{Sep}$ and update $G'$ by removing the edges in $E_{Sep}$. Let the union of all edge separators be denoted by $F$ and again observe that $G' = G \setminus F$. It is not hard to see that our scheme still ensures properties \ref{prop:VUpdate}-\ref{prop:ContractLittle} in \Cref{def:ATO}. We now argue that the quality improved to $\tilde{O}(n/2^i)$.

Partition $F$ into sets $F_0, F_1, \dots, F_{\lg n}$ where each $F_j$ contains all separator edges found on a graph of size $[n/2^{j+1}, n/2^j)$. We again have that every edge $(u,v) \in F_j$ has $|\tau(X^u) - \tau(X^v)| \leq n/2^{j}$. Let us establish an upper bound on the number of edges in $F_j$ on any shortest path $\pi_{s,t}$. Consider therefore any edge $e \in \pi_{s,t}$, at a stage where both endpoints of $e$ are in a SCC $X$ of size $[n/2^{j+1}, n/2^j)$ and where we compute a tuple $(V_{Sep}, E_{Sep})$ from some $r \in X$. Now, observe that if $\mathbf{tail}(e) \not\in V_{Sep}$, then $e$ can not be in $E_{Sep}$. So assume that $\mathbf{tail}(e) \in V_{Sep}$. Then, we have $e$ joining $F_j$ with probability 
\[
\mathbb{P}[e \in E_{Sep} | r \leadsto_{H \setminus E_{Sep}} \mathbf{tail}(e)] = \frac{\zeta}{2^i|X|/n} w(e) = \tilde{O}\left(\frac{2^j w(e)}{2^i}\right)
\]
according to \Cref{lma:sepIntro}, where we set $\zeta = \tilde{O}(1)$ to obtain high success probability. But if $e$ did not join $F_j$ at that stage, then it is now in a SCC of size at most $n/2^{j+1}$ (recall we chose $|B^{out}(r, d)| \leq |X|/2$, and we have $V_{Sep} \subseteq B^{out}(r, d)$). Thus, $e$ cannot join $F_j$ at any later stage. 

Now, it suffices to sum over edges on the path $\pi_{s,t}$ and indices $j$ to obtain that
\[
    \mathcal{T}'(\pi_{s,t}, \tau) \leq \sum_{e \in \pi_{s,t}} \sum_{j} \mathbb{P}[e \in F_j] \cdot n/2^{j} = \sum_{e \in \pi_{s,t}} \sum_{j} \tilde{O}\left(\frac{2^j w(e)}{2^i}\right) \cdot n/2^{j} = \tilde{O}\left(\frac{ w(\pi_{s,t}) n}{2^i}\right)
\]
giving quality $\tilde{O}\left(\frac{n}{2^i}\right)$. 

\paragraph{Efficiently Maintaining $G'$.} 
\newcommand{\asssp}{\mathcal{A}_{SSSP}}
As shown above, maintaining an $\mathcal{ATO}(G,2^i)$ requires detecting when any SCC in $G' = G \setminus E(S)$ has diameter above $2^i$. We start by showing how to do this efficiently if we are given a black-box algorithm $\asssp$ that maintains distance estimates up to depth threshold $2^i$  (i.e. if a vertex is at distance less than $2^i$ from the source vertex, there is a distance estimate with good approximation ratio). 

We use the random source scheme introduced in \cite{roditty2008improved} along with some techniques developed in \cite{chechik2016decremental, bernstein2019decremental, GutenbergW20a}: we choose for each SCC $X$ in $G'$ a center vertex $\textsc{Center}(X) \in X$ uniformly at random, and use $\asssp$ to maintain distances from $\textsc{Center}(X)$ to depth $2^i|X|/n \leq 2^i$. Since the largest distances between the vertex $\textsc{Center}(X)$ and any other vertex in $X$ is a $2$-approximation on the diameter of $G[X]$, this is sufficient to monitor the diameter and to trigger the separator procedure in good time.


Cast in terms of our new ATO-framework, the previous algorithm of \cite{GutenbergW20a} used a regular Even and Shiloach tree for the algorithm $\asssp$. We instead use a recursive structure, where $\mathcal{ATO}s$ of bad quality (large $q$) are used to build $\mathcal{ATO}s$ of better quality (small $q$). Recall that our goal is to build an $\mathcal{ATO}(G, 2^{i})$ of quality $\tilde{O}(n/2^i)$ and say that $X$ is some SCC of $G' = G \setminus E(S)$ whose diameter we are monitoring. Now, using the lower level of the recursion, we inductively assume that we can maintain an $\mathcal{ATO}(G'[X], 2^{i-1})$ of quality $\tilde{O}(n/2^{i-1})$ in time $\tilde{O}(|X|^2)$. Plugging this $\mathcal{ATO}$ into \Cref{thm:SSSPEfficientIntro} gives us an algorithm for maintaining distances up to depth $2^i$ in $X$ with total update time $\tilde{O}(|X|^2)$, Summing over all components $X$ in $G'$, we get an $\tilde{O}(n^2)$ total update time to maintain $\mathcal{ATO}(G, 2^{i})$, as desired.


We actually cheated a bit in the last calculation, because the scheme above could incur an additional logarithmic factor for computing $\mathcal{ATO}(G, 2^i)$ from all the $\mathcal{ATO}(G'[X], 2^{i-1})$, so we can only afford a sublogarithmic number of levels, which leads to an extra $n^{o(1)}$ factor in the running time. However, a careful bootstrapping argument allows us to avoid this extra term.


\subsection{Organization}

We recommend the reader to carefully study \Cref{sec:overview} to gain necessary intuition for our approach. In \Cref{sec:ATORealImpl}, we give an efficient reduction from maintaining an approximate topological order to depth-restricted SSSP. This section is the centerpiece of the article and its main result, \Cref{lma:reductionATOtoSSSP}, is one of our main technical contributions. We then show how to use \Cref{lma:reductionATOtoSSSP} to obtain a SSSP data structure for dense graphs in \Cref{sec:SSSPDense}.

\section{Reducing Maintenance of an ATO to \texorpdfstring{ $\alpha$-approximate $\delta$-restricted $\mathcal{SSSP}$}{SSSP}}
\label{sec:ATORealImpl}

In this section, we show how to obtain an $\mathcal{ATO}$ given an $\alpha$-approximate $\delta$-restricted $\mathcal{SSSP}$ data structure. We start by defining such a data structure and then give a reduction.

\begin{definition}
\label{def:SSSPGeneric}
Let $\mathcal{A}$ be a data structure that given any decremental directed weighted graph $G$, a dedicated source $r \in V$, an approximation parameter $\alpha > 1$, maintains for each vertex $v \in V$, distance estimates $\widetilde{\mathbf{dist}}(r,v)$ and $\widetilde{\mathbf{dist}}(v,r)$ such that at any stage of $G$, for every pair $(s,t) \in (\{r\} \times V) \cup (V \times \{r\})$
\begin{itemize}
    \item we have $\mathbf{dist}(s,t) \leq \widetilde{\mathbf{dist}}(s,t)$, and
    \item if $\mathbf{dist}(s,t) \leq \delta$, then $\widetilde{\mathbf{dist}}(s,t) \leq \alpha \mathbf{dist}(s,t)$.
\end{itemize}
Then, we say $\mathcal{A}$ is an $\alpha$-approximate $\delta$-restricted $\mathcal{SSSP}$ data structure with running time $T_{SSSP}(m,n,\delta,\alpha)$. We require only that $\mathcal{A}$ runs against a non-adaptive adversary. 
\end{definition}
\begin{remark}
In the rest of the article, we implicitly assume that all $\mathcal{SSSP}$ data structures have $T_{SSSP}(m,n,\delta,\alpha)$ monotonically increasing in the first two parameters.
\end{remark}

We also need another definition that makes it more convenient to work with $\mathcal{ATO}$s that only have \emph{expected} quality (see \Cref{def:ATO}). However, we require high probability bounds in our constructions and it will further be easier to work with deterministic objects. This inspires the definition of an $\mathcal{ATO}$-bundle which is a collection of $\mathcal{ATO}$s such that for each path of interest, there is at least one $\mathcal{ATO}$ in the bundle that has good quality for the path at-hand.

\begin{definition}[$\mathcal{ATO}(G, \eta_{diam}, \ell)$-bundle]
\label{def:ATObundle}
Given a decremental weighted directed graph $G = (V,E,w)$ and parameter $\eta_{diam} \geq 0$. We call $\mathcal{S} = \{ (\mathcal{V}_i, \tau_i)\}_{i \in [1, \ell]}$ an $\mathcal{ATO}(G, \eta_{diam}, \ell)$-bundle of quality $q$ if every $(\mathcal{V}_i, \tau_i)$ is an $\mathcal{ATO}(G, \eta_{diam})$ and for any two vertices $s, t \in V$, there exists an $i \in [1, \ell]$, such that the shortest path $\pi_{s,t}$ in $G$ satisfies $\mathcal{T}(\pi_{s,t}, \tau_i) \leq q \cdot w(\pi_{s,t}) + n$.
\end{definition}

Without further due, let us state and prove the main result of this section.

\begin{restatable}[$\mathcal{ATO}$-bundle from $\mathcal{SSSP}$]{theorem}{atoBundle}
\label{lma:reductionATOtoSSSP}
Given an algorithm $\mathcal{A}$ to solve $2$-approximate $\delta$-restricted $\mathcal{SSSP}$ on any graph $H$ in time $T_{SSSP}(m,n,\delta)$, and for any $c > 0$, we can maintain an $\mathcal{ATO}(G, 2\alpha\delta, 40 c \log n)$-bundle of quality $\frac{(c+2)40000 n \log^5 n}{\delta}$ in total expected update time 
\begin{equation}\label{eq:totalRunningTime}
O\left(\sum_{j=0}^{\lceil\lg \delta \rceil} \; \sum_{k=0}^{2^{j+3} c \log^2 n} T_{SSSP}(m_{j,k}, n/2^{j}, \delta, 2) + m \log^3 n\right)
\end{equation}
where $\sum_j m_{j,k} \leq 16c \cdot m \log^2 n$ for all $k$. The algorithm runs correctly with probability $1 - n^{-c}$ for any $c > 0$. 
\end{restatable}
\begin{remark}
\label{rmk:LaminarFamilyOfGraphs}
The graphs that the data structure $\mathcal{A}$ runs upon during the algorithm are vertex-induced subgraphs of $G$. The data structure $\mathcal{A}$ is further allowed to maintain distances on a larger subgraph of $G$, i.e. when $\mathcal{A}$ is applied to a graph $G[X]$, it can run instead on $G[Y]$ for any set $X \subseteq Y \subseteq V$.
\end{remark}

We point out that \Cref{rmk:LaminarFamilyOfGraphs} is only of importance at a later point at which the reader will be reminded and can safely be ignored for the rest of this sections (the reader is however invited to verify its correctness which is easy to establish).

We now describe how to obtain an efficient algorithm that obtains an $\mathcal{ATO}(G,2\alpha \delta)$ henceforth denoted by $(\mathcal{V}, \tau)$. The next sections describe how to initialize $(\mathcal{V}, \tau)$, how to maintain useful data structures to maintain the diameter, give the main algorithm and then a rigorous analysis. Finally, we obtain a $\mathcal{ATO}(G, 2\alpha\delta, 40 c \log n)$-bundle by running $40 c \log n$ independent copies of the algorithm below.

\subsection{Initializing the Algorithm} 
\label{subsec:InitATO}

As described in \Cref{sec:overview}, our goal is to maintain a graph $G'$ that is a subgraph of $G$ and satisfies that no SCC $X$ in $G'$ has weak diameter $\mathbf{diam}(X,G)$ larger than $\frac{ \delta |X|}{n}$. Throughout, we maintain the generalized topological order $(\mathcal{V}, \tau)$ on $G'$ where $\tau$ has the nesting property as described in \Cref{thm:SCCinDecrGraph}. 

To ensure the diameter constraint initially, we use the following partitioning procedure whose proof can be found in \cite{nearOptDenseSSSP}.

\begin{restatable}[Partitioning Procedure]{lemma}{partition}
\label{lma:partitionFull}
Given an algorithm $\mathcal{A}$ to solve $2$-approximate $\delta$-restricted $\mathcal{SSSP}$. There exists a procedure $\textsc{Partition}(G, d, \zeta)$ that takes weighted digraph $G$, a depth threshold $d \leq \delta$ and a success parameter $\zeta > 0$, and returns a set $E_{Sep} \subseteq E$ such that
\begin{enumerate}
    \item for each SCC $X$ in $G \setminus E_{Sep}$, we have for any vertices $u,v \in X$ that $\mathbf{dist}_{G \setminus E_{Sep}}(u,v) \leq d$, and \label{prop:partitionDiameter}
    \item \label{prop:partitionProb}
    for $e \in E$, we have
    $\mathbb{P}[e \in E_{Sep}] \leq \frac{240 \zeta \log^2 n}{d}w(e).$
\end{enumerate}
The algorithm runs in total expected time 
\[
O\left(\sum_{j=0}^{\lceil\lg \delta \rceil} \; \sum_{k=0}^{2^{j+1}} T_{SSSP}(m_{j,k}, n/2^{j}, \delta, 2) + m \log^2 n\right)
\]
where we have that $\sum_{k=0} m_{j,k} \leq 2m$ for every $i$. The algorithm terminates correctly with probability $1-e^{-\zeta}$ for any $c > 0$. 
\end{restatable}
\begin{remark}
During the execution, the graphs on which we use the $\mathcal{SSSP}$ structure upon have the properties as described in \Cref{rmk:LaminarFamilyOfGraphs}.
\end{remark}

\begin{algorithm}
\caption{$\textsc{Init}()$}
\label{alg:init}
Let $G'$ be initialized to $G$.\;
\For{$i = 0$ to $\lceil\lg \delta \rceil$}{
    Compute the SCCs $\mathcal{V}$ of $G'$\;
    \ForEach{SCC $X$ in $\mathcal{V}$, $|X| \leq n/2^i$}{
        $E_{Sep} \gets \textsc{Partition}(G[X], \delta/ 2^i, (c+2)\log n)$\;
        $G' \gets G' \setminus E_{Sep}$\;
    }
}
\Return $G'$\;
\end{algorithm}

Using this procedure, it is straight-forward to initialize our algorithm. The pseudo-code of the initialization procedure is given in \Cref{alg:init}. Here, we iteratively apply the partitioning procedure to SCCs of small size to decompose them further if their diameter is too large. It is not hard to establish that the graph $G'$ returned by the procedure, satisfies that every SCC $X$ in $G'$ has $\mathbf{diam}(X,G)\leq \frac{\delta |X|}{n}$. 

\subsection{Maintaining Information about SCC Diameters}
\label{subsec:maintainSSSPs}

Before we describe how to maintain $G'$ to satisfy the guarantees given above, we address the issue of maintaining information about the diameter of the current SCCs in $G'$. 

Therefore, we maintain a set $S$ of random sources throughout the algorithm, and from each $s \in S$, we run an $\alpha$-approximate $\delta$-restricted $\mathcal{SSSP}$ data structure $\mathcal{A}_s$. Initially $S = \emptyset$, and whenever there is an SCC $X$ in $\mathcal{V}$ (which is maintained by the data structure on $G'$), and we find $S \cap X = \emptyset$, we pick a vertex $s$ uniformly at random from $X$ and add it to $S$. Once added, we initialize and maintain an $\alpha$-approximate $\delta$-restricted $\mathcal{SSSP}$ data structure $\mathcal{A}_s$ on the current version of $G[X]$. That is, even if $X$ does not form an SCC at later stages, the data structure is run until the rest of the algorithm on the graph $G[X]$. This ensures that once the algorithm is invoked, all edge updates are determined by the adversary formulating updates to $G$. Since we assume that the adversary is non-adaptive, we have that the $\mathcal{SSSP}$ data structure only has do deal with updates from a non-adaptive adversary\footnote{If we would instead remove vertices from the data structure, we would do so based on the information gathered from the data structure. Thus, the data structure would be required to work against an adaptive adversary. A similar problem arises when running on $G'$.}.

We point out that since we maintain $G'$ to be a decremental graph, we have that $\mathcal{V}$ forms a refinement of previous versions at any stage i.e. the SCC sets only decompose over time in $G'$. Therefore, we can never have multiple center vertices in the same SCC $X \in \mathcal{V}$. For convenience, we let for each $X \in \mathcal{V}$, the vertex $\{s\} = X \cap S$ be denoted by $\textsc{Center}(X)$. By the above argument, this function is well-defined.

\subsection{Maintaining \texorpdfstring{$G'$}{G'}}
\label{subsec:maintainGPrime}

Let us now describe the main procedure of our algorithm: the part that efficiently handles violations of the diameter constraint by finding new separators. The implementation of this procedure is given by \Cref{alg:diameterVio}. Let us now provide some intuition and detail as to how the algorithm works.

\begin{algorithm}
\caption{$\textsc{ResolveDiameterViolations}()$}
\label{alg:diameterVio}
\While(\label{lne:loop}){there exists an $X \in \mathcal{V}$, where $\mathcal{A}_{\textsc{Center}(X)}$ has a distance estimate $\widetilde{\mathbf{dist}}(\textsc{Center}(X), t)$ or $\widetilde{\mathbf{dist}}(t, \textsc{Center}(X))$ exceeding $\frac{\delta|X|}{n}$ for some vertex $t \in X$}{
    \tcc{Find separator sets that decompose $X$.}
    \If(\label{lne:ifFar}){$\widetilde{\mathbf{dist}}(t,\textsc{Center}(X)) > \frac{|X|\delta}{n}$}{
        $(E_{Sep}, C) \gets \textsc{OutSeparator}(t, G'[X], \frac{|X|\delta}{2n}, (c+2)\log n)$ \label{lne:DelSep}
    }
    \Else{
        $(E_{Sep}, C) \gets \textsc{OutSeparator}(t, \overleftarrow{G'[X]}, \frac{|X|\delta}{2n}, (c+2)\log n)$ \label{lne:DelSepRev}
    }
    $E'_{Sep} \gets \textsc{Partition}(G'[C], \frac{|X|\delta}{4n}, (c+2)\log n))$\label{lne:partitionXC}\;
    \tcc{Update $G'$, $\mathcal{V}$ and $\tau$ to reflect the changes.}
    $G' \gets G' \setminus (E_{Sep} \cup E'_{Sep})$\label{lne:updateGraph}\;
    \textbf{Wait Until} the generalized topological order $(\mathcal{V}, \tau)$ of $G'$ was updated, each SCC $Z$ in $G'$ has a center $\textsc{Center}(Z)$, and all data structures $\mathcal{A}_{s}$ are updated.\label{lne:WaitUntil}
}
\end{algorithm}

The algorithm runs a while-loop starting in \Cref{lne:loop} that checks whether there exists a SCC $X \in \mathcal{V}$, such that the $\alpha$-approximate $\delta$-restricted SSSP data structure $\mathcal{A}_{\textsc{Center}(X)}$ has one of its distance estimates $\widetilde{\mathbf{dist}}(\textsc{Center}(X), t)$ (or $\widetilde{\mathbf{dist}}(t, \textsc{Center}(X))$) exceeding $\frac{\delta|X|}{n}$ for some vertex $t$ in the same SCC $X$ in $G'$. The goal of the while-loop iteration, is then to find a separator $E_{Sep}$ between $\textsc{Center}(X)$ and $t$ and to delete the edges from $G'$.

Let us describe a loop-iteration where some distance estimate $\widetilde{\mathbf{dist}}(\textsc{Center}(X), t)$ was found that exceeded $\frac{\delta|X|}{n}$ and where $t \in X$ (the case where we have a distance estimate $\widetilde{\mathbf{dist}}(\textsc{Center}(X), t)$ exceed the threshold value is analogous and therefore omitted). In this case, we find a separator $E_{Sep}$ that separates vertices in $C$ (where $t \in C$) from vertices in $X \setminus C$ (where $\textsc{Center}(X) \in X \setminus C$) in $G'$. Further, we invoke the procedure $\textsc{Partition}(G'[C], \frac{\delta|X|}{4n}, \zeta)$ on $C$ and obtain a separator $E'_{Sep}$ in $G'$ such that each SCC in $G'[C] \setminus E'_{Sep}$ has small diameter. We point out that while the first separator procedure is necessary to separate the vertices $\textsc{Center}(X)$ and $t$ in $G'$, the partitioning procedure is run for technical reasons only since we cannot ensure an efficient implementation without this step.

Finally, we wait until the data structures that maintain the generalized topological order and the distance estimates from random sources are updated before we continue with the next iteration. On termination of the while-loop, we have that all distance estimates between centers and vertices in their SCC (with regard to $G'$) are small (with regard to $G$).

\subsection{Analysis}

We establish \Cref{lma:reductionATOtoSSSP} by establishing four lemmas establishing for $(\mathcal{V}, \tau)$ correctness (\Cref{lma:correctness}), running time (\Cref{lma:runningTime}) and success probability (\Cref{lma:successRate}) and finally establishing that $c \log n$ independent copies of $(\mathcal{V}, \tau)$ form an $\mathcal{ATO}(G,2\alpha\delta)$-bundle with the guarantees given in \Cref{lma:reductionATOtoSSSP}, as required.

\begin{lemma}[Correctness]
\label{lma:correctness}
Given that no procedure returns $\textbf{Fail}$, we have that the algorithm maintains $(\mathcal{V}, \tau)$ to be an $\mathcal{ATO}(G,2\alpha \delta)$ of expected quality $\frac{(c+2)20000 n \log^5 n}{\delta}$.
\end{lemma}

Let us first prove that the diameter of SCCs in $G'$ remains small.

\begin{claim}
\label{lma:smallDiam}
After invoking \Cref{alg:diameterVio}, we have that each set $X \in \mathcal{V}$ satisfies
\[
\mathbf{diam}(X, G) \leq \frac{2\alpha\delta|X|}{n}
\] 
\end{claim}
\begin{proof}
First, recall that when the while-loop in \Cref{lne:loop} terminates, we have that every $X \in \mathcal{V}$ has that no distance estimate $\widetilde{\mathbf{dist}}(\textsc{Center}(X), t)$ or $\widetilde{\mathbf{dist}}(t, \textsc{Center}(X))$ exceeds $\frac{\delta|X|}{n}$ for any $t \in X$. 

Next, observe that the algorithm maintains the following invariant on the while-loop in \Cref{lne:loop}: every $X \in \mathcal{V}$ contains exactly one center is only marked in the data structure $\mathcal{E}_{\textsc{Center}(X)}$. This follows by resampling centers in SCCs $X$ that do not have a center yet and the by \Cref{lne:WaitUntil} which ensures that at the end of each while-loop iteration, there is time to resample.

Combined, this implies that on termination of the while-loop, for every $x, y \in X$, in any $X \in \mathcal{V}$, we have
\begin{align}
\begin{split}
\label{eq:upperBound}
    \mathbf{dist}_G(x,y) &\leq \mathbf{dist}_G(x,\textsc{Center}(X)) + \mathbf{dist}_G(\textsc{Center}(X),y) \\
    &\leq \widetilde{\mathbf{dist}}(x,\textsc{Center}(X)) + \widetilde{\mathbf{dist}}(\textsc{Center}(X),y) \\
     & \leq \frac{2\alpha\delta|X|}{n}
\end{split}
\end{align}
where we used the triangle inequality, \Cref{def:SSSPGeneric} and the fact that $\mathcal{A}_{\textsc{Center}}$ maintains distances with regard to a vertex-induced subgraph of $G$ (adding edges can only decrease distances, thus distances in $G$ are smaller than in $G[Y] \subseteq G$ for any $Y$). 
\end{proof}

Let us now bound the quality of the approximate topological order $(\mathcal{V}, \tau)$, i.e. upper bound for any $s$-to-$t$ path $\pi_{s,t}$ the amount $\mathcal{T}(\pi_{s,t}, \tau)$. As in the overview section, we focus on the "negative" terms in $\mathcal{T}(\pi_{s,t}, \tau)$, which are captured by 
\begin{equation}
        \mathcal{T}'(\pi_{s,t}, \tau) \stackrel{\text{def}}{=} \sum_{(u,v) \in \pi_{s,t}} \min\{0, \tau(X^v) - \tau(X^u)\}
\end{equation}
which is the definition of $\mathcal{T}'$ already given in equation \ref{eq:redefT}. It is not hard to see that $\mathcal{T}(\pi_{s,t} , \tau) = 2\mathcal{T}'(\pi_{s,t} , \tau) + |\tau(X^s) - \tau(X^t)| \leq 2\mathcal{T}'(\pi_{s,t} , \tau) + n$. It, thus, only remains to establish the following lemma.

\begin{claim}
\label{lma:ProbOfEi}
At any stage of $G$, for any path $\pi_{s,t}$ in $G$, we have 
\[
    \mathbb{E}[\mathcal{T}'(\pi_{s,t}, \tau)] \leq \frac{(c+2)10000 n \log^5 n}{\delta} w_G(\pi) 
\]
throughout the course of the algorithm.
\end{claim}

Before, we provide a proof, let us state the following lemma which has been shown in the last chapter. 

\participation*

\begin{proof}[Proof of \Cref{lma:ProbOfEi}]
We proof this lemma for edges $(u,v) \in E$. Then, the result follows straight-forwardly by summing over the path edges. Let us start by observing that we have $\mathcal{T}'((u,v), \tau) \neq 0$ if and only if $X^v$ strictly precedes $X^u$ in $\tau$ (where $X^z$ denotes the set in $\mathcal{V}$ that contains vertex $z \in V$). But since $(\mathcal{V}, \tau)$ forms a generalized topological order of $G'$, we have that $(u,v)$ cannot be contained in $E(G')$. 

However, we only remove edges from $E(G')$ in  \Cref{lne:updateGraph} of our algorithm, after being added to $E_{Sep}$ in  \Cref{lne:DelSep} or \ref{lne:DelSepRev}, or to $E'_{Sep}$ in  \Cref{lne:partitionXC}. Having $(u,v) \in E_{Sep}$ occurs by \Cref{lma:sepIntro} only if $(u,v)$ is contained in $G'[X]$ and if at least one of the endpoints is in $C$ (depending on whether the separator is computed on $G'[X]$ or $\overleftarrow{G'[X]}$ it is $u$ or $v$). In this case, the probability that $(u,v)$ is added to $E_{Sep}$ is at most $\frac{(c+2)\log n 2n}{|X|\delta} w_G(u,v)$, again by \Cref{lma:sepIntro}. 

However, if $(u,v)$ is not added to $E_{Sep}$ (and not already removed from $G'$) then it is completely contained in $G'[C]$. Thus, by \Cref{lma:partitionFull} it is sampled into $E'_{Sep}$ with probability at most
$\frac{(c+2)240 \log^4(n) \cdot 4n}{|X|\delta}w_G(u,v) \leq \frac{(c+2)960 n \log^4 n}{|X|\delta}w_G(u,v)$.

Observe that if $(u,v)$ is sampled into either $E_{Sep}$ or $E'_{Sep}$, then since it was contained in $X$ and by the nesting property of $\tau$ which is guaranteed by \Cref{thm:SCCinDecrGraph}, we have that during the rest of the algorithm, we have $|\tau(X^u) - (X^v)| < |X|$ where $X^u$ (resp. $X^v$) denotes the set in $\mathcal{V}$ that contains $u$ (resp. $v$).

Thus, a while-loop iteration where $u$ or $v$ participate in $C$ adds to $\mathbb{E}[\mathcal{T}'(\pi_{s,t}, \tau)]$ at most
\[
|X| \cdot \frac{(c+2)1000 n \log^4 n}{|X|\delta}w_G(u,v) = \frac{(c+2)1000 n \log^4 n}{\delta}w_G(u,v).
\]
Since by \Cref{lma:EStreeprob} each vertex only occurs during $2 \lg n$ while-loop iterations in $C$, we can establish the final bound.
\end{proof}

Combining the fact that $(\mathcal{V}, \tau)$ is a $\textsc{GeneralizedTopologicalOrder}(G')$ at all stages and $G' \subseteq G$ where $\tau$ has the nesting property, combined with \Cref{lma:smallDiam} and \cref{lma:ProbOfEi}, we derive \Cref{lma:correctness}.

\begin{lemma}[Running Time]
\label{lma:runningTime}
The algorithm to maintain $(\mathcal{V}, \tau)$ requires at most expected time
\[
O\left(\sum_{j=0}^{\lceil\lg \delta \rceil} \; \sum_{k=0}^{2^{j+3} \lceil \lg \delta\rceil} T_{SSSP}(m_{j,k}, n/2^{j}, \delta, 2) + m \log^2 n\right)
\]
where $\sum_j m_{j,k} \leq 16m \lceil \lg \delta\rceil$.
\end{lemma}
\begin{proof}
Again, our proof crucially relies on the following lemma.

\participation*

We first observe that the initialization procedure described in \Cref{subsec:InitATO} initializes $G'$ in $O(m)$ time and then runs $O(\log n)$ iterations where in each iteration it invokes the procedure $\textsc{Partition}(\cdot)$ on a set of disjoint subgraphs of $G$ to update $G'$. By \Cref{lma:partitionFull}, we can implement all of these calls in time 
\[
O\left(\log n \left(\sum_{j=0}^{\lceil\lg \delta \rceil} \; \sum_{k=0}^{2^{j+1}} T_{SSSP}(m_{j,k}, n/2^{j}, \delta, 2) + m \log n\right)\right).
\]
The latter term in this expression subsumes the time spend on updating $G'$ once a separator is returned.

Next, let us bound the time spend on maintaining the $\mathcal{SSSP}$ data structures as described in \Cref{subsec:maintainSSSPs}. It is here that we use \Cref{lma:EStreeprob}: we have that initially each vertex (and edges) is in exactly one data structure. Further, every second time a vertex $v$ participates in $C$ as computed in \Cref{lne:DelSep} or \Cref{lne:DelSepRev}, the SCC it is contained in in $G'$ is halved in size (i.e. in the number of vertices). Since new $\mathcal{SSSP}$ data structures are initialized on the new SCCs that are contained in the $C$ set, we have that each vertex $v$, in expectation, only participates $2$ times in an $\mathcal{SSSP}$ structure with running time $T_{SSSP}(m_{j,k}, n/2^{j}, \delta, 2) + m \log n$ for any $j$. Since each edge is incident to only two vertices, we have a similar argument on edges and can therefore bound the total amount of time spend on $\mathcal{SSSP}$ data structures by
\[
O\left(\sum_{j=0}^{\lceil\lg \delta \rceil} \; \sum_{k=0}^{2^{j+1}} T_{SSSP}(m_{j,k}, n/2^{j}, \delta, 2)\right).
\]
where $\sum_j m_{j,k} \leq 4m$.

Finally, let us bound the time spend in calls to \Cref{alg:diameterVio}. We observe that each while-loop iteration takes time 
\[
O\left(\sum_{j=0}^{\lceil\lg \delta \rceil} \; \sum_{k=0}^{2^{j+1}} T_{SSSP}(m'_{j,k}, n'/2^{j}, \delta, 2) + m' \log n\right)
\]
where $\sum_j m'_{j,k} \leq 4m'$ for $m' = |E_G(C)|$ and $n' = |V_G(C)|$. This follows since the $\textsc{OutSeparator}(\cdot)$ procedure runs in time almost-linear in the number of edges incident to $C$ and afterwards the call of the procedure $\textsc{Partition}(\cdot)$ which dominates the costs of the procedure is only on the graph $G'$ induced by the vertices in $C$. Thus, this insight follows straight-forwardly from \Cref{lma:sepIntro} and \Cref{lma:partitionFull} and the insight that the cost of the remaining operations is subsumed in the bounds.

Finally, we again use \Cref{lma:EStreeprob} which gives that summing over all while-loop iterations is at cost at most 
\[
O\left(\sum_{j=0}^{\lceil\lg \delta \rceil} \; \sum_{k=0}^{2^{j+2} \lceil \lg \delta\rceil} T_{SSSP}(m_{j,k}, n/2^{j}, \delta, 2) + m \log^2 n\right)
\]
where $\sum_j m_{j,k} \leq 8m \lceil \lg \delta\rceil$. Combining the parts of the algorithm, we thus get the total bound.
\end{proof}

\begin{lemma}[Success Probability]
\label{lma:successRate}
The algorithm reports $\textbf{Fail}$ with probability at most $2n^{-c-1}$. 
\end{lemma}
\begin{proof}
We point out that we can only get a $\textbf{Fail}$ due to procedures $\textsc{OutSeparator}(\cdot)$ and $\textsc{Partition}(\cdot)$. 

Since each separator found in the while-loop in \Cref{lne:loop} refines $\mathcal{V}$, we can bound the number of while-loop iterations in the course of the algorithm by $n-1$. Thus, we make at most $n-1$ calls to procedures $\textsc{OutSeparator}(\cdot)$ and $\textsc{Partition}(\cdot)$. Each of the former calls returns $\textbf{Fail}$ with probability at most $n^{-(c+2)}$ and each of the latter with probability at most $n^{-(c+2)}$. 

Taking a union bound over all events, the lemma follows.
\end{proof}

Finally, let us put everything together and prove our main theorem.

\atoBundle*

\begin{proof}
We maintain a collection of $40 c \log n$ independent $\mathcal{ATO}(G, 2\alpha\delta)$ instances 
\[
(\mathcal{V}_1, \tau_1), (\mathcal{V}_2, \tau_2), \dots ,(\mathcal{V}_{40c \log n}, \tau_{40c \log n})
\]
as described earlier in this section and let $\mathcal{S}$ denote the collection of these instances.

The total running time to maintain these $\mathcal{ATO}(G, 2\alpha\delta)$'s is clearly bounded by the term given in equation \ref{eq:totalRunningTime} by \Cref{lma:runningTime}.

Now, since by \Cref{lma:correctness}, each $\mathcal{ATO}(G, 2\alpha\delta)$ has expected quality $q = \frac{(c+2)20000 n \log^5 n}{\delta}$, we have by Markov's inequality and a simple Chernoff bound, that for each shortest path $\pi_{s,t}$ in $G$ at some stage $t$, we have that there exists an $i$, such that $\mathcal{T}(\pi_{s,t}, \tau_i) \leq 2q$ with probability at least $1 - e^{- 40 c \log n/8} = 1 - n^{-5 c}$. Since $c > 1$, we have that the probability that any shortest-path at any stage fails, is at most $1-n^{-c}/2$ by union bounding over at most $n^2$ stages and at most $n^2$ shortest-paths, for $n$ large enough. Moreover, the total probability that any instance returns $\textbf{Fail}$ is at most $n^{-c}/2$ by \Cref{lma:successRate} and a union bound over the instances. Thus, we have established that with probability at least $1-n^{-c}$, $\mathcal{S}$ forms an $\mathcal{ATO}(G,2\alpha \delta)$-bundle of quality $2q$ as defined in \Cref{def:ATObundle}.
\end{proof}

\section{The SSSP Algorithm}
\label{sec:SSSPDense}

We now give a proof of  \Cref{thm:TopologicalOrderMaintenanceOverview} which implies our main result,  \Cref{thm:ContributionSSSPResult}, as a corollary. Our proof is in two steps: we first show how to implement an $\alpha$-approximate $\delta$-restricted $\mathcal{SSSP}$ as described in \Cref{def:SSSPGeneric} given access to approximate topological orders. We then show how to bootstrap the reductions to maintains different $\mathcal{SSSP}$ data structures to cover all depths.

\subsection{\texorpdfstring{$\alpha$-approximate $\delta$-restricted $\mathcal{SSSP}$}{SSSP} via an ATO}
\label{subsec:alphaDeltaSSSPReduction}

The main objective of this section is to prove the following theorem which gives a reduction from $(1+\epsilon)$-approximate $\delta$-restricted $\mathcal{SSSP}$ to approximate topological orders. In this theorem, we only assume access to an $\mathcal{ATO}(G, \eta_{diam})$ denoted by $(\mathcal{V}, \tau)$ where we assess the quality individually for each path. If the quality for a certain tuple is below a threshold $q$, we show how to exploit the approximate topological order to maintain the distance estimate for the tuple efficiently, otherwise we provide no guarantees.

\begin{restatable}{theorem}{ssspSimple}
\label{thm:SSSPEfficient}
Given $G=(V,E,w)$, a decremental weighted digraph, a source $r \in V$, a depth threshold $\delta > 0$, a quality parameter $q$, an approximation parameter $\epsilon > 0$, and access to $(\mathcal{V}, \tau)$ an $\mathcal{ATO}(G, \eta_{diam})$.

Then, there exists a deterministic data structure that maintains a distance estimate $\widetilde{\mathbf{dist}}(r,v)$ for every vertex $v \in V$ such that at each stage of $G$, $\mathbf{dist}_G(r,v) \leq \widetilde{\mathbf{dist}}(r,v)$ and if $\mathbf{dist}_G(r,v) \leq \delta$ and $\mathcal{T}(\pi_{r,v},\tau) \leq q \cdot \delta + n$, then 
\[
\widetilde{\mathbf{dist}}(r,v) \leq \mathbf{dist}_G(r,v) + \eta_{diam} + \epsilon\delta.
\]
The total time required by this structure is 
\[
O(n\delta q \log n/\epsilon + n^2 \log n)
\]
\end{restatable}
\begin{restatable}{remark}{ssspRemark}
\label{rmk:ssspRemark}
Technically, we require the approximate topological order $(\mathcal{V}, \tau)$ to encode changes efficiently and pass them the SSSP data structure. Since the SSSP data structure is updated only through delete operations, we require, that with each edge update, the data structure receives changes to $(\mathcal{V}, \tau)$ since the last stage. More precisely, we require that the user passes a set of pointers to each set $Y$ that occurred in $\mathcal{V}$ at the previous stage (denoted $\mathcal{V}^{OLD}$), but did not occur in $\mathcal{V}$ at the current stage (denoted $\mathcal{V}^{NEW}$), i.e. each $Y \in \mathcal{V}^{OLD} \setminus \mathcal{V}^{NEW}$. Additionally, we require with each such $Y$ that was split into subsets $Y_1, Y_2, \dots, Y_k \in \mathcal{V}^{NEW}$ that form a partition of $Y$, pointers to each new element $Y_i$. We further require worst-case constant query time of $\tau$, and each element $Y \in \mathcal{V}$ (for any version) can be queried for its size in constant time and returns its vertex set in time $O(|Y|)$. For the rest of the paper, this detail will be concealed in order to improve readability.
\end{restatable}

Before we show how to implement such a data structure, let us emphasize that the above theorem directly implies \Cref{thm:SSSPEfficientIntro} that we introduced in the overview. It can further also be used to derive the following corollary which is at the heart of our proof in the next section. Its proof is rather straight-forward and we therefore refer the reader to \cite{nearOptDenseSSSP}.

\begin{restatable}{corollary}{corSSSPrestrictedFromATO}
\label{cor:SSSPtoATO}
Given $G=(V,E,w)$, a decremental weighted digraph, a source $r \in V$, a depth threshold $\delta > 0$, an approximation parameter $\epsilon > 0$, and access to a collection $\mathcal{S} = \{ S_i \}_{1 \leq i \leq \mu}$ for $\mu = \lfloor\lg \delta\rfloor-1$ where each $\mathcal{S}_i$ forms an $\mathcal{ATO}(G, 2^i, 40 c\log n)$-bundle of quality $q_i$. Then, there exists an implementation for $(1+\epsilon)$-approximate $\delta$-restricted $\mathcal{SSSP}$ where $T_{SSSP}(n,m,\delta, \epsilon) = O(n
(\max_{1 \leq i \leq \mu} \{\frac{\delta q_i}{2^i}\} + n) \log^3 n / \epsilon
^2)$.
\end{restatable}

Let us now describe the implementation of a data structure $\mathcal{E}_r$ that stipulates the guarantees given in \Cref{thm:SSSPEfficient}. 
Since the proof that this is indeed a valid implementation of \Cref{thm:SSSPEfficient} is quite similar to the proof sketch we give in \Cref{sec:overview}, we refer the reader to the full version of the article \cite{nearOptDenseSSSP}.

\paragraph{Initialization.} Throughout the algorithm, we define $\delta_{max} = \lceil (1+\epsilon)\delta + \epsilon n/q \rceil$ and define the complete graph $\mathcal{H} = (\mathcal{V}, \mathcal{V}^2, w)$\footnote{Here, we are again slightly abusing notation by referring to $\mathcal{V}$ as partition and node set, however, context and the fact that this implicitly refers to a one-to-one correspondence between partition sets and nodes ensures that no ambiguity arises.} with weight function 
\[
w(X,Y) = \inf\{ w(x,y) | (x,y) \in E(X,Y)\}
\]
for $X, Y \in \mathcal{V}$. We use the convention that the infimum of the empty set is $\infty$. We use $\mathcal{H}$ to avoid dealing explicitly with $G/ \mathcal{V}$ which is a multi-graph, instead in $\mathcal{H}$ we use the same node set with simple edges as the infimum over weights of the multi-edges (even if there is no such edge). 

We use a standard min-heap data structure\footnote{See for example \cite{cormen2009introduction}.} $Q_{X,Y}$ over the set $E(X,Y)$ for each ordered pair $(X,Y)$ to maintain the weight $w(X,Y)$. We henceforth denote by $Q_{X,Y}.\textsc{MinValue}$ the value $w(X,Y)$ and by $Q_{X,Y}.\textsc{MinElem}$ a corresponding edge $(x,y)$ with $x \in X, y \in Y$, and use the convention of denoting the node in $\mathcal{V}$ that contains vertex $x \in V$ by $X^x$. We initialize the data structure $\mathcal{E}_r$ by constructing $\mathcal{H}$ and by running Dijkstra's algorithm\footnote{See \cite{cormen2009introduction} for an efficient implementation.} from $X^r$ on $\mathcal{H}$. We then initialize a distance estimate $\widetilde{\mathbf{dist}}(X^r, Y)$ for each $Y \in \mathcal{V}$ to $\mathbf{dist}_{\mathcal{H}}(X^r,Y)$. If we have $\widetilde{\mathbf{dist}}(X^r,Y) > \delta_{max}$ at any point in the algorithm, we set it to $\infty$. Further, we also maintain the distance estimates $\widetilde{\mathbf{dist}}(r, u)$ for each $u \in V$ equal to $\widetilde{\mathbf{dist}}(X^r, X^u) + \eta_{diam}$, i.e. every time we increase $\widetilde{\mathbf{dist}}(X^r, X^u)$, we also increase $u$'s distance estimate\footnote{We will show that $\widetilde{\mathbf{dist}}(X^r, X^u)$ is a monotonically increasing value over time.}. This allows us to henceforth focus on the distance estimates of nodes which is easier to describe. We also store the corresponding shortest-path tree $T$ truncated at distance $\delta_{max}$ that serves as a certificate of the distance estimates.

Finally, we partition for each node $X \in \mathcal{V}$, the in-neighbors set in $\mathcal{H}$ of $X$ into different buckets based on their $\tau$-distance: for each $X$, we initialize bucket $B_{-1}(X) = \{X\}$ and for $0 \leq j \leq \lg n$ we initialize the bucket $B_j(X)$ to 
\[ 
\{ 2^{j} \leq \chi(X,Y, \tau) < 2^{j+1} | Y \in N_{\mathcal{H}}^{in}(X), X \neq Y\}
\]
where we define 
\begin{equation} \label{eq:chi}
 \chi(X ,Y, \tau) \stackrel{\text{def}}{=}          \begin{cases}
                    \tau(Y) - (\tau(X) + |X| - 1) & \text{if } \tau(X) < \tau(Y) \\
                    \chi(Y, X, \tau) & \text{otherwise}
                 \end{cases}
\end{equation}
that is $\chi(\cdot)$ is similar to $\mathcal{T}(\cdot)$ (in fact $ \chi(X ,Y, \tau) \leq \mathcal{T}(X, Y, \tau)$), however, as $\tau$ maps nodes $X$ and $Y$ to disjoint intervals, $\mathcal{T}(\cdot)$ measures the distance between the starting points of the intervals, while $\chi(\cdot)$ measures the distance between the intervals (i.e. the closest endpoints of the intervals).

At any stage, we let $B_{\leq j}(X) = \bigcup_{j' \leq j} B_{j'}(X)$. We store each set $B_j(X)$ explicitly as a linked list, store for each $Y \in B_j(X)$ a pointer to the bucket, and maintain the buckets to partition the in-neighbors of each $X$. 

\paragraph{Handling Edge Deletions.} The edge deletion procedure takes two parameters: the edge to be deleted $(u,v)$ and a collection of tuples $U$ that encode refinements of $\mathcal{V}$ during this stage. To handle the update, we initialize a min-heap $Q = \emptyset$ that keeps track of the nodes in $\mathcal{H}$, that cannot be reached from $X^r$ in the truncated shortest-path tree $T$ (i.e. whose certificate for the current distance estimate was compromised).

We start our update procedure by processing updates to $(\mathcal{V}, \tau)$ (check the remark of \Cref{thm:SSSPEfficient} for a description of these updates encoded by $U$). For any node $X \in \mathcal{V}^{OLD} \setminus \mathcal{V}^{NEW}$, that was split into subsets $X_1, X_2, \dots, X_k \in V^{NEW}$ (i.e. for every tuple $(X, X_1, X_2, \dots, X_k) \in U$), we query for each $X_i$, its size. Then, we let the largest node $X_i$ \emph{inherit} the original node $X$ (that is the nodes are equal in our data structure at this stage although the partition sets are not), and create a new node $X_{i'}$ in $\mathcal{H}$ for other $i' \neq i$, and new heap structures $Q_{X_{i'}, Y}$ for every $Y \in \mathcal{V}$. Then for each $X_{i'}$, $i' \neq i$, we scan each edge $(x,y)$ in $E(X_{i'}, V \setminus X_{i'})$, remove it from the heap $Q_{X_i, X^y}$ and add it to the new heap $Q_{X_{i'}, X^y}$. We also initialize the distance estimates for each $X_{i'}$, $\widetilde{\mathbf{dist}}(X^r, X_{i'})$ to take the value $\widetilde{\mathbf{dist}}(X^r, X)$, also for $X_i$. We then find the edge $(w,x)$ in $T$ where $X = X^x$. Clearly, we now have that $X_{i'} = X^x$ for some ${i'}$ and we connect $X_{i'}$ in the tree $T$ since this edge is now a certificate for $X_{i'}$'s distance estimate. The rest of the nodes, i.e. the nodes $X_1, X_2, \dots, X_{{i'}-1}, X_{{i'}+1}, \dots, X_k$, we add to $Q$ since we do not have a certificate for them yet.

Finally, when \emph{all} the node splits for the current stage where processed, we update the buckets $B_j(Y)$ for all $j$ and $Y$ to \emph{almost} stipulate the initialization rules. We point out that all edges that have to be assigned to a different bucket have to be incident to $X_1, X_2, \dots, X_k$ by property \ref{prop:tauUpdate}. Then, for each $X_{i'}$ (also $X_i$), we compute $j$ to be the largest integer such that some number in $|X_{i'}|, |X_{i'}|+1, \dots, |X| - 1$ is divisible by $2^j$. Then, we update all nodes in $B_{\leq j+1}(X_{i'})$ by scanning and reassigning them, and similarly reassign $X_{i'}$ to a new bucket for each $Y$ where $X_{i'} \in B_{\leq j}(Y)$. 

Finally, when all node splits are processed, the node set of $\mathcal{H}$ reflects the current $\mathcal{V}$, and we can delete the edge $(u,v)$ from $\mathcal{H}$ by deleting it from the heap it is contained in. If $(u,v)$ was equal to $Q_{X^u, X^v}.\textsc{MinElem}$, and $(X^u, X^v) \in T$, we delete it from $T$ and insert $X^v$ into $Q$.

We then rebuild our certificate $T$: we take the node $Y$ from $Q$ with the smallest distance estimate $\widetilde{\mathbf{dist}}(X^r, Y)$ until $Q$ is empty or the smallest distance estimate $\infty$. Now, let $j$ be the largest integer such that $ \widetilde{\mathbf{dist}}(X^r, Y)$ is divisible by $\lceil 2^{j} \cdot \frac{\epsilon}{q} \rceil$. We then check if there exists a node $X \in B_{\leq j}(Y)$ such that
\[
    \widetilde{\mathbf{dist}}(X^r, X) + Q_{X,Y}.\textsc{MinValue} \leq \widetilde{\mathbf{dist}}(X^r, Y).
\]
In this case, the edge $(x,y) = Q_{X,Y}.\textsc{MinElem}$ serves as a certificate that the distance from $X^r$ to $Y$ is at most $\widetilde{\mathbf{dist}}(X^r, Y)$ and therefore we add $(x,y)$ to $T$. If there exists no such vertex $X$, then we increase the value $\widetilde{\mathbf{dist}}(X^r, Y)$ by one or to $\infty$ if it is currently at least $\delta_{max}$ and reinsert $Y$ and children $Z_1, Z_2, \dots, Z_k$ of $Y$ in $T$ into $Q$ (after deleting the edges $(Y, Z_i)$ from $T$). This completes the description of the algorithm. Again, we refer the reader interested in the proof of \Cref{thm:SSSPEfficient} to \cite{nearOptDenseSSSP}.

\subsection{Bootstrapping an Algorithm for Unrestricted Depth}
\label{subsec:bootstrapDense}

Next, let us prove the following theorem which show is a detailed version of \Cref{thm:TopologicalOrderMaintenanceOverview}. Combined with \Cref{cor:SSSPtoATO} (where we set the depth threshold parameter $\delta$ to $Wn$), this immediately implies our main result, \Cref{thm:ContributionSSSPResult}.

\begin{theorem}
\label{thm:TopologicalOrderMaintenanceBootstrap}
For any $0 \leq i \leq \lg(Wn)$, given a decremental digraph $G=(V,E,w)$, we can maintain a hierarchy $\mathcal{S} =\{\mathcal{S}_i\}_i$ where each $\mathcal{S}_i$ is a $\mathcal{ATO}(G, 2^i, 40 c \log n)$-bundle of expected quality $\tilde{O}(n/2^i)$. The algorithm runs in total expected update time $O(c^5 n^2 \log^{17} n \lg^{5}(Wn))$ against a non-adaptive adversary and is correct with probability at least $1-n
^{-c+2}$ for any failure probability parameter $c \geq 2$.
\end{theorem}
\begin{proof}
In order to prove our theorem formally, we need to fix the constants hidden by the big-$O$ notation in some of our statements. We therefore henceforth denote the constant hidden by \Cref{cor:SSSPtoATO} to maintain the SSSP data structure by $c_{SSSP}$, the constant hidden in \Cref{lma:reductionATOtoSSSP} to obtain an $\mathcal{ATO}$-bundle from an SSSP data structure by $c_{SSSP \to ATO}$ and finally, we denote the constant hidden in the theorem that we want to prove by $c_{Total}$ where we require that $c_{Total} \geq (c_{SSSP \to ATO} \cdot c_{SSSP})^2 \cdot 2^{49}$.

Without further due, let us prove the theorem by induction on $n$, the number of vertices in graph $G$. The base case with $n \leq 1$ is easily established since there are no paths in a graph of only one vertex thus we obtain arbitrarily good quality and the running time is a small constant (at least smaller than $c_{Total}$).

Let us now give the inductive step $n \mapsto n+1$: for each $i$, we iteratively construct an $\mathcal{ATO}(G, 2^i, 40 c \log n)$-bundle $\mathcal{S}_i$ as described in \Cref{lma:reductionATOtoSSSP}. Thus, we have to show how to implement a $2$-approximate $2^{i-2}$-restricted $\mathcal{SSSP}$ data structure required in the reduction (note that for $i \leq 2$ this task is trivial, so we omit handling it as special levels).

Note that each data structure $\mathcal{SSSP}$ that we are asked to implement for an $\mathcal{ATO}$-bundle $\mathcal{S}_i$ at level $i$ is run on a different graph $H \subseteq G$. To obtain an efficent algorithm, we will implement the data structure differently depending on the size of such $H$. If $H$ has at least $n/2^\gamma$ vertices for some $\gamma = \Theta( \lg \log Wn)$ that we fix later, we call $H$ a \emph{large} graph. Otherwise, we say $H$ is \emph{small}. Now, we implement $\mathcal{SSSP}$ as follows:
\begin{itemize}
    \item if $H$ is \emph{small}, then we use the induction hypothesis, find a $\mathcal{ATO}$-bundle $\mathcal{S}'$ for $H$ and and invoke \Cref{cor:SSSPtoATO} on $\mathcal{S}'$. We note that we need to set the parameter that controls the failure probability for $\mathcal{S}'$ to $c \cdot 4 \log n$ to ensure that it succeeds with high probability (this is since $\mathcal{S}'$ only succeeds with probability polynomial in $|V(H)|$ which might be very small).
    \item if $H$ is \emph{large}, we exploit \Cref{rmk:LaminarFamilyOfGraphs} which states that when the reduction asks to maintain approximate distances on some graph $H \subseteq G$, it is sufficient to maintain distance estimates on any graph $F$ such that $H \subseteq F \subseteq G$ and in particular, it is ok to simply run on the entire graph $G$. Therefore, we simply use $\mathcal{S}_{i-1}$ in combination with \Cref{cor:SSSPtoATO} and maintain distances in $G$. 
\end{itemize}

Let us now analyze the total running time. We start by calculating the running time required by each level $i$ separately. For some fixed $i$, we have that by \Cref{lma:reductionATOtoSSSP} we have running time 
\begin{equation} \label{eq:totalCostBoundATO}
c_{SSSP \to ATO} \left(\sum_{j=0}^{\lceil\lg 2^{i-2} \rceil} \; \sum_{k=0}^{2^{j+3} c \log^2 n} T_{SSSP}(m_{j,k}, n/2^{j}, 2^{i-2}, 2) + n^2 \log^3 n\right)
\end{equation}
where $\sum_j m_{j,k} \leq 16c \cdot m \log^2 n$ for all $k$, to maintain $\mathcal{S}_i$.

Let us analyze the terms $T_{SSSP}(m_{j,k}, n/2^{j}, 2^{i-2}, 2)$ based on whether $j < \gamma$ or not (i.e. depending on how the SSSP data structure was implemented):
\begin{itemize}
    \item if $j < \gamma$: then, by the induction hypothesis, we require time at most 
        \begin{align*}
    \begin{split}
    & c_{Total} (c 4 \log n)^5 \cdot \left(\frac{n}{2^j}\right)^2 \log^{17}(n) \lg^{5}(Wn) 2^{-2j} \\
    & < c_{Total} \cdot 2^8 \cdot c^5 n^2 \log^{22}(n) \lg^{5}(Wn)  2^{-j}2^{-\gamma}
    \end{split}
    \end{align*}
    to maintain the new $\mathcal{ATO}$-bundle $\mathcal{S}'$ on the graph $H$ and again by the induction hypothesis we have that the bundle has quality $\frac{(4\log n c +2)40000n \log^5 n 2^{-j}}{2^{i-3}} \log^3 n$. 
    
    Thus, maintaining SSSP on $H$ using $\mathcal{S}'$ as described in \Cref{cor:SSSPtoATO} can be done in time 
    \[
     c_{SSSP}((c \cdot 4\log n+2)40000 \cdot 2^3)(n^2 \log^8 n 2^{-2j}) < c_{SSSP} \cdot c \cdot 2^{22} (n^2 \log^9 n 2^{-j} 2^{-\gamma}).
    \]
    Combined, we obtain that we can implement the entire SSSP data structure with total running time at most
    \begin{align*}
    \begin{split}
    &c_{Total} \cdot 2^8 \cdot c^5 n^2 \log^{22}(n) \lg^{5}(Wn) 2^{-j}2^{-\gamma} + c_{SSSP} \cdot c \cdot 2^{22} (n^2 \log^9 n 2^{-j}2^{-\gamma})\\
    &\leq c_{Total} \cdot c_{SSSP} \cdot c^5 \cdot 2^{-j}2^{-\gamma} \cdot 2^{22} \cdot n^2 \log^{22}(n) \lg^{5}(Wn).
    \end{split}
    \end{align*}

    Combining these bounds and summing over all \emph{small} graph terms in equation \ref{eq:totalCostBoundATO}, we obtain that the total contribution is at most
    \begin{align*}
    \begin{split}
    & c_{SSSP \to ATO} \cdot \lg(nW) \cdot (2^3 c\log^2 n) \left(c_{Total} \cdot c_{SSSP} \cdot c^5 \cdot 2^{-\gamma} \cdot 2^{22} \cdot n^2 \log^{22}(n) \lg^{5}(Wn)\right) \\
    &\leq c^6 \cdot c_{SSSP \to ATO} \cdot c_{Total} \cdot c_{SSSP} \left(2^{22} \cdot n^2 \log^{24}(n) \lg^{6}(Wn) 2^{-\gamma}\right).
    \end{split}
    \end{align*}
    This completes the analysis of the small graph data structures.
    \item otherwise ($j \geq \gamma$): then, we run the SSSP structure from \Cref{cor:SSSPtoATO} on $\mathcal{S}_{i-1}$ which gives running time at most
    \[
    c_{SSSP}\left(n 2^i \cdot \frac{(c+2)40000n \log^5 n}{2^{i-3}} \log^3 n\right) \leq c_{SSSP} \cdot c ( 2^{20} \cdot n^2 \log^8 n)
    \] 
    where we used $c \geq 2$. Since there are at most $c_{SSSP \to ATO} \cdot c \cdot (2^3 \lg(nW)\log^2 n 2^{\gamma})$ terms for large graphs, where $j \geq \gamma$, we have that the total cost of all $\mathcal{SSSP}$ data structures on \emph{large} graphs is at most
    \begin{align*}
    \begin{split}
   &c_{SSSP \to ATO} \cdot c \cdot (2^3 \lg(nW)\log^2 (n) 2^{\gamma}) \cdot \left(c_{SSSP} \cdot c ( 2^{20} \cdot n^2 \log^8 n) \right)\\
   &= c_{SSSP \to ATO} \cdot c^2 \cdot c_{SSSP} \cdot \left(2^{23} \cdot n^2 \log^{10} (n)  \lg(nW) 2^{\gamma}\right)
    \end{split}
    \end{align*}
\end{itemize}

It now only remains to choose $\gamma$ and combine the two bounds. We set $\gamma = 24 \lceil\lg (c^2 \cdot c_{SSSP \to ATO} \cdot c_{SSSP} \cdot \lg^{3}(Wn) \log^{7}(n)) \rceil$, and obtain that the total running time summed over large and small graphs is at most
\begin{align*}
\begin{split}
& c_{SSSP \to ATO}^2 \cdot c_{SSSP}^2 \cdot c^4 \cdot 2^{48}(n^2 \log^{17} n \lg^{4}(Wn))) \\ &+ \frac{c^4 c_{Total} (n^2 \log^{17} n \lg^{3}(Wn))}{2}
\\ &\leq  c^4 c_{Total} (n^2 \log^{17} n \lg^{4}(Wn))
\end{split}
\end{align*}
where we used our initial assumption on the size of $c_{Total}$.

Finally, we point out that there are at most $\lg(nW)$ levels $i$ and therefore, the total update time is at most
\[
 c_{Total} (c^4 n^2 \log^{17} n \lg^{5}(Wn))  
\]
as required.

Further, we point out that every $\mathcal{ATO}$-bundle $\mathcal{S}_i$ that was constructed runs correctly with high probability at least $1-n^{-c}$, while every $\mathcal{ATO}$-bundle $\mathcal{S}'$ is maintained correctly with probability at least $1-n^{-4c}$ (recall that we set the failure parameter of these data structures to $c \cdot 4 \log n$). Noting that we only have $\lg(nW)$ instances of the former bundles, and at most $n^3$ of the latter, taking a simple union bound over the events that any bundle instance fails gives a total failure probability of at most $n^{-c+2}$. 
\end{proof}

%% file: det_scc.tex
\label{chap_deterministic_scc}

In this chapter, we are concerned with obtaining the first deterministic data structures that improve upon the classic ES-trees in decremental graphs.

\ContributionDetSCCResult*

\ContributionDetSSSPResult*

While we will touch on the expander techniques developed in the article that this section is based on, we emphasize the above theorems as the main results in the context of this thesis and refer readers that are mostly interested in the expander tools to \cite{DetDecrementalSSSP}.

\section{Overview}
 
We now give an overview of the data structures. We first focus on the problem of obtaining a data structure for the decremental SCC problem and only in the latter part of this overview show how to obtain a decremental SSSP data structure.

\paragraph{Directed Expanders.} We start with the definition of directed \emph{expanders}.

\begin{defn}
[Expanders]A directed graph $G$ is a \emph{$\phi$-vertex expander}
if it has no $\phi$-vertex-sparse vertex-cut. Similarly, $G$ is
\emph{$\phi$-(edge) expander} if it has no $\phi$-sparse cut.\footnote{Note that an isolated vertex is an expander (in both edge and vertex
versions).} 
\end{defn}

Intuitively, expanders are graphs that are ``robustly connected''
and, in particular, they are strongly connected. It is well-known
that many problems become much easier on expanders. So, given a problem
on general graph, we would like to reduce the problem to expanders. 

It is further well-known that every \emph{undirected} graph admits the following
\emph{expander decomposition}: for any $\phi>0$, a $\tilde{O}(\phi$)-fraction
of vertices/edges can be removed so that the remaining is a set of
vertex-disjoint $\phi$-vertex/edge expander. Unfortunately, this
is impossible in directed graphs. Consider, for example, a DAG. However,
a DAG is the only obstacle; for any $\phi>0$, we can remove a $\Otil(\phi)$-fraction
of vertices/edges, so that the remaining part can be partitioned into a DAG and a set of vertex-disjoint $\phi$-vertex/edge
expanders. This observation can be made precise as follows.
\begin{fact}
[Directed $\phi$-Expander Decomposition]\label{prop:exp decomp}Let $G=(V,E)$
be any directed $n$-vertex graph and $\phi>0$ be a parameter. There
is a partition $\{R,X_{1},\dots,X_{k}\}$ of $V$ such that 
\begin{enumerate}
\item $|R|\le O(\phi n\log n)$;
\item $G[X_{i}]$ is a $\phi$-vertex expander for each $i$;
\item Let $D$ be obtained from $G$ by deleting $R$ and contracting each
$X_{i}$. Then, $D$ is a DAG.
\end{enumerate}
\end{fact}

The edge version of \Cref{prop:exp decomp} can be stated as follows:
for any unweighted $m$-edge graph $G=(V,E)$, there is a partition
$\{X_{1},\dots,X_{k}\}$ of $V$ and $R\subset E$ where $|R|\le O(\phi m\log m)$,
each $G[X_{i}]$ is a $\phi$-expander, and $D$ is a DAG (where $D$
is defined as above). It can be generalized to weighted graphs as
well. 

\paragraph{High-level Framework.} This decomposition motivates the framework of our algorithm, although for the sake of efficiency we only maintain an approximate version. The decomposition suggests that we need four main ingredients: 
\begin{enumerate}
\item a dynamic expander decomposition in directed graphs, 
\item a fast algorithm on vertex-expanders,
\item a fast algorithm on DAGs, and
\item a way to deal with the small remaining part $R$.
\end{enumerate}
We point out however that component 3. will only be required by the SSSP data structure.

\paragraph{The High-Level Algorithm for decremental SCC.} In order to develop some intuition for the above framework, let us derive our data structure for decremental SCC.

First, let us state the following Theorem which is an idealized version of our main technical result. We will pretend for the first part of the overview that we can indeed prove this Theorem but we stress that we obtain a considerably weaker guarantees in our algorithm. 

\begin{restatable}[Idealized Expander Decomposition]{thm}{ExpanderDecompTheorem}
	\label{thm:expanderDecom}
	There is a deterministic algorithm $\mathcal{A}$ that given an unweighted, directed, decremental graph $G$ and a parameter $\phi \in (0,1)$, maintains a directed $\phi$-expander decomposition, such that $R$ is an incremental set of final size $\Ohat(\phi^{-1})$. The algorithm to maintain the expander decomposition runs in time $\Ohat(m\phi^{-2})$.
\end{restatable}

Observe that every set $X$ in $G$, that forms a $\phi$-expander (for any $\phi$) is part of an SCC in $G$. Thus, given the above Theorem, the graph $G$ where expanders $X_1, X_2, \dots, X_k$ are contracted, forms a condensation of $G \setminus R$. 

To this end, we note that in \cite{chechik2016decremental} (which we explained in some detail in \Cref{chap:rand_scc}), a data structure by Lacki was used to deal with separator vertices set $S$ to restore SCCs in $G$ given the set $R$ and the condensation of $G \setminus S$. We state his result below (here the result is stated as a reduction which is straight-forward to obtain and is proven in \cite{detDiSSSP}).  
  
\begin{restatable}[see \cite{Lacki11, chechik2016decremental}]{theorem}{LackiProp}
\label{thm:lacki}
Let $G=(V,E)$ be a decremental graph. Let $\mathcal{A}$ be a data structure that {\bf 1)} maintains a monotonically growing set $S \subseteq V$ and after every adversarial update reports any additions made to $S$ and {\bf 2)} maintains the SCCs in $G \setminus S$ explicitly in total update time $T(m,n)$ and supports SCC path queries in $G \setminus S$ in almost-path-length query time.

Then, there exists a data structure $\mathcal{B}$ that maintains the SCCs of $G$ explicitly and supports SCC path-queries in $G$ (in almost-path-length query time). The total update time is $O(T(m,n) + m|S|\log n)$, where $|S|$ refers to the final size of the set $S$. 
\end{restatable}

Using these two Theorems in conjunction, it is not hard to see that setting $\phi=n^{-1/3}$, we can obtain a data structure that maintains SCCs in a decremental graph in total update time $\Ohat(mn^{2/3})$.

\paragraph{Certifying Directed Expanders.} In order to maintain an expander decomposition as described in \Cref{thm:expanderDecom} (which we cannot do! We obtain a considerable weaker guarantee), a basic procedure that is required to even compute an expander decomposition is an algorithm to efficiently certify that $G$ is a $\phi$-expander \emph{or} outputs a $\phi$-vertex-sparse vertex-cut. 

We introduce the notion of embeddings which has been used heavily in the static setting to certify expanders \cite{khandekar2009graph, orecchia2008partitioning, Louis10}:

\begin{defn}
[Embedding and Embedded Graph]Let $G=(V,E)$ be a directed graph.
An \emph{embedding} $\pset$ in $G$ is a collection of simple directed
paths in $G$ where each path $P\in\pset$ has associated \emph{value}
$\val(P)>0$. We say that $\pset$ has \emph{length} $\len$ if every
path $P\in\pset$ contains at most $\len$ edges. We say that $\pset$
has \emph{vertex-congestion} $\congest$ if, for every vertex $v\in V$,
$\sum_{P\in\pset_{v}}\val(P)\le\congest$ where $\pset_{v}$ is the
set of paths in $\pset$ containing $v$. We say that $\pset$ has
\emph{edge-congestion }$\congest$ if, for every edge $e\in E$, $\sum_{P\in\pset_{e}}\val(P)\le\congest$
where $\pset_{e}$ is the set of paths in $\pset$ containing $e$.

Given an embedding $\pset$, there is a corresponding weighted directed
graph $W$ where, for each path $P\in\pset$ from $u$ to $v$, there
is a directed edge $(u,v)$ with weight $\val(P)$. We call $W$ an
\emph{embedded graph }corresponding to $\pset$ and say that $\pset$
embeds $W$ into $G$. 
\end{defn}

The following fact shows that, to certify that $G$ is a vertex expander,
it is enough to \emph{embed} an (edge)-expander $W$ into $G$ with small
congestion. 
\begin{fact}
\label{fact:certify vertex expansion}Let $G=(V,E)$ be a graph. Let
$W=(V,E',w)$ be a $\phi$-expander with minimum weighted degree $1$.
If $W$ can be embedded into $G$ with vertex congestion $\congest$,
then $G$ is a $(\phi/\congest)$-vertex expander. 
\end{fact}

\begin{proof}
Consider a vertex cut $(L,S,R)$ in $G$ where $|L| \le |R|$. 
Suppose that $E(L,R) = \emptyset$, otherwise $E(R,L) = \emptyset$ and the proof is symmetric.
Observe that each edge $e\in E_{W}(L,V\setminus L)$ in $W$ corresponds to a path in $G$ that goes
out of $L$ and, hence, must contain some vertex from $S$.
So the total weight of these edges in $W$ can be at most $\delta_{W}^{out}(L)\le|S|\cdot\congest$.
At the same time, $\delta_{W}^{out}(L)\ge\phi\vol_{W}(L)\ge\phi|L|$
as $W$ is a $\phi$-expander with minimum weighted degree 1. So $|S|\ge\frac{\phi}{\congest}|L|$
as desired.
\end{proof}
In our actual algorithm, instead of certifying that $G$ is a vertex
expander (i.e.~$G$ has no sparse vertex-cut), we relax to the task
to only certifying that $G$ has no balanced sparse vertex-cut. This, in turn, 
motivates the definition of $\phi$-witness:
\begin{defn}
[Witness]\label{def:witness}We say that $W$ is a\emph{ $\phi$-witness}
of $G$ if $V(W)\subseteq V(G)$, $W$ is a $\Omegahat(1)$-(edge)-expander
where $9/10$-fraction of vertices have weighted degree at least $1/2$,
and there is an embedding of $W$ into $G$ with vertex-congestion
$1/\phi$. (Note that $E(W)$ does not have to be a subset of $E(G)$.)
We say that $W$ is a $\phi$-short-witness if it is a $\phi$-witness
and the embedding has length $\Ohat(1/\phi)$. We say that $W$ is
a \emph{large} witness if $|V(W)|\ge9|V(G)|/10$.\footnote{The constant $9/10$ is somewhat arbitrary.}
\end{defn}

We sometimes informally refer to a graph that contains a large witness
as an \emph{almost vertex-expander}. This is because of the below
fact whose proof is similar to \Cref{fact:certify vertex expansion}.
\begin{fact}
Let $G=(V,E)$ be a graph that contains a large $\phi$-witness $W$.
Then $G$ has no $1/3$-vertex-balanced $(\phi/n^{o(1)})$-vertex-sparse
vertex cut. 
\end{fact}

We can now finally state the result we obtain for certifying an expander.

\begin{thm}
\label{lem:certify-witness} 
There is a deterministic algorithm $\certifywitness(G,\phi,\eps)$ that takes as input a directed $n$-vertex graph $G=(V,E)$, $\phi\in(0,1/\log^2(n)]$, and $\eps \in (0,1)$ in $\Ohat(m/\phi)$ time, either
	\begin{itemize}
		\item finds a $\Otil(\phi)$-vertex-sparse $\Omega(\epsilon/n^{o(1)})$-vertex-balanced
		cut $S$, or
		\item certifies that there exists a $\phi$-witness $W$ of $G$ such that $|V(W)|\ge(1-\epsilon)n$ and every edge in $W$ has weight at least $1$. Let $\alphaex = 1/n^{o(1)}$ be the precise expansion factor of $W$ guaranteed by this lemma (we will use this parameter in other lemmas).
	\end{itemize}
\end{thm}

\paragraph{(Almost) Expander Decomposition using Certification.} Note that finding large $\phi$-vertex-sparse vertex-cuts in expander certification allows us not only to add these cuts to the set $R$ and then recurse on the remaining strongly-connected subgraph but also to do so efficiently (since both sides of the cut are roughly of equal size so each vertex only participates a polylogarithmic number of times in such a recursion). It is straight-forward to show that adding sparse cuts and recursing leads to the bound on $R$ as described in \Cref{thm:expanderDecom} and certainly has $R$ as an incremental set. 

However, certifying $\phi$-expanders is a static procedure. Therefore a second procedure is required: essentially, once a $\phi$-expander $X$ with induced graph $G' = G[X]$ is certified, we have that for $\Omegatil(|X| \phi)$ updates, a subset $X'$ of $X$ still forms a $\phi$-expander and that $X'$ is of size $|X|/2$ during this period. In fact, we can maintain efficiently a set $\overline{P}$ that comes fairly close to $X'$:

\begin{restatable}[Directed Expander Pruning]{thm}{expanderPruning}
\label{thm:pruning} There is a deterministic algorithm that given a directed unweighted decremental multi-graph $G'$ with $n$ vertices and $m$ edges that is initially a $\phi$-expander and a parameter $L \ge 1$. The algorithm maintains an incremental set $P\subseteq V$ using $\tilde{O}\left(\frac{m n^{1/L}}{\gamma_{_L}(\phi)}\right)$ total update time such that for $\overline{P} = V(G') \setminus P$, we have that $G'[\overline{P}]$ is a $\gamma_{_L}(\phi)$-expander and $\vol_{G'}(P)\le O\left(\frac{t n^{1/L}}{\gamma_{_L}(\phi)}\right)$ after $t$ updates, where $\gamma_{_L}(\phi) = \phi^{3^{O(L)}}$.
\end{restatable}

Now, observe that the $\phi$ dependency of the above algorithm is fairly bad. However for very large $\phi$ (say $\sim 1/n^{\sqrt{\log n}} = 1/n^{o(1)}$), and using $L = \Theta(\log \log n)$, the above guarantees are only off from the optimal case by a subpolynomial factor. 

Luckily, given a graph $G'$ that is a $\phi$-expander, we can use the witness graph $W$ that we used to certify that $G'$ is a $\phi$-expander in-place of $G'$ at the expense that whenever we remove an edge from $W$, we have to remove up to $\Ohat(\phi^{-1})$ vertices from the expander. Using that $W$ is a $\phi'$-expander for $\phi'$ very large, it follows that we can maintain a set that comes close to $X'$ for $\Omegahat(|X| \phi)$ updates even for $G'$ in total time $\Ohat(m \phi^{-1})$ (we note that while the running time for \Cref{thm:pruning} is only $\Ohat(m)$, with every update it might cause us to remove $\Ohat(\phi^{-1})$ vertices from $X'$ explicitly).

\paragraph{Expander Set Maintenance.} Let us now address the problem of the last paragraph: we might have some vertices in $X \setminus X'$, i.e. vertices that were removed in \Cref{thm:pruning} from our graph $G'$ that was a $\phi$-expander at the first stage and was slightly decomposed during the last updates. We next have to determine what to do with vertices in $X \setminus X'$. We resort again to an idea from Chechik et al. \cite{chechik2016decremental}: We maintain an ES-tree in the graph $G' / X'$ from $X'$ to depth $\delta = \phi^{-1}$ in the graph $G$ where $X'$ is a single node obtained from contraction. Technically, $X'$ is a decremental set and thus not remain the same over time, we can use the augmented ES-tree from \Cref{lma:AugmentedGES} to deal with this technicality.

Then, we implement the following procedure: whenever a vertex $v \in X \setminus X'$ is a distance larger than $\phi^{-1}$, we find a balanced separator from $v$ (see for example \Cref{subsec:separators}), add separator vertices to $R$ as well, and recurse on the SCCs that do not contain $X'$ by trying to certify that they are expanders again. By basic balanced separator arguments we can bound the number of vertices we add to $R$ in this way again by $\Otil(\phi^{-1})$. Instead of refering to the set $X'$ however, we will henceforth simply refer to the (pruned) witness $W$. We summarize this in the somehow technical Theorem below which also bounds the running (which we need in all of its explicity to derive our SSSP data structure).

\begin{restatable}{thm}{PathToWitnessTheorem}
\label{thm:ES from witness} \label{thm:path-to-witness} There is a data structure $\pathtowitness(G,W,\phi)$ that takes as input an $n$-vertex $m$-edge graph $G = (V,E)$, a set $W \subseteq V$ with $|W| \geq |V|/2$ and a parameter $\phi > 0$. The algorithm must process two kinds of updates. The first deletes any edge $e$ from $E$; the second removes a vertex from $W$ (but the vertex remains in $V$), while always obeying the promise that $|W| \geq |V|/2$. The data structure must maintain a forest of trees $\fout$ such that every tree $T \in \fout$ has the following properties: all edges of $T$ are in $E(G)$; $T$ is rooted at a vertex of $W$; every edge in $T$ is directed away from the root; and $T$ has depth $\Ohat(1/\phi)$. The data structure also maintains a forest $\fin$ with the same properties, except each edge in $T$ is directed towards the root. 
	
At any time, the data structure may perform the following operation:  it finds a $\Ohat(\phi)$-sparse vertex cut $(L,S,R)$ with $W \cap (L \cup S) = \emptyset$ and replace $G$ with $G[R]$. (This operation is NOT an adversarial update, but is rather the responsibility of the data structure.) The data structure maintains the invariant that every $v \in V$ is present in exactly one tree from $\fout$ and exactly one from $\fin$; given any $v$, the data structure can report the roots of these trees in $O(\log(n))$ time. (Note that as $V$ may shrink over time, this property only needs to hold for vertex $v$ in the \emph{current} set $V$.) The total time spent processing updates and performing sparse-cut operations is $\Ohat(m/\phi)$. 
\end{restatable}

Although the data structure works for \emph{any} set $W$, $W$ will always correspond to a $\phi$-witness in the higher-level algorithm. The adversarial update that removes a vertex from $W$ corresponds to the event that the witness shrinks in the higher-level algorithm. 
The forests $\fin$ and $\fout$ allow the algorithm to return paths of length $\Ohat(1/\phi)$ from any $v \in V(G)$ to/from $W$: find the tree that contains $v$ and follow the path to the root, which is always in $W$. The requirement that each tree has low-depth will be necessary to reduce the update time. But once we add this requirement, we encounter the issue that some vertices may be very far from $W$, so we need to give the data structure a way to remove them from $V(G)$. This is the role of the sparse-cut operation: we will show in the proof that if $v$ is far from $W$, it is always possible to find a sparse vertex cut $(L,S,R)$ such that $v$ is in $L$ and hence removed from $G$. (The higher-level algorithm will process this operation by adding $S$ to $\Shat$, so that $L$ becomes part of a different SCC in $\gstar[\vstar \setminus \Shat]$.)

To summarize this paragraph, while we do not get the guarantees of an idealized expander decomposition, we can ensure that the expander sets computed at some stage are maintained such that their diameter is bound by $\Ohat(\phi^{-1})$ quite efficiently. This turns out to be enough for our applications.

\paragraph{Maintaining Shortest Paths in Expander Sets.} We also point out that since we would like to return the shortest paths in our SSSP data structure, \Cref{thm:path-to-witness} is not sufficient since it only allows us to find shortest paths that do not cross the witness graph $G$. However, we also require a second data structure that explicitly maintains approximate shortest paths between all pairs of vertices in an expander (here the approximation factor will be subpolynomial, that is quite huge). The input $W$ will always correspond to a large $\phi$-witness, and will thus have expansion $1/n^{o(1)}$. This data structure is not new to our paper, as it is essentially identical to an analogous structure for undirected graphs in \cite{ChuzhoyS20_apsp}. The only major difference is that we need to plug in our new expander pruning algorithm for directed graphs (Theorem \ref{thm:pruning}). Note that the theorem below will only allow us to find paths in $E(W)$, not $E(G')$, the graph of interest. However it is possible to use the embedding of $W$ to recover the correspond short paths in $E(G)$. 

\begin{restatable}{thm}{OracleTheorem}
	\label{thm:short-path-oracle} 
	There is a deterministic data structure $\shortoracle(W)$ that takes as input an $n$-vertex $m$-edge $1/n^{o(1)}$-expander $W$ subject to decremental updates. Each update can delete an arbitrary batch of vertices and edges from $W$, but must obey the promise that the resulting graph remains a $\phi$-expander. Given any query $u,v \in V(W)$, the algorithm returns in $n^{o(1)}$ time a directed simple path $P_{uv}$ from $u$ to $v$ and a directed simple path $P_{vu}$
	of $v$ to $u$, both of length at most $n^{o(1)}$. The total update time of the data structure is $\Ohat(m)$.
\end{restatable}

This completes our overview of components on witness graphs that certify that some graph $G'$ is a $\phi$-expander. It remains to address a central problem.

\paragraph{Sparse Witnesses in Dense Graphs.} In the previous paragraphs, we have reduced the problem of certifying an almost vertex-expander to maintaining a large witness (and then maintaining the witness using pruning) and sketched the ideas to maintain low-diameter components and a small set $R$ of removed vertices. However, although finding a low congestion embeddings in vertex expanders can be done very efficiently in the static setting (using the well known cut-matching game), there is
one crucial obstacle in the dynamic setting:

Consider the following simple scenario. We start with a complete graph
$G$ and parameter $\phi=\Omegahat(1)$. A standard (static) construction
of a large $\phi$-witness runs in $\Ohat(m)$ time and gives an \emph{unweighted} $\Omegahat(1)$-expander $W$ where all vertex degrees are $\Theta(\log n)$.
Let $\pset$ be the embedding of $W$. Observe that each path from $\pset$ has
value $1$ and $|\pset|=O(n\log n)$. 

Unfortunately, once the adversary knows $\pset$, he can destroy each embedding path $P\in\pset$ by deleting any edge in $P$. In total,
he can delete only $O(n\log n)$ edges in $G$ to destroy the whole
embedding of $W$. The algorithm would then have to construct a new
witness, which the adversary could again destroy with $O(n\log n)$
deletions. This process continues until $G$ has a balanced, sparse
vertex-cut, which might not happen until $\Omega(n^{2})$ deletions.
That is, this standard approach requires the algorithm to re-embed
a new witness $\tilde{\Omega}(n)$ times, which is not only slow,
but requires too many changes to the witness. 

\paragraph{Congestion Balancing.} To overcome this obstacle, we introduce a technique that we call \emph{congestion balancing} to maintain a witness $W$ that only needs to be re-embedded $\Ohat(1/\phi)$ times.

\begin{restatable}[Robust Witness Maintenance]{thm}{RobustWitnessTheorem}
	\label{thm:robust witness}There is a
	deterministic algorithm $\rwitness(G,\phi)$ that takes as input  a directed decremental $n$-vertex
	graph $G$ and a parameter $\phi \in (0,1/\log^2(n)]$. The algorithm maintains a large (weighted) $\phi$-short-witness $W$ of $G$ using $\Ohat(m/\phi^{2})$
	total update time such that every edge weight in $W$ is a positive multiple of $1/d$, for some number $d \leq 2\davg$, where $\davg$ is the initial average degree of $G$. The total edge weight in $W$ is $O(n\log n)$. After every edge deletion, 
	the algorithm either updates $W$ or outputs a $(\phi n^{o(1)})$-vertex-sparse
	$(1/n^{o(1)})$-vertex-balanced vertex-cut and terminates. 
	
	Let $W^{(i)}$ be $W$ after the $i$-th update. There exists a set
	$R$ of \emph{reset indices} where $|R|=\Ohat(\phi^{-1})$, such that
	for each $i\notin R$, $W^{(i)}\supseteq W^{(i+1)}$. That is, the
	algorithm has $\Ohat(\phi^{-1})$ phases such that, within each phase,
	$W$ is a decremental graph. The algorithm reports when each phase begins. It explicitly maintains the embedding $\pset$ of $W$ into $G$ and reports all changes made to $W$ and $\pset$.
\end{restatable}

To obtain the above result, we observe that an arbitrary embedding $\pset$ might not be robust to adversarial deletions, because a small number of edges might have most of the flow. To balance the edge-congestion, we introduce a capacity $\kappa(e)$ on each edge. Initially we set $\kappa(e) = 1/d$, where $d$ is the average degree in the input graph. At each step, the algorithms uses approximate flows and the cut-matching game to try to find a witness with vertex congestion $\Otil(1/\phi)$ and edge-congestions $\kappa(e)$. If it fails, the subroutine finds a low-capacity cut $C$; it then doubles capacities in $C$ and tries again. Since we assume a witness does exist, the algorithm will eventually find a witness once the edge-capacities are high enough.

Once we have a witness $W$ with embedding $\pset$, we use the lazy approach. Say the adversary deletes an edge $(u,v)$. Because our embedding obeyed capacity constraints, this can remove at most edges from $W$ of total weight at most $\kappa(u,v)$. To maintain expansion, we feed these deletions into our expander pruning algorithm (Theorem \ref{thm:pruning}) to yield a pruned set $P$, and shrink our witness to $W[V(W) - P]$. To guarantee that $W$ remains a large witness, we end the phase once the pruned set $P$ it too large. We will show that we end a phase only after the adversary deletes $\Omegahat(n)$ edge-capacity from the graph.

As with $\rmatching$, the crux of our analysis will be to show that the total of number of doubling steps is $\Ohat(1/\phi)$. To do so, we again use costs $c(e) = \log(d\kappa(e))$ and use a potential function $\Pi(G,\kappa)$ which measures the \emph{min-cost embedding} in $G$ among all \emph{very large} $\phi$-witness. As the vertex congestion is $1/\phi$, this potential $\Pi(G,\kappa)$ is at most $n/\phi$. Also, we are able to show that each doubling step increases the potential by $\Omegahat(n)$ using an argument that is more involved than the one for matching. Therefore, there are at most $\Ohat(1/\phi)$ doubling steps as desired. 

Given this bound, we can bound the total number of phases: each doubling step adds at most $n$ to the total capacity $\kappa$, and the initial capacity is at most $1/d \cdot m = n$. So the final total capacity is at most $\Ohat(n/\phi)$. As each phase must delete $\Omegahat(n)$ capacity, there are at most $\Ohat(1/\phi)$ phases. 

Again, to highlight the strength of this result, the above theorem shows we only need to re-embed a witness $\Ohat(\phi^{-1})$ times throughout the entire sequence of deletions, whereas we showed that the standard technique might require $\Omegatil(n)$ re-embeddings in the worst case.

\paragraph{A Deterministic Data Structure for decremental SSSP.} Finally, we note that much like in \Cref{chap_rand_decr_sssp}, we can actually use a low-diameter decomposition to find an approximate generalized topological order (ATO) which we can then use in conjunction with the ES-tree to obtain a fast implementation of decremental SSSP. The main difference between our data structure in the deterministic setting apart from the way that the low-diameter decomposition is maintained, is that we can also not allow a randomized separator but have to use deterministic separator vertices when constructing the ATO. We use the rest of this article to prove that it is indeed possible to obtain a non-trivial ATO even under these restrictions.

\section{Deterministic SSSP in decremental Graphs}

In this section, we prove our main results: \Cref{thm:ContributionDetSSSPResult}. Recall that our decremental SSR/SCC result combines our new expander-based framework with earlier techniques for decremental SCC in \cite{Lacki11,ChechikHILP16}. Our decremental SSSP results uses the new framework in a similar way, but now combines it with earlier tools for decremental SSSP in \cite{GutenbergW20a,nearOptDenseSSSP}. In particular, we start with the following proposition, which essentially combined Proposition \ref{thm:lacki} and Theorem \ref{thm:path-to-witness}.

\begin{restatable}{prop}{PGWNProp}
	\label{prop:PGWN}
Let $G=(V,E,w)$ be a weighted decremental graph, and $s \in V$ a fixed source. Let $\mathcal{A}$ be a data structure given some integer $d >0$, that processes edge deletions to $E$ and after every edge deletion ensures that {\bf 1)} $G$ is strongly-connected and has diameter at most $d$ and {\bf 2)} supports path queries between any two vertices in $G$ that returns a path of length $\Ohat(d)$ in almost-path-length query time and runs in total update time $T(m,n,d)$ (here we assume $T(m_1, n_1, d_1) + T(m_2, n_2, d_2) \leq T(m,n,d)$ for all choices $m,n,d$ and $m_1, m_2, n_1, n_2, d_1, d_2$ such that $m = m_1 + m_2$, $n = n_1 + n_2$ and $d = d_1 + d_2$). At any time the data structure may perform the following operation: it finds and outputs a $\Ohat(1/d)$-sparse cut $(L,S,R)$ where $|L| \leq |R|$ and replaces $G$ with $G[R]$; here we only require the algorithm to output $L$ and $S$ explicitly. (This sparse-cut operation is not an adversarial update, but is rather something the data structure can do of its own accord at ay time.)

Then, there exists a deterministic data structure $\mathcal{B}$ that can report $(1+\epsilon)$-approximate distance estimates and corresponding paths from $s$ to any vertex $v \in V$ in the graph $G$ in almost-path-length query time and has total update time $\Ohat((T(m,n,\delta) + n^3/\delta + n^2 \delta + mn^{2/3})\log W/\epsilon)$ for any choice of $\delta, \epsilon > 0$. (Note that the data structure can cause $V(G)$ to shrink over time via sparse-cut operations, so it only has to answer queries for vertices $u,v$ in the \emph{current} graph.)
\end{restatable}

It is straight-forward to obtain \Cref{thm:ContributionDetSSSPResult} from the proposition, and \Cref{thm:robust witness}.

\begin{proof}[Proof of \Cref{thm:ContributionDetSSSPResult}.]
We now show how to implement the data structure $\mathcal{A}$ required by the setup of \Cref{prop:PGWN}, with $T(m,n,\delta) = \Ohat(m\delta^2)$ as follows. Given the graph $G$, we can invoke the algorithm described in \Cref{thm:robust witness} with parameter $\phi = \hat{\Theta}(1/\delta)$, such that the algorithm  maintains a $\phi$-short-witness $W$ that restarts up to $\Ohat(1/\phi) = \Ohat(\delta)$ times. Whenever $W$ starts a new \emph{phase}, we use the data structures from \Cref{thm:path-to-witness} and 	\Cref{thm:short-path-oracle} on $G$ and $W$ until the phase ends. We forward the sparse cuts $(L,S,R)$ found in the algorithm from \Cref{thm:path-to-witness} and  \Cref{thm:robust witness} and update $G$ accordingly. Thus after the algorithm from \Cref{thm:robust witness} terminates, the graph $G$ contains only a constant fraction of the vertices that the algorithm in \Cref{thm:robust witness} was initialized upon. We then repeat the above construction and note that after at most $O(\log n)$ times, the graph $G$ is the empty graph. 

We note that to obtain a path between any two vertices in the current graph $G$, we can query the data structures from \Cref{thm:path-to-witness} and \Cref{thm:short-path-oracle} to obtain such a path of length $\Ohat(\delta)$ in almost-path-length time. We further observe that if we set $\phi$ to $\frac{1}{\delta n^{o(1)}}$, for a large enough subpolynomial factor $n^{o(1)}$, then we can ensure that vertices in $G \setminus W$ are at all times at most $\delta/3$ away from some vertex in $W$ by \Cref{thm:path-to-witness}, have that any two vertices in $W$ are at distance at most $\delta/3$ to each other in $G$ by \Cref{thm:short-path-oracle} and \Cref{thm:robust witness}, and again, that there exists a path to every vertex in $G \setminus W$ to a vertex in $W$ of length at most $\delta/3$. But this implies that any two vertices in $G$ are at all times at distance at most $\delta$ and therefore the diameter of $G$ is upper bounded by $\delta$, as required.

The total update time of the data structure $\mathcal{A}$ is at most $\Ohat(m/\phi^2) = \Ohat(m\delta^2)$ by adding the running time of \Cref{thm:robust witness} with the running time induced by the algorithms in \Cref{thm:path-to-witness} and 	\Cref{thm:short-path-oracle} which are restarted in $\Ohat(\delta)$ phases. 

We thus derive an algorithm $\mathcal{B}$ as specified in \Cref{prop:PGWN}, where we use the above data structure $\mathcal{A}$ and where we set $\delta = n^{1/3}$ which gives total update time
\begin{align*}
&\Ohat((T(m,n,\delta) + n^3/\delta + n^2 \delta + mn^{2/3})\log W/\epsilon)= n^{2+2/3 + o(1)} \log W/\epsilon.
\end{align*}
\end{proof}

The rest of this section is dedicated to prove \Cref{prop:PGWN}. We refer the reader to the additional preliminaries from last chapter in \Cref{sec:prelim} for notation used in the next section.  We then again use the abstraction of approximate topological orders which we reduce the problem to and finally prove that an approximate topological order can be maintained efficiently.

\subsection{SSSP via Approximate Topological Orders}

Let us briefly restate the concept of an approximate topological order which we define slightly different but still quite similar to \Cref{chap_rand_decr_sssp} and which we implement similar to \cite{GutenbergW20a}. Recall that the main idea of an approximate topological order is as follow: consider the generalized topological order $(\mathcal{V}, \tau)$ of a graph $G$. Then $G / \mathcal{V}$ is a directed acyclic graph by definition. But this implies that for any (shortest) $s$-to-$t$ path $\pi_{s,t}$ in $G$ we have that  every edge $(X,Y)$ on $\pi_{s,t} / \mathcal{V}$ in $G / \mathcal{V}$ has $\tau(X) < \tau(Y)$. Since further $\tau$ maps to numbers between $1$ and $n$, we have thus that summing along the topological difference of the edges of  $\pi_{s,t} / \mathcal{V}$, we that $\mathcal{T}(\pi_{s,t}, (\mathcal{V}, \tau))$ is at most $n$.

Next, let us assume that the sum of diameters of all SCCs in $\mathcal{V}$ is at most $\epsilon \delta$, then for any shortest path $\pi_{s,t}$, we can upper bound the difference in weight between $\pi_{s,t}/\mathcal{G}$ path in $G / \mathcal{V}$ as opposed to $\pi_{s,t}$ in $G$ by an additive term of $\epsilon \delta$. So, if $\pi_{s,t}$ is of weight at least $\delta$, the additive term can be subsumed in a multiplicative error of $(1\pm\epsilon)$.

Now, the gist of this set-up is that given this upper bound on $\mathcal{T}(\pi_{s,t}, (\mathcal{V}, \tau))$, we can implement a fast SSSP data structure as follows. We know that on a path of length $\delta$ in $G / \mathcal{V}$ there are at most $\delta / 2^i$ edges that have topological order difference more than $2^i n/\delta$ by the pigeonhole principle for any $i$. But this implies that adding an additive error of $\epsilon 2^i$ on each such edge would only amount to an $(1+\epsilon)$ multiplicative error of a shortest path of length $\delta$. But this allowance for a significant additive error can be exploited to speed-up the SSSP data structure significantly because it allows for vertices to consider the neighbors that are close in topological order difference more closely while being more lenient when passing updates to vertices that are far in terms of topological order difference. 

Before we state a data structure from \Cref{chap_rand_decr_sssp} that exploits this very efficiently, let us now state more formally the construct of an approximate topological order. Here, we point out one last issue: we cannot assume that SCCs in $G$ have small diameter in general. Therefore we maintain the generalized topological order on a graph $G'$ initialized to $G$ where we, additionally to adversarial edge updates to $G$, also take vertex separators $S$ such that edges incident to $S$ are deleted from $G'$. This ensures that all SCCs in $G'$ have small diameter. Relating back to $G$ (where no separator was deleted) we have that $\mathcal{T}(\pi_{s,t}, (\mathcal{V}, \tau))$ might be increased by this operation since some edge $(X,Y)$ on $\pi_{s,t} /\mathcal{V}$ with $X$ or $Y$ containing a separator vertex $S$, such that $(X,Y)$ might now go "backwards" in the topological order, i.e. have $\tau(X) > \tau(Y)$. This increases $\mathcal{T}(\pi_{s,t}, (\mathcal{V}, \tau))$ by up to $2n-2$ for every separator vertex since we might move along $(X,Y)$ all the way back in the topological order and then forward again. However, by choosing small separators, we can still bound $\mathcal{T}(P, (\mathcal{V}, \tau))$ by a non-trivial upper bound.

Without further due, let us give the formal definition of an approximate topological order.

\begin{restatable}{defn}{ATO}
\label{defn:ATOdecomposition}
Given a decremental weighted digraph $G=(V,E,w)$ and parameter $\eta \leq n$ and $\nu \leq W$, we say a dynamic tuple $(\mathcal{V}, \tau)$ where $\mathcal{V}$ partitions $V$, and $\tau : \mathcal{V} \rightarrow [1,n]$, is an $\mathcal{ATO}(G, \eta, \nu)$ if at each stage
\begin{enumerate}
    \item $\mathcal{V}$ forms a refinement of all earlier versions of $\mathcal{V}$ and $\tau$ is a \emph{nesting} function, i.e. $\tau$ initially assigns each set in $X$ in the initial version of $\mathcal{V}$ a number $\tau(X)$, such that no other set $Y$ in $\mathcal{V}$ has $\tau(Y)$ in the interval $[\tau(X), \tau(X) + |X| - 1]$. If some set $Y \in \mathcal{V}$ is split at some stage into disjoint subsets $Y_1, Y_2, .., Y_l$, then we let $\tau(Y_1) = \tau(Y)$ and $\tau(Y_{i+1}) = \tau(Y_i) + |Y_i|$. We then return a pointer to each new subset $Y_i$ such that all vertices in $Y_i$ can be accessed in time $O(|Y_i|)$. The value $\tau(X)$ for each $X \in \mathcal{V}$ can be read in constant time. \label{prop:TauUpdate}
    \item each set $X$ in $\mathcal{V}$ has weak diameter $\mathbf{diam}(X, G) \leq \frac{|X| \eta \nu}{n}$, and \label{prop:contractLittle2}
    \item At each stage, for any vertices $s, t \in V$, the shortest-path $\pi_{s,t}$ in $G$ satisfies $\mathcal{T}(\pi_{s,t}, (\mathcal{V}, \tau)) = \Ohat
    \left(\frac{n^2}{\eta} + n\cdot \frac{\mathbf{dist}_G(s,t)}{\nu}\right)$. \label{prop:TauTotal2}
\end{enumerate}
\end{restatable}

Here, we captured in Property \ref{prop:TauUpdate}, that the vertex sets in $\mathcal{V}$ decompose over time, that $\tau$ is \emph{nesting} and that all sets are easily accessible. In Property \ref{prop:contractLittle2}, we capture that the sum of diameters of the vertex sets in $\mathcal{V}$ is small. It is not hard to see that by summing the upper bound on the diameter of all such sets $X$ in $\mathcal{V}$, we get that the sum of diameters is bounded by $\eta \nu$. Finally, we give an upper bound for the topological order difference for any shortest-path in $G$.

The main result of the next section, shows that we can maintain an $\mathcal{ATO}$ using data structure $\mathcal{A}$ from \Cref{prop:PGWN}.

\begin{restatable}{lemma}{AtoDecomp}
\label{lem:ATOdecomposition}
Given a decremental weighted digraph $G=(V,E,w)$, parameters $\eta \leq n, \nu \leq W$, and a data structure $\mathcal{A}$ as described in \Cref{prop:PGWN} that can for each SCC $X$ in $\mathcal{V}$ at any point return a path between any two vertices $u,v \in X$ of length $\Ohat(\frac{|X|\eta \nu}{n})$ in near-linear time. Then, we can deterministically maintain a $\mathcal{ATO}(G, \eta, \nu)$ in total update time $\Ohat(T(m,n,\eta) + mn^{2/3})$.
\end{restatable}

We now restate \Cref{thm:SSSPEfficient} from \Cref{chap_rand_decr_sssp} although we present it in a slightly modified form to adapt it to the modified definition of an $\mathcal{ATO}$ that we use for this paper. However, the adaption is obtained straight-forwardly and we refer the reader to \Cref{thm:SSSPEfficient} to verify.

\begin{restatable}[see \Cref{thm:SSSPEfficient}]{theorem}{ssspSimple2}
\label{thm:SSSPEfficient2}
Given $G=(V,E,w)$, a decremental weighted digraph, a source $r \in V$, an approximation parameter $\epsilon > 0$, and access to $(\mathcal{V}, \tau)$ an $\mathcal{ATO}(G, \eta, \nu)$.

Then, there exists a deterministic data structure that maintains a distance estimate $\widetilde{\mathbf{dist}}(r,v)$ for every vertex $v \in V$ such that at each stage of $G$, $\mathbf{dist}_G(r,v) \leq \widetilde{\mathbf{dist}}(r,v)$ and if $\mathbf{dist}_G(r,v) \in [\eta \nu/\epsilon, 2\eta \nu/\epsilon)$, then 
\[
\widetilde{\mathbf{dist}}(r,v) \leq (1+\epsilon)\mathbf{dist}_G(r,v)
\]
and the algorithm can for each such vertex $v$, report a path of length $(1+\epsilon)\mathbf{dist}_G(r,v)$ in the graph $G / \mathcal{V}$ in almost-path-length time. The total time required by this structure is 
\[
\Ohat\left(\frac{n^3}{\eta\epsilon} + \cdot \frac{n^2 \eta}{\epsilon}\right).
\]
\end{restatable}

We can now prove \Cref{prop:PGWN}.

\begin{proof}[Proof of \Cref{prop:PGWN}.]
For every $0 \leq i \leq \lg W$, where $W$ is the aspect ration of $G=(V,E,w)$, we maintain at level $i$, an $\mathcal{ATO}(G, \delta, 2^i)$ using \Cref{lem:ATOdecomposition}, and then running \Cref{thm:SSSPEfficient2} on $G$ and the $\mathcal{ATO}(G, \delta, 2^i)$ from our source vertex $s$ to depth $\delta \cdot 2^i$. Thus, each such data structure maintains for every vertex $v$ at distance $[\delta \cdot 2^i/\epsilon', \delta \cdot 2^{i+1}/\epsilon')$ from $s$ an $(1+\epsilon')$-approximate distance estimate. We can therefore find for every vertex $v$ at distances larger than $\delta/\epsilon'$ from $s$ a distance estimate in some of these data structures that gives the right approximation, and since all data structures overestimate the distance, we can find the right distance estimate by comparing all distance estimates $\widetilde{\mathbf{dist}}(s,v)$. Finally, we can maintain a simple ES-tree in time $O(m \delta/\epsilon')$ to obtain exact distances from $s$ to every vertex at distance at most $\delta$.

It is not hard to verify that the total update time of all data structures is
\begin{align*}
&\sum_{0 \leq i \leq \lg W} \left( \Ohat\left(\frac{n^3}{\delta\epsilon'} + \cdot \frac{n^2 \delta}{\epsilon'}\right) + \Ohat(T(m,n,\delta) + mn^{2/3})\right) \\
&= \Ohat((T(m,n,\delta) + n^3/\delta + n^2 \delta + mn^{2/3})\log W/\epsilon').
\end{align*}
for $\epsilon'$ to be set $\epsilon' = \epsilon/n^{o(1)}$ which is again subsumed in the $\Ohat$-notation.

To answer path queries for a $s$-to-$v$ path $\pi_{s,v}$, we query the corresponding shortest path data structure where we found a $(1+\epsilon')$-approximation. This gives us the path $\widetilde{\pi_{s,v}}$ in $G / \mathcal{V}$ for some $\mathcal{ATO}$ $(\mathcal{V}, \tau)$. We then identify for every vertex $x$ on $\widetilde{\pi_{s,v}}$ the corresponding SCC in $\mathcal{V}$ and the two endpoints in $G$ of the incident edges on $\widetilde{\pi_{s,v}}$. We can then query for a path between these two vertices in the $\mathcal{ATO}$ data structure. Summing over all exposed paths, by \Cref{lem:ATOdecomposition}, we can extend the path $\widetilde{\pi_{s,v}}$ to a path in $G$ of length $(1+\epsilon')\mathbf{dist}_G(s,v) + \Ohat(\eta \nu)$. But we have that $\mathbf{dist}_G(s,v) \geq \delta/\epsilon'$. Thus, setting $\epsilon'$ to $\epsilon/2$ divided by the subpolynomial factor hidden in $\Ohat(\eta \nu)$, we obtain a path of length $(1+\epsilon)\mathbf{dist}_G(s,v)$. Since each piece on the path can be obtained in almost-path-length time, we can also construct the extension of path $\widetilde{\pi_{s,v}}$ to a path in $G$ in almost-path-length time. This completes the proof.
\end{proof}
\subsection{A Deterministic Algorithm to Maintain an Approximate Topological Order}

Finally, let us prove the main ingredient to achieve our result. 

\AtoDecomp*
\begin{proof}
We start the proof by partitioning the edge set $E$ of the initial graph $G$ into edge set $E^{heavy}$ and $E^{light}$. We assign every edge $e \in E$ to $E^{heavy}$ if its weight $w(e)$ is larger than $\nu$, and to $E^{light}$ if $w(e) \leq \nu$. 

We now describe our algorithm where we focus on the graph $G$ where the edge set $E^{heavy}$ is removed. As we will see later, there can only be few edges from $E^{heavy}$ on any shortest path. Let us start the proof by giving an overview and then a precise implementation. We finally analyze correctness and running time. 

\paragraph{Algorithm.} Our goal is subsequently to maintain an incremental set $\hat{S} \subseteq V$ such that every SCC $X$ in $G' = G \setminus E(\hat{S}) \setminus E^{heavy}$ has unweighted diameter at most $\frac{|X|\eta}{n}$. Since each edge weight is at most $\nu$ this will imply that every SCC $X$ in the weighted version of $G'$ has diameter at most $\frac{|X|\eta\nu}{n}$. 

We then maintain $(\mathcal{V}, \tau)$ as the generalized topological order of $G'$ using the data structure described in \Cref{thm:SCCinDecrGraph} which is a straight-forward extension of \Cref{thm:ContributionSCCmain} using internally the algorithm by Tarjan \cite{tarjan1972depth} as described in \cite{GutenbergW20a, nearOptDenseSSSP}.

To maintain $G'$, we initialize a data structure $\mathcal{A}$ on every SCC $X$ in the initial set $\mathcal{V}$ on the graph $G'[X]$ with parameter $d = \frac{|X|\eta}{2n}$. Then, whenever such a data structure $\mathcal{A}$ that currently operates on some graph $G'[Y]$, announces a sparse cut $(L, S, R)$ and sets its graph to $G'[R]$, we add $S$ to $\hat{S}$ and then initialize a new data structure $\mathcal{A}'$ on $G'[L]$ with parameter $d = \frac{|L|\eta}{2n}$. Further, if the data structure $\mathcal{A}$ was initialized on a graph with vertex set at least twice as large as $R$, we delete $\mathcal{A}$, and initialize a new data structure $\mathcal{A}''$ on $G'[R]$ with $d = \frac{|R|\eta}{2n}$. This completes the description of the algorithm.

\paragraph{Correctness of the Algorithm.} We prove each property of the theorem individually:
\begin{itemize}
    \item \underline{Property \ref{prop:TauUpdate}}: It is straight-forward to see that since $(\mathcal{V}, \tau)$ is the generalized topological order of $G' \subseteq G$ and since it is maintained to satisfy the nesting property, that Property \ref{prop:TauUpdate} follows immediately. 
    \item \underline{Property \ref{prop:contractLittle2}:} Observe that $\mathcal{V}$ is the set of SCCs in $G'$. Further, observe that we maintain the data structures $\mathcal{A}_1, \mathcal{A}_2, \dots$ such that vertex set of all graphs that they run on spans all vertices in $V \setminus S$. For the vertices in $S$ we have that each $s \in \hat{S}$ forms a trivial SCC and therefore certainly satisfies the constraint. For each set $X$ that some data structure $\mathcal{A}$ runs upon, we have that the unweighted diameter is at most the $d$ that $\mathcal{A}$ was initialized with. Observe that we delete data structures if the size of the initial vertex set $Y$ is decreased by factor $2$. Thus, we have that the data structure $\mathcal{A}$ was initialized for some $d = \frac{|Y|\eta}{2n} \leq \frac{|X|\eta}{n}$. Since the largest edge weight in $G'$ is $\nu$, we thus have that for each SCC $X$ in $\mathcal{V}$, we have $\mathbf{diam}(X,G') \leq \frac{|X|\eta \nu}{n}$. Adding edges in $E(\hat{S})$ and $E^{heavy}$ can further only decrease the weak diameter and therefore we finally obtain that,
    \[
        \mathbf{diam}(X,G) \leq \frac{|X|\eta \nu}{n}.
    \]
    \item \underline{Property \ref{prop:TauTotal2}}: In order to establish the last property, let us partition the set $\hat{S}$ into sets $S_1, S_2, \dots, S_{\lg n}$ where a vertex $s$ is in $S_i$ if it joined $\hat{S}$ after a data structure $\mathcal{A}$ announced it that was initialized on a graph $G'[Y]$ where $Y$ was of size $[n/2^{i+1},n/2^i)$. Since we delete data structures after their initial vertex set has halved in size, we have that are such data structure that added vertices to a set $S_i$ ran with $d \geq \frac{(n/2^{i+1}) \eta}{2n} = \frac{n \eta}{2^{i+2}}$. Since each such set of vertices $S$ that was added to $S_i$ is $\Ohat(1/d)$-sparse and we then only compute sparse cuts on the induced subgraphs of the cut, we further have that there are at most $\Ohat(n/d) = \Ohat(2^i/\eta)$ vertices in $S_i$ at the end of the algorithm. Further, we observe that every edge $(u,v)$ that was contained in the subgraph $G'[Y]$ when $\mathcal{A}$ was initialized has both endpoints in $Y$ and therefore by property \ref{prop:TauUpdate}, we have $|\tau(u) - \tau(v)| < |Y| \leq n/2^{i-1}$. 
    
    Now, let us fix any shortest path $\pi_{s,t}$ in $G$ (in the current version). Instead of analyzing $\mathcal{T}(\pi_{s,t}, (\mathcal{V}, \tau))$, let us analyze 
    \[
    \mathcal{T}'(\pi_{s,t}, (\mathcal{V}, \tau)) \stackrel{\text{def}}{=} \sum_{(u,v) \in \pi_{s,t}} \max\{ 0 , \tau(u) - \tau(v)\}.
    \]
    which only considers the edges on the path that go "backwards" in the topological order. However, it can be seen that for every path $\mathcal{T}(\pi_{s,t}, (\mathcal{V}, \tau)) \leq 2\mathcal{T}'(\pi_{s,t}, (\mathcal{V}, \tau)) + n$.

    For edges on $\pi_{s,t}$ in $E^{heavy}$, we observe that each such edge $(u,v)$ can contribute to $\mathcal{T}'(\pi_{s,t}, (\mathcal{V}, \tau))$ at most $n$ since $\tau(u) - \tau(v) \leq n$ (trivially since both numbers are taken from the interval $[1,n]$). Further, since each such edge adds weight at least $\nu$ to the shortest path, there are at most $\frac{\mathbf{dist}(s,t)}{\nu}$ such edges. Thus, the total contribution by all these edges is at most $n \frac{\mathbf{dist}(s,t)}{\nu}$.
    
    For the edges on $\pi_{s,t}$ in $E^{light}$, we observe that each edge $(u,v)$ that contributes to $\mathcal{T}'(\pi_{s,t}, (\mathcal{V}, \tau))$ is not in $G'$ since $(\mathcal{V}, \tau)$ is a generalized topological order of $G'$ and therefore directed "forwards" (recall the definition in \cref{sec:prelim}). Thus, each such edge is in $E(\hat{S})$ and therefore incident to some vertex $s$ in some $S_i$. But then it adds at most $n/2^{i-1}$ to $\mathcal{T}'(\pi_{s,t}, (\mathcal{V}, \tau))$ by our previous discussion. Since a path only visits each vertex once, and by our bound on the size of $S_i$, we can now bound the total contribution by
    \[
        \mathcal{T}'(\pi_{s,t}, (\mathcal{V}, \tau)) \leq \sum_{i} |S_i| n/2^{i-1} = \Ohat(n^2/\eta + n \frac{\mathbf{dist}(s,t)}{\nu})
    \]
\end{itemize}

\paragraph{Bounding the Running Time.} Observe that for any vertex $x \in V$, that between any two times that it part of a graph $G'[Y]$ that a data structure $\mathcal{A}$ is invoked upon and of graph $G'[X] \subseteq G'[Y]$, the set $X$ is of at most half the size of $Y$. This follows by the definition of data structure $\mathcal{A}$ which whenever a sparse cut $(L,S,R)$ is output, continues on the graph $G'[R]$ where $R$ is larger than $L$ while no data structure is thereafter initialized on a graph containing any vertex in $S$. 

But if the SCC that some vertex $x$ is contained in halves in size every time between two data structures $\mathcal{A}$ are initialized upon $x$, then we have that $x$ participates in at most $\lg n$ data structures over the entire course of the algorithm. Since each edge $(x,y)$ or $(y,x)$ for any $y \in V$ is only present in the induced graph containing $x$, we have that no data structure that is not initialized on a graph with vertex set contain $x$ has $(x,y)$ or $(y,x)$ in its graph. Thus, every edge only participates in $\lg n$ graphs. 

Finally, we observe that the distance parameter $d$ that each data structure $\mathcal{A}$ is upper bounded by $\eta/2$. Thus, by the (super-)linear behavior of the function $T(m,n,d)$, we have that the total update time for all data structures in $\Ohat(T(m,n,\eta))$. Further, we have by \Cref{thm:SCCinDecrGraph} that the data structure maintaining $(\mathcal{V}, \tau)$ can be implemented in time $\Ohat(mn^{2/3})$. The time required for all remaining operations is subsumed in both bounds. 

\paragraph{Returning the Paths.} For any SCC $X$ in $\mathcal{V}$, we have that there is a data structure $\mathcal{A}$ on $G'[X]$ that allows for SCC queries. Since by our previous discussion each such data structure runs with $d$ at most $\frac{|X|\eta}{n}$ and each edge on the path has weight at most $\nu$ (recall that $G'$ only contains edges of small weight), we can return the path from data structure $\mathcal{A}$ on query.   
\end{proof}

%% file: incr_sssp.tex
\label{chap:incr_sssp}

In this chapter, we prove the second part of \Cref{thm:ContributionSSSPResult}. We state the precise Theorem.

\begin{theorem}[Incremental Part of \Cref{thm:ContributionSSSPResult}]
\label{thm:ContributionIncrSSSPResult}
Given a decremental input graph $G=(V,E,w)$ with $n = |V|, m=|E|$ and aspect ratio $W$, a dedicated source $r \in V$ and $\epsilon > 0$, there is a deterministic algorithm that maintains a distance estimate $\widetilde{\mathbf{dist}}(r,x)$, for every $x \in V$, such that
\[
    {\mathbf{dist}}_G(r,x) \leq \widetilde{\mathbf{dist}}(r,x) \leq (1+\epsilon) {\mathbf{dist}}_G(r,x)
\]
at any stage of $G$. The algorithm has total update time $\tilde{O}(n^2 \log W/\poly(\epsilon))$. Distance queries are answered in $O(1)$ time, and a corresponding path $P$ can be returned in $O(|P|)$ time.
\end{theorem}

As opposed to previous sections, we do not start this chapter with an overview, but instead present a first $O(n^{2+2/3}/\epsilon)$ total update time algorithm that is simple and captures our main ideas. We then give an overview that explains how to extend the ideas to obtain the running time stated in \Cref{thm:ContributionIncrSSSPResult}. The rest of the chapter is then concerned with proving \Cref{thm:ContributionIncrSSSPResult} (however, we only sketch some of the more technical proofs). We point out that in order to avoid clutter, we first prove \Cref{thm:ContributionIncrSSSPResult} for unweighted graphs, and only in the last section show that in a weighted graph the dependency on the weight ratio becomes only polylogarithmic, and therefore we can use the reduction stated in \Cref{chap:prelim} to reduce the dependency on $W$ to $\log W$ as claimed.

\section{Warm-up: An $O(n^{2+2/3}/\epsilon)$ Time Algorithm}

In this section we describe an algorithm for incremental SSSP on unweighted directed graphs with total update time  $O(n^{2+2/3}/\epsilon)$. This algorithm illustrates the main ideas used in our $\tilde{O}(n^2 \log W/\epsilon)$ algorithm.

\begin{theorem}
There is a deterministic algorithm that given an unweighted directed graph $G=(V,E)$ subject to edge insertions, a vertex $r \in V$, and $\epsilon>0$, maintains for every vertex $v$ an estimate $\widetilde{\mathbf{dist}}(r,v)$ such that after every update $\mathbf{dist}(r,v)\leq \widetilde{\mathbf{dist}}(r,v)\leq (1+\epsilon)\mathbf{dist}(r,v)$, in total time $O(n^{2+2/3}/\epsilon)$.
\end{theorem}

To obtain this result, we take inspiration from a simple property of \emph{undirected} graphs: \emph{Any two vertices at distance at least $3$ have disjoint neighborhoods}. This observation is crucial in several spanner/hopset constructions as well as other graph algorithms (for example \cite{awerbuch1985complexity, elkin20041, bodwin2016better, elkin2016hopsets, henzinger2016dynamic}), as well as partially dynamic SSSP on undirected graphs \cite{bernstein2016deterministic, bernstein2017deterministic}. In~\cite{bernstein2016deterministic}, Bernstein and Chechik exploit this property for partially dynamic undirected SSSP in the following way. The property implies that for any vertex $v$ on a given shortest path from $r$ to some $t$, the neighborhood of $v$ is disjoint from almost all of the other vertices on this shortest path. Thus there cannot be too many high-degree vertices on any given shortest path, and therefore high-degree vertices are allowed to induce large additive error which can be exploited to increase the efficiency of the algorithm (much like in \Cref{chap_rand_decr_sssp}).

Whilst we would like to argue along the same lines, this property is unfortunately not given in \emph{directed} graphs: {there could be two vertices $u$ and $v$ at distance $3$, and a third vertex $z$ that only has in-coming edges from $u$ and $v$. Clearly, $u$ and $v$ can now still be at distance $3$ whilst their out-neighborhoods overlap}. We overcome this issue by introducing \emph{forward neighborhoods} $\mathcal{FN}(u)$ that only include vertices from the out-neighborhood $\mathcal{N}^{out}(u)$ that are estimated to be further away from the source vertex $r$ than $u$. Now, suppose there are two vertices $u$ and $v$ both appear on some shortest path from $r$ to some $t$ and whose forward neighborhoods overlap. Let $w$ be a vertex in $\mathcal{FN}(u) \cap \mathcal{FN}(v)$. Since $w$ has a larger distance estimate than $u$ and the edge $(u,w)$ is in the graph, the distance estimates of $u$ and $w$ must be close, assuming that each distance estimate does not incur much error.
Similarly, the distance estimates of $v$ and $w$ must be close. But then the distance estimates of $u$ and $v$ must also be close. Therefore, the forward neighborhood of each vertex on a long shortest path must only overlap with the forward neighborhoods of few other vertices on the path. In summary, our extension of the property to directed graphs is that if the distance estimates of $u$ and $v$ differ by a lot, then $u$ and $v$ have disjoint forward neighborhoods.

\paragraph{The data structure}

In order to illustrate our approach, we present a data structure that only maintains approximate distances for vertices $u$ that are at distance $\mathbf{dist}(r,u) > n^{2/3}$. This already improves the state of the art since  we can maintain the exact distance $\mathbf{dist}(r,u)$ if $\mathbf{dist}(r,u)\leq n^{2/3}$ simply by using a classic ES-tree to depth $n^{2/3}$ which runs in time $O(mn^{2/3})$. 

To understand the motivation behind our main idea, let us first consider a slightly modified version of the classic ES-trees that achieves the same running time: We maintain for each vertex $u\in V$ an array $A_u$ with $n$ elements where $A_u[i]$ is the set of all vertices $v \in \mathcal{N}^{out}(u)$ with $\mathbf{dist}(r,v)=i$. Then, when $\mathbf{dist}(r,u)$ decreases, the set of vertices in $\mathcal{N}^{out}(u)$ whose estimated distance from $r$ decreases is exactly the set of vertices stored in $A_u[\mathbf{dist}^{NEW}(r,u) + 2, n]$ which we call the \emph{forward neighborhood} $\mathcal{FN}(u)$ of $u$. (Recall that $A_u[i,j]$ is  the subarray of $A$ from index $i$ to index $j$, inclusive.) That is since each such vertex $v$ has at estimated distance more than $\mathbf{dist}(r,u) + 1$ thus relaxing the edge $(u,v)$ is ensured to decrease $v$'s distance. Thus, we only need to scan edges with tail $u$ and head $v \in \mathcal{FN}(u)$, however, we also need to update $A_u$ whenever an in-neighbor of $u$ decreases its distance estimate.

For our data structure which we call a ``lazy'' ES-tree, we relax several constraints and use a lazy update rule. Instead of maintaining the exact value of $\mathbf{dist}(r,v)$ for all $v \in \mathcal{N}^{out}(u)$, we only maintain an \emph{approximate} distance estimate $\widetilde{\mathbf{dist}}(r,v)$. Whilst we still maintain an array for each vertex $u \in V$, we now only update the position of $v$ only after $\widetilde{\mathbf{dist}}(r,v)$ has decreased by at least $n^{1/3}$ or if $(u,v)$ was scanned by $u$. To emphasize that this array is only updated occasionally, instead of using the notation $A_u$, we use the notation $\texttt{Cache}_u$. Again, we define $\texttt{Cache}_u[\widetilde{\mathbf{dist}}(r,u) + 2, n]$ to be the \emph{forward neighborhood} of $u$ denoted $\mathcal{FN}(u) \subseteq \mathcal{N}^{out}(u)$. Further, if $\mathcal{FN}(u)$ is small (say of size $O(n^{2/3})$), we say $u$ is \emph{light}. Otherwise, we say that $u$ is \emph{heavy}.

Now, we distinguish two scenarios for our update rule: if $u$ is \emph{light}, then we can afford to update the distance estimates of the vertices in $\mathcal{FN}(u)$ after every decrease of $\widetilde{\mathbf{dist}}(r,u)$. However, if $u$ is \emph{heavy}, then we only update the vertices in $\mathcal{FN}(u)$ after the distance estimate $\widetilde{\mathbf{dist}}(r,u)$ has been decreased by at least $n^{1/3}$ since the last scan of $\mathcal{FN}(u)$.

Additionally, for each edge $(u,v)$, every time $\widetilde{\mathbf{dist}}(r,v)$ decreases by at least $n^{1/3}$, we update $v$'s position in $\texttt{Cache}_u$.

Finally, we note that $|\mathcal{FN}(u)|$ changes over time and so we need to define the rules for when a vertex changes from light to heavy and vice versa more precisely. Initially, the graph is empty and we define every vertex to be light. Once the size of $\mathcal{FN}(u)$ is increased to $\gamma = 6n^{2/3}/\epsilon$, we set $u$ to be heavy. On the other hand, when $|\mathcal{FN}(u)|$ decreases to $\gamma/2$, we set $u$ to be light. Whenever $u$ becomes light, we immediately scan all $v\in \mathcal{FN}(u)$ and decrease each $\widetilde{\mathbf{dist}}(r,v)$ accordingly. This completes the description of our algorithm.

\paragraph{Running time analysis}

Let us now analyze the running time of the lazy ES-tree. For each vertex $u$, every time $\widetilde{\mathbf{dist}}(r,u)$ decreases by $n^{1/3}$, we might scan $u$'s entire in- and out-neighborhoods. Since $\widetilde{\mathbf{dist}}(r,u)$ can only decrease at most $n$ times, the total running time for this part of the algorithm is $O(nm /n^{1/3}) = O(mn^{2/3})$. 

For every light vertex $u$, we scan $\mathcal{FN}(u)$ every time $\widetilde{\mathbf{dist}}(r,u)$ decreases. Since $\widetilde{\mathbf{dist}}(r,u)$ can only decrease at most $n$ times and since $u$ is light, the total running time for all vertices spent for this part of the algorithm is $O(\sum_{v \in V} n\gamma )=O(n^{2+2/3}/\epsilon)$.

Whenever a vertex $u$ changes from heavy to light, we scan $\mathcal{FN}(u)$. If $u$ only changes from heavy to light once per value of $\widetilde{\mathbf{dist}}(r,u)$, then the running time is $O(n^{2+2/3}/\epsilon)$ by the same argument as the previous paragraph. So, we only consider the times in which $u$ toggles between being light and heavy whilst having the same value of $\widetilde{\mathbf{dist}}(r,u)$. Since the position of vertices in $\texttt{Cache}_u$ can only decrease, the only way for $u$ to become heavy while keeping the same value of $\widetilde{\mathbf{dist}}(r,u)$ is if an edge is inserted. Since $\gamma/2$ edges must be inserted before $u$ becomes heavy since it last became light, there were $\gamma/2$ edge insertions with tail $u$. Since each inserted edge is only added to a single $\mathcal{FN}(u)$ (namely to the forward neighborhood of its tail), we can amortize the cost of scanning the $\gamma/2$ vertices in $\mathcal{FN}(u)$ over the $\gamma/2$ insertions.

Combining everything, and since the classic ES-tree to depth $n^{2/3}$ takes at most $O(mn^{2/3})$ update time when run to depth $n^{2/3}$, we establish the desired running time.

\paragraph{Analysis of correctness}
Let us now argue that our distance estimates are maintained with multiplicative error $(1+\epsilon)$. The idea of the argument can be roughly summarized by the following points: 
\begin{enumerate}
    \item the light vertices do not contribute any error,
    \item we can bound the error contributed by pairs of heavy vertices whose forward neighborhoods overlap, and
    \item the number of heavy vertices on any shortest path with pairwise disjoint forward neighborhoods is small.
\end{enumerate}
We point out that while the main idea of allowing large error in heavy parts of the graphs is similar to \cite{bernstein2016deterministic}, we rely on an entirely new method to prove that this incurs only small total error. We start our proof by proving the following useful invariant. 
\begin{invariant}\label{inv:warm}
After every edge update, if $v\in \mathcal{FN}(u)$ then $|\widetilde{\mathbf{dist}}(r,v) - \widetilde{\mathbf{dist}}(r,u)| \leq n^{1/3}$. 
\end{invariant}
\begin{proof}
First suppose that $\widetilde{\mathbf{dist}}(r,u) \leq \widetilde{\mathbf{dist}}(r,v)$. Since $\widetilde{\mathbf{dist}}(r,u)$ and $\widetilde{\mathbf{dist}}(r,v)$ can only decrease, we wish to show that $\widetilde{\mathbf{dist}}(r,u)$ cannot decrease by too much without $\widetilde{\mathbf{dist}}(r,v)$ also decreasing. This is true simply because every time $\widetilde{\mathbf{dist}}(r,u)$ decreases by at least $n^{1/3}$, $\widetilde{\mathbf{dist}}(r,v)$ is set to at most $\widetilde{\mathbf{dist}}(r,u)+1$. 

Now suppose that $\widetilde{\mathbf{dist}}(r,u) > \widetilde{\mathbf{dist}}(r,v)$. Since $\widetilde{\mathbf{dist}}(r,u)$ and $\widetilde{\mathbf{dist}}(r,v)$ can only decrease, we wish to show that $\widetilde{\mathbf{dist}}(r,v)$ cannot decrease by too much while remaining in $\mathcal{FN}(u)$. This is true simply because every time $\widetilde{\mathbf{dist}}(r,v)$ decreases by at least $n^{1/3}$, we update $v$'s position in $\texttt{Cache}_u$. If $\widetilde{\mathbf{dist}}(r,v)< \widetilde{\mathbf{dist}}(r,u)+2$ and $v$'s position in $\texttt{Cache}_u$ is updated, then $v$ leaves $\mathcal{FN}(u)$.
\end{proof}

Consider a shortest path $\pi_{r,t}$ for any $t \in V$, at any stage of the incremental graph $G$. Let $t_0=s$. Then, for all $i$, let $r_{i+1}$ be the first heavy vertex after $t_{i}$ on $\pi_{r,t}$ and let $t_{i+1}$ be the last vertex on $\pi_{r,t}$ whose forward neighborhood intersects with the forward neighborhood of $r_{i+1}$ (possibly $t_{i+1} = r_{i+1}$). Thus, we get pairs $(r_1, t_1), (r_2, t_2), \dots, (r_k, t_k)$. Additionally, let $r_{k+1} = t$. Since the forward neighborhoods of all $r_i$'s are disjoint and of size at least $\gamma/2$ (recall that $r_i$ is heavy), we have that there are at most $k\leq 2n/\gamma$ pairs $(r_i, t_i)$.

For each $i$, let $v_i$ be some vertex in $\mathcal{FN}(r_i)\cap \mathcal{FN}(t_i)$. Note that $v_i$ exists by definition of $t_i$. By Invariant~\ref{inv:warm}, $|\widetilde{\mathbf{dist}}(r, r_i) - \widetilde{\mathbf{dist}}(r, v_i)| \leq n^{1/3}$ and $|\widetilde{\mathbf{dist}}(r, t_i) - \widetilde{\mathbf{dist}}(r, v_i)| \leq n^{1/3}$. Thus, $\widetilde{\mathbf{dist}}(r, t_i) - \widetilde{\mathbf{dist}}(r, r_i) \leq 2n^{1/3}$.

Let $t'_i$ be the vertex on $\pi_{r,t}$ succeeding $t_i$ (except $t'_0 = r$). If $t'_i\in \mathcal{FN}(t_i)$ then by Invariant~\ref{inv:warm}, $\widetilde{\mathbf{dist}}(r, t'_i) - \widetilde{\mathbf{dist}}(r, t_i) \leq n^{1/3}$. Otherwise, $t'_i\not\in \mathcal{FN}(t_i)$ so $\widetilde{\mathbf{dist}}(r, t'_i)\leq \widetilde{\mathbf{dist}}(r, t_i)+1$. So regardless, we have $\widetilde{\mathbf{dist}}(r, t'_i) - \widetilde{\mathbf{dist}}(r, t_i) \leq n^{1/3}$ and therefore, since $\widetilde{\mathbf{dist}}(r, t_i) - \widetilde{\mathbf{dist}}(r, r_i) \leq 2n^{1/3}$, we have $\widetilde{\mathbf{dist}}(r, t'_i) - \widetilde{\mathbf{dist}}(r, r_i) \leq 3 n^{1/3}$.

We will show that if $u$ is a light vertex and $(u,v)$ is an edge, then $\widetilde{\mathbf{dist}}(r,v)\leq \widetilde{\mathbf{dist}}(r,u)+1$. Consider the last of the following events that occurred: a) edge $(u,v)$ was inserted, b) $\widetilde{\mathbf{dist}}(r,u)$ was decremented, or c) $\widetilde{\mathbf{dist}}(r,u)$ became light. In case a), the algorithm decreases $\widetilde{\mathbf{dist}}(r,v)$ to be at most $\widetilde{\mathbf{dist}}(r,u)+1$. In cases b) and c), the algorithm updates the distance estimate of all vertices in $\mathcal{FN}(u)$, so if $\widetilde{\mathbf{dist}}(r,v)>\widetilde{\mathbf{dist}}(r,u)+1$ then $\widetilde{\mathbf{dist}}(r,v)$ is decreased to $\widetilde{\mathbf{dist}}(r,u)+1$. Thus we have shown that $\widetilde{\mathbf{dist}}(r, r_{i+1})-  \widetilde{\mathbf{dist}}(r, t'_i) = d(t'_i, r_{i+1})$. 

Putting everything together, $\pi_{r,t}$ can be partitioned into (possibly empty) path segments $\pi_{r,t}[t'_i, r_{i+1}]$ and $\pi_{r,t}[r_{i+1}, t'_{i+1}]$. Observe that by definition for each path segment $\pi_{r,t}[t'_i, r_{i+1}]$, the vertices of all edge tails on that segment are light. Thus, by preceding arguments, we can now bound $\widetilde{\mathbf{dist}}(r, t)$ by
\begin{align*}
    \widetilde{\mathbf{dist}}(r, t) &\leq \sum_{i=0}^k \widetilde{\mathbf{dist}}(r, r_{i+1}) - \widetilde{\mathbf{dist}}(r, t'_i) + \sum_{i=0}^{k-1} \widetilde{\mathbf{dist}}(r, t'_{i+1}) - \widetilde{\mathbf{dist}}(r, r_{i+1})\\
    &< \sum_{i=0}^k d(t'_i, r_{i+1}) + 3k n^{1/3} \leq \mathbf{dist}(r, t) + n^{2/3}\epsilon
\end{align*}
The last inequality comes from our bound on $k$ and the definition of $\gamma$. Thus, if $\mathbf{dist}(r, t)>n^{2/3}$ then $\widetilde{\mathbf{dist}}(r, t) \leq (1+\epsilon)\mathbf{dist}(r, t)$. Otherwise, $\mathbf{dist}(r, t)\leq n^{2/3}$ so the classic ES-tree up to depth $n^{2/3}$ finds the exact value of $\mathbf{dist}(r, t)$.

\section{Overview}

Let us now describe how to improve the construction above to derive an $\tilde{O}(n^2 \log W/\epsilon^{2.5})$ algorithm. For the rest of this and the next sections, we focus on proving the theorem below which only deals with unweighted graphs, and extend the theorem using standard edge rounding techniques.

\begin{theorem}[Unweighted version of Theorem \ref{thm:ContributionIncrSSSPResult}]
There is a deterministic algorithm that given an unweighted directed graph $G=(V,E)$ subject to edge insertions, a vertex $r\in V$, and $\epsilon>0$, maintains for every vertex $v$ an estimate $\widetilde{\mathbf{dist}}(r,v)$ such that after every update $\mathbf{dist}(r,v)\leq \widetilde{\mathbf{dist}}(r,v)\leq (1+\epsilon)\mathbf{dist}(r,v)$, and runs in total time $\tilde{O}(n^2/\epsilon)$. A query for the approximate shortest path from $r$ to any vertex $v$ can be answered in time linear in the number of edges on the path.
\end{theorem}

There are two main differences between our $\tilde{O}(n^2/\epsilon)$ time algorithm and our warm-up $O(mn^{2/3}/\epsilon)$ time algorithm from the previous section:

\begin{enumerate}

\item Recall that the warm-up algorithm consisted of 1) a classic ES-tree of bounded depth to handle small distances, and 2) a ``lazy'' ES-tree (of depth $n$) to handle large distances. For our $\tilde{O}(n^2/\epsilon)$ time algorithm we will have $\log n$ ES-trees of varying degrees of laziness and to varying depths where each ES-tree is suited to handle a particular range of distances. In particular, for each $i$ from 0 to $\log n-1$, we have one lazy ES-tree that handles distances between $2^i$ and $2^{i+1}$. The ES-trees that handle larger distances can tolerate more additive error, and are thus lazier. 

\item Recall that in the warm-up algorithm, each vertex $v$ was of one of two types: light or heavy, depending the size of the forward neighborhood $\mathcal{FN}(v)$. For our $\tilde{O}(n^2/\epsilon)$ time algorithm, each vertex will be in one of $\Theta(\log n)$ heaviness levels. Roughly speaking, a vertex has heaviness $i$ in the lazy ES-tree up to depth $\tau$ if $|\mathcal{FN}(u)|\approx 2^i \frac{n}{ \tau}$. 
\end{enumerate}

Consider one of our $\log n$ lazy ES-trees. Let $\tau$ be its depth and let $\widetilde{\mathbf{dist}}_{\tau}(r,v)$ be its distance estimate for each vertex $v$. A central challenge caused by introducing $\log n$ heaviness levels for each lazy ES-tree is handling the event that a vertex changes heaviness level. We describe why unlike in the warm-up algorithm, handling changes in heaviness levels is not straightforward and requires careful treatment. In the warm-up algorithm, whenever a vertex $u$ changes from heavy to light, we scan all $v\in \mathcal{FN}(u)$ and decrease each $\widetilde{\mathbf{dist}}(r,v)$ accordingly. Then, in the analysis of the warm-up algorithm, we argued that if $u$ only changes from heavy to light once per value of $\widetilde{\mathbf{dist}}(r,u)$, we get the desired running time. Now that we have many heaviness levels and we are aiming for a running time of $\tilde{O}(n^2/\epsilon)$, we can no longer allow each vertex to change heaviness level every time we decrement $\widetilde{\mathbf{dist}}_{\tau}(r,u)$. In particular, suppose we are analyzing a lazy ES-tree up to depth $D$. Suppose for each vertex $u$, every time we decrement $\widetilde{\mathbf{dist}}_{\tau}(r,u)$, we change $u$'s heaviness level and scan $\mathcal{FN}(u)$ as a result. Then since $|\mathcal{FN}(u)|$ could be $\Omega(n)$, the final running time would be $\Omega(n^2D)$, which is too large. Thus, unlike in the warm-up algorithm, we require that the heaviness of each vertex does not change too often. 

Without further modification of the algorithm, the heaviness level of a vertex $u$ can change a number of times in succession. Suppose each index of  $\texttt{Cache}_u$ from index $\widetilde{\mathbf{dist}}_{\tau}(r,u)-\log n+2$ to index $\widetilde{\mathbf{dist}}_{\tau}(r,u)+1$ contains many vertices such that each of the next $\log n$ times we decrement $\widetilde{\mathbf{dist}}_{\tau}(r,u)$, $\mathcal{FN}(u)$ increases by enough that $u$ increases heaviness level upon each decrement of $\widetilde{\mathbf{dist}}_{\tau}(r,u)$. We would like to forbid $u$ from changing heaviness levels so frequently. To address this issue, we \emph{change the definition} of the forward neighborhood $\mathcal{FN}(u)$. 

In particular, if $\texttt{Cache}_u$ contains many vertices in the set of indices that closely precede $\texttt{Cache}_u[\widetilde{\mathbf{dist}}_{\tau}(r,u)]$, we \emph{preemptively} add these vertices to $\mathcal{FN}(u)$. In the above example, instead of increasing the heaviness of $u$ for every single decrement of $\widetilde{\mathbf{dist}}_{\tau}(r,u)$, we would preemptively increase the heaviness of $u$ by a lot to avoid increasing its heaviness again in the near future. Roughly speaking, vertex $u$ has heaviness $h(u)$ if $h(u)$ is the maximum value such that there are $\sim\frac{2^{h(u)} n}{\tau}$ vertices in $\texttt{Cache}_u[\widetilde{\mathbf{dist}}_{\tau}(r,u)-2^{h(u)},\tau]$. (Note that this definition of heaviness is an oversimplification for the sake of clarity.)

Like in the warm-up algorithm, the heaviness level of a vertex $u$ determines how often we scan $\mathcal{FN}(u)$. If a vertex $u$ has heaviness $h(u)$, this means that we scan $\mathcal{FN}(u)$ whenever the value of $\widetilde{\mathbf{dist}}_{\tau}(r,u)$ becomes a multiple of $2^{h(u)}$. 

In summary, when we decrement $\widetilde{\mathbf{dist}}_{\tau}(r,u)$, the algorithm does roughly the following:
\begin{itemize}
    \item If the value of $\widetilde{\mathbf{dist}}_{\tau}(r,u)$ is a multiple of $2^{h(u)}$, scan all $v\in \mathcal{FN}(u)$ and decrement $\widetilde{\mathbf{dist}}_{\tau}(r,v)$ if necessary.
    \item If the value of $\widetilde{\mathbf{dist}}_{\tau}(r,u)$ is a multiple of $2^{h(u)}$, increase the heaviness of $u$ if necessary.
    \item Regardless of the value of $\widetilde{\mathbf{dist}}_{\tau}(r,u)$, check if $u$ has left the forward neighborhood of any other vertex $w$, and if so, decrease the heaviness of $w$ if necessary.
\end{itemize}

\section{The Data Structure}

For each number $\tau$ between 1 and $n$ such that $\tau$ is a power of 2, we maintain a ``lazy ES-tree'' data structure $\mathcal{E}_\tau$. The guarantee of the data structure $\mathcal{E}_\tau$ is that for each vertex $v \in V$ with $\mathbf{dist}(r,v)\in [\tau, 2\tau)$, the estimate $\widetilde{\mathbf{dist}}_{\tau}(r,v)$ maintained by $\mathcal{E}_\tau$ satisfies  $\mathbf{dist}(r,v)\leq \widetilde{\mathbf{dist}}_{\tau}(r,v) \leq (1+\epsilon)d(s, v)$. Let $\tau_{max}=2\tau(1+\epsilon)$. Since $\mathcal{E}_{\tau}$ does not need to provide a $(1+\epsilon)$-approximation for distances $\mathbf{dist}(r,v)> 2\tau$, the largest distance estimate maintained by $\mathcal{E}_{\tau}$ is at most $\tau_{max}$. We use the distance estimate $\tau_{max}+1$ for all vertices that do not have distance estimate at most $\tau_{max}$. For all $u\in V$, the final distance estimate $\widetilde{\mathbf{dist}}(r,u)$ is the minimum distance estimate $\widetilde{\mathbf{dist}}_{\tau}(r,u)$ over all data structures $\mathcal{E}_\tau$, treating each $\tau_{max}+1$ as $\infty$.

\paragraph{Definitions.}

We begin by making precise the definitions and notation from the algorithm overview section. For each data structure $\mathcal{E}_\tau$ and for each vertex $u\in V$ we define the following:
\begin{itemize}
\item $\widetilde{\mathbf{dist}}_{\tau}(r,u)$ is the distance estimate maintained by the data structure $\mathcal{E}_\tau$.
    \item $\texttt{Cache}_u$ is an array of 
$\tau_{max}$ lists of vertices whose purpose is to store (possibly outdated) information about $\widetilde{\mathbf{dist}}_{\tau}(r,v)$ for all $v\in \mathcal{N}^{out}(u)$. Every time we update the position of a vertex $v\in \mathcal{N}^{out}(u)$ in $\texttt{Cache}_u$, we move $v$ to $\texttt{Cache}_u[\widetilde{\mathbf{dist}}_{\tau}(r,v)]$.
    \item $h(u)$ is the \emph{heaviness} of $u$. Intuitively, if $u$ has large heaviness, this means that $u$ has a large \emph{forward neighborhood} (defined later) and that we scan $u$'s forward neighborhood infrequently.
    \item $\texttt{CacheIndex}(u)=\lfloor \widetilde{\mathbf{dist}}_{\tau}(r,u)-1 \rfloor_{2^{h(u)}}$. (Recall that $\lfloor x \rfloor_y$ is the largest multiple of $y$ that is at most $x$.) The purpose of $\texttt{CacheIndex}(u)$ is to define the forward neighborhood of $u$, which we do next.
    \item The \emph{forward neighborhood} of $u$, denoted $\mathcal{FN}(u)$ is defined as the the set of vertices in $\texttt{Cache}_u[\texttt{CacheIndex}(u),\tau_{max}]$. Note that $\mathcal{FN}(u)$ is defined differently from the warm-up algorithm due to reasons described in the algorithm overview section.
    \item $\texttt{Expire}_u$ is an array of $\tau_{max}$ lists of vertices whose purpose is to ensure that $u$ leaves $\mathcal{FN}(v)$ once $\widetilde{\mathbf{dist}}_{\tau}(r,u)$ becomes less than $\texttt{CacheIndex}(v)$. In particular, $v\in \texttt{Expire}_u[i]$ if $u\in \mathcal{FN}(v)$ and $\texttt{CacheIndex}(v)=i$.
    \item We also define $\texttt{CacheIndex}$ with a second parameter, which will be useful for calculating the heaviness of vertices. Let $\texttt{CacheIndex}(v,2^i)=\lfloor \widetilde{\mathbf{dist}}_{\tau}(r,v)-1 \rfloor_{2^i}$. Note that $\texttt{CacheIndex}(u,2^{h(u)})$ is the same as $\texttt{CacheIndex}(u)$.
\end{itemize}

\paragraph{Initialization.}

We assume without loss of generality that the initial graph is the empty graph. To initialize each $\mathcal{E}_{\tau}$, we initialize $\widetilde{\mathbf{dist}}_{\tau}(r,r)$ to 0, and for each $u\in V\setminus \{r\}$, we initialize $\widetilde{\mathbf{dist}}_{\tau}(r,u)$ to $\tau_{max} + 1$. Additionally, for each $u\in V\setminus \{r\}$ we initialize the heaviness $h(u)$ to $0$, and we initialize the arrays $\texttt{Cache}_u$ and $\texttt{Expire}_u$ by setting each of the $\tau_{max} + 1$ fields in each array to an empty list.

\paragraph{The edge update algorithm.}

The pseudocode for the edge update algorithm is given in Algorithm~\ref{alg:insertEdge1}. We also outline the algorithm in words.

The procedure $\textsc{InsertEdge}(u,v)$ begins by updating $\texttt{Cache}_u$ and $\texttt{Expire}_v$ to reflect the new edge. Then, it calls $\textsc{IncreaseHeaviness}(u)$ to check whether the heaviness of $u$ needs to increase due to the newly inserted edge. Then, it initializes a set $H$ storing edges. 

Initially $H$ contains only the edge $(u,v)$. The purpose of $H$ is to store edges $(x,y)$ after the distance estimate $\widetilde{\mathbf{dist}}_{\tau}(r,x)$ has changed. We then extract one edge at a time and check whether the decrease in $x$'s distance estimate also translates to a decrease of $y$'s distance estimate by checking whether $\widetilde{\mathbf{dist}}_{\tau}(r,y) > \widetilde{\mathbf{dist}}_{\tau}(r,x)+1$. If so, then $\widetilde{\mathbf{dist}}_{\tau}(r,y)$ can be decremented and we keep the edge in $H$. Otherwise, we learned that $(x,y)$ cannot be used to decrease $\widetilde{\mathbf{dist}}_{\tau}(r,y)$ and we remove $(x,y)$ from $H$. We point out that in our implementation a decrease of $\Delta$ is handled in the form of $\Delta$ decrements where the edge is $\Delta + 1$ times extracted from $H$ until it is removed from $H$.

\begin{algorithm}
{
\fontsize{10}{10}\selectfont
\caption{Algorithm for handling edge updates.}
\label{alg:insertEdge1}
\SetKwProg{procedure}{Procedure}{}{}
\procedure{$\textsc{InsertEdge}(u, v)$}{
    Add $v$ to $\texttt{Cache}_u[\widetilde{\mathbf{dist}}_{\tau}(r,v)]$\;\label{line:insertcache}
    \If{$\widetilde{\mathbf{dist}}_{\tau}(r,v) \geq \normalfont{\texttt{CacheIndex}}(u)$}{
        Add $u$ to $\texttt{Expire}_v[\texttt{CacheIndex}(u)]$
    }
    $\textsc{IncreaseHeaviness}(u)$\;
    
    \If{$\widetilde{\mathbf{dist}}_{\tau}(r,v)>\widetilde{\mathbf{dist}}_{\tau}(r,u)+1$ \label{line:ifThenhLoop}}{
        Let $H$ be a set storing edges $(x,y)$ \;
        $H.\textsc{Insert}(u,v)$\; \label{line:Hinsert1}
        \While{$H \neq \emptyset$}{\label{line:while}
            Let tuple $(x,y)$ be any tuple in $H$\;
            \If{$\widetilde{\mathbf{dist}}_{\tau}(r,y) > \widetilde{\mathbf{dist}}_{\tau}(r,x)+1$ \label{line:ifThenDecrement}}{
                $\textsc{Decrement}(x,y)$\;\label{line:calldecrement}
            }\Else{
                $H.\textsc{Remove}(x,y)$
            }
        }
    }
}

\label{alg:decrease1}
\SetKwProg{procedure}{Procedure}{}{}
\procedure{$\textsc{Decrement}(u, v)$}{
    $\widetilde{\mathbf{dist}}_{\tau}(r,v) = \widetilde{\mathbf{dist}}_{\tau}(r,v) - 1$\;

    \If{$\widetilde{\mathbf{dist}}_{\tau}(r,v)$ \normalfont{ is a multiple of } $2^{h(v)}$}
    {
        $\textsc{IncreaseHeaviness}(v)$\;
        \ForEach{$w \in \mathcal{FN}(v)$\label{line:fndecrement}}{
            Move $w$ to $\texttt{Cache}_v[\widetilde{\mathbf{dist}}_{\tau}(r,w)]$\;
            Move $v$ to $\texttt{Expire}_w[\texttt{CacheIndex}(v)]$\;
            $H.\textsc{Insert}(v,w)$\label{line:Hinsert2}
        }
    }
    \ForEach{$w\in \normalfont{\texttt{Expire}}_v[\widetilde{\mathbf{dist}}_{\tau}(r,v)+1]$\label{line:exp}}{
        Move $v$ to $\texttt{Cache}_w[\widetilde{\mathbf{dist}}_{\tau}(r,v)]$\;\label{line:begin}
        Remove $w$ from $\texttt{Expire}_v$\;
        $\textsc{DecreaseHeaviness}(w)$\label{line:end}
    }
}

\label{alg:incheav1}
\SetKwProg{procedure}{Procedure}{}{}
\procedure{$\textsc{IncreaseHeaviness}(u)$}{
    $i' \gets \argmax_{i \in \mathbb{N}}\{| \texttt{Cache}_u[\texttt{CacheIndex}(u,2^i), \tau_{max}]| \geq (2^i-1) \frac{12 n \log n }{\epsilon\tau}\} $\label{line:calc1}\;

    \If{$i' > h(u)$\label{line:ifloop1}}{
        \ForEach{$v\in \normalfont{\texttt{Cache}}_u[\texttt{CacheIndex}(u,2^{i'}), \tau_{max}]$\label{line:loop1}}{
            Move $v$ to $\texttt{Cache}_u[\widetilde{\mathbf{dist}}_{\tau}(r,v)]$\;
            Remove $u$ from $\texttt{Expire}_v$\;
        }

        $h(u) \gets \argmax_{i \leq i'}\{| \texttt{Cache}_u[\texttt{CacheIndex}(u,2^i), \tau_{max}]| \geq  (2^i-1) \frac{6 n \log n }{\epsilon\tau}\}$\label{line:calc2} \;
        
        \ForEach{$v\in \mathcal{FN}(u)$\label{line:loop2}}{
            Add $u$ to $\texttt{Expire}_v[\texttt{CacheIndex}(u)]$\;
        }
    }
}

\SetKwProg{procedure}{Procedure}{}{}
\procedure{$\textsc{DecreaseHeaviness}(u)$}{
    $i' \gets \argmax_{i \in \mathbb{N}}\{| \texttt{Cache}_u[\texttt{CacheIndex}(u,2^i), \tau_{max}]| \geq (2^i-1) \frac{6 n \log n }{\epsilon\tau} \}$\;\label{line:i'dec}

    \If{$i' < h(u)$\label{line:preloop3}}{
        \ForEach{$v\in \mathcal{FN}(u)$\label{line:loop3}}{
            Move $v$ to $\texttt{Cache}_u[\widetilde{\mathbf{dist}}_{\tau}(r,v)]$\;
            Remove $u$ from $\texttt{Expire}_v$\;
        }

        $h(u)\gets \argmax_{i \in \mathbb{N}}\{|  \texttt{Cache}_u[\texttt{CacheIndex}(u,2^i), \tau_{max}]| \geq  (2^i-1) \frac{6 n \log n }{\epsilon\tau}\}$\label{line:sethd}\;
 
        \ForEach{$v\in \mathcal{FN}(u)$\label{line:loop4}}{
            Add $u$ to $\texttt{Expire}_v[\texttt{CacheIndex}(u)]$\;
            $H.\textsc{Insert}(u,v)$\label{line:Hdec}
        }
    }
}
}
\end{algorithm}

The procedure $\textsc{Decrement}(u,v)$ begins by decrementing $\widetilde{\mathbf{dist}}_{\tau}(r,v)$. Then, it checks whether $\widetilde{\mathbf{dist}}_{\tau}(r,v)$ is a multiple of $2^{h(v)}$. If so, it calls $\textsc{IncreaseHeaviness}(v)$ to check whether the recent decrements of $\widetilde{\mathbf{dist}}_{\tau}(r,v)$ have caused $\mathcal{FN}(v)$ to increase by enough that the heaviness $h(v)$ has increased. Also, if $\widetilde{\mathbf{dist}}_{\tau}(r,v)$ is a multiple of $2^{h(v)}$, $\texttt{CacheIndex}(v)$ and thus $\mathcal{FN}(v)$ have changed. Thus, we scan each vertex $w \in \mathcal{FN}(v)$ and update the position of $w$ in $\texttt{Cache}_v$. Then, we insert for each such vertex $w \in \mathcal{FN}(v)$ the edge $(v,w)$ into $H$ which has the eventual effect of decreasing $\widetilde{\mathbf{dist}}_{\tau}(r,w)$ to value at most $\widetilde{\mathbf{dist}}_{\tau}(r,v) + 1$. Since we perform these actions every $2^{h(u)}$ decrements of $\widetilde{\mathbf{dist}}_{\tau}(r,v)$, as we show later, we incur roughly $2^{h(u)}$ additive error on each out-going edge of $v$.

Additionally, the procedure $\textsc{Decrement}(u,v)$ checks whether decrementing $\widetilde{\mathbf{dist}}_{\tau}(r,v)$ has caused $v$ to expire from any of the forward neighborhoods that contain $v$. The vertices whose forward neighborhood $v$ needs to leave are stored in $\texttt{Expire}_v[\widetilde{\mathbf{dist}}_{\tau}(r,v)+1]$. For each $w\in \texttt{Expire}_v[\widetilde{\mathbf{dist}}_{\tau}(r,v)+1]$, we update $v$'s position in $\texttt{Cache}_w$ which causes $v$ to leave $\mathcal{FN}(w)$. Then, we call $\textsc{DecreaseHeaviness}(u)$ to check whether removing $v$ from $\mathcal{FN}(w)$ has caused the heaviness of $w$ to decrease. 

The procedures $\textsc{IncreaseHeaviness}(u)$ and $\textsc{DecreaseHeaviness}(u)$ are similar. We first describe $\textsc{DecreaseHeaviness}(u)$. On line~\ref{line:sethd} in $\textsc{DecreaseHeaviness}(u)$, $h(u)$ is set to $\argmax_{i \in \mathbb{N}}\{|  \texttt{Cache}_u[\texttt{CacheIndex}(u,2^i), \tau_{max}]| \geq  (2^i-1) \frac{6 n \log n }{\epsilon\tau}\}$. We note that $\texttt{Cache}_u$ may contain out-of-date information when $\textsc{DecreaseHeaviness}(u)$ is called, however, we wish to update $h(u)$ based on up-to-date information. Thus, before line~\ref{line:sethd}, we update $\texttt{Cache}_u$. However, we do not have time to update \emph{every} index of $\texttt{Cache}_u$, so instead we only update the relevant indices. To do so, it suffices to first calculate the value $i'$, which is the expression for $h(u)$ but using the out-of-date version of $\texttt{Cache}_u$, and then scan all $v\in \normalfont{\texttt{Cache}}_u[\texttt{CacheIndex}(u,2^{i'}), \tau_{max} ]$, updating the position of each such $v$ in $\texttt{Cache}_u$.

Recall that a smaller value of $h(u)$ means that we scan $\mathcal{FN}(u)$ more often. Thus, after we decrease $h(u)$ in $\textsc{DecreaseHeaviness}(u)$, the vertices $v\in \mathcal{FN}(u)$ might not have been scanned recently enough according to the new value of $h(u)$. Thus, to conclude the procedure $\textsc{DecreaseHeaviness}(u)$, we scan each $v\in \mathcal{FN}(u)$ and add $(u,v)$ to the set $H$ so that $\textsc{Decrement}(u,v)$ is called later.

The main difference between $\textsc{IncreaseHeaviness}(u)$ and $\textsc{DecreaseHeaviness}(u)$ is that the constants in the expressions for calculating $i'$ and $h(u)$ are different from each other, which ensures that $u$ does not change heaviness levels too often. Additionally, the last step of  $\textsc{DecreaseHeaviness}(u)$ where we insert into $H$ is not necessary for $\textsc{IncreaseHeaviness}(u)$.

\section{Analysis of correctness}

For each vertex $t$, the algorithm obtains the distance estimate $\widetilde{\mathbf{dist}}(r, t)$ by taking the minimum $\widetilde{\mathbf{dist}}_{\tau}(r,t)$ over all $\tau$ (excluding when $\widetilde{\mathbf{dist}}_{\tau}(r,t)=\tau_{max}$+1). The goal of this section, is to prove that 
\[
\mathbf{dist}(r, t) \leq \widetilde{\mathbf{dist}}(r, t) \leq (1+\epsilon)\mathbf{dist}(r, t)
\]
for $\mathbf{dist}(r, t) \in [\tau, 2\tau]$. We prove this statement in two steps starting by giving a lower bound on $\widetilde{\mathbf{dist}}(r, t)$. 

\begin{lemma}\label{lem:correct1}
At all times, for all $\tau$, for any $t \in V$, we have $\mathbf{dist}(r, t) \leq \widetilde{\mathbf{dist}}_{\tau}(r,t)$.
\end{lemma}
\begin{proof}
It suffices to show that we only decrement $\widetilde{\mathbf{dist}}_{\tau}(r,v)$ if $v$ has an in-coming edge from a vertex with distance estimate more than 1 below $\widetilde{\mathbf{dist}}_{\tau}(r,v)$.
We only invoke the procedure $\textsc{Decrement}(u,v)$ from line \ref{line:calldecrement}, and we invoke it under the condition that $(u,v)$ is an edge and $\widetilde{\mathbf{dist}}_{\tau}(r,v) > \widetilde{\mathbf{dist}}_{\tau}(r,u) + 1$. Therefore after running $\textsc{Decrement}(u,v)$ we still have $\widetilde{\mathbf{dist}}_{\tau}(r,v) \geq \widetilde{\mathbf{dist}}_{\tau}(r,u) + 1$. 
\end{proof}

Let us next prove a small, but helpful lemma.

\begin{lemma}\label{lem:cachedecreases}
For all vertices $u,v\in V$, the index of $\normalfont{\texttt{Cache}_u}$ containing $v$ can only decrease over time.
\end{lemma}
\begin{proof}
Whenever we insert $v$ into to $\texttt{Cache}_u$ or move $v$ to a new index in $\texttt{Cache}_u$, $v$ is placed in $\texttt{Cache}_u[\widetilde{\mathbf{dist}}_{\tau}(r,v)]$. Since $\widetilde{\mathbf{dist}}_{\tau}(r,v)$ is monotonically decreasing over time, the lemma follows.
\end{proof}

Now, before giving an upper bound on the stretch of the distance estimate, we prove the following invariant which is analogous to Invariant~\ref{inv:warm} from the warm-up algorithm.

\begin{invariant}\label{inv:correct}
For all $u,v\in V$, after processing each edge update, if $v\in \mathcal{FN}(u)$ then $|\widetilde{\mathbf{dist}}_{\tau}(r,v) - \widetilde{\mathbf{dist}}_{\tau}(r,u)| \leq 2^{h(u)}$.
\end{invariant}
\begin{proof} 
We first note that the invariant is initially satisfied since $\mathcal{FN}(u)$ is initially empty.
First we prove that there is no $v \in \mathcal{FN}(u)$ with $\widetilde{\mathbf{dist}}_{\tau}(r,u) - \widetilde{\mathbf{dist}}_{\tau}(r,v) > 2^{h(u)}$. We first note that if $v\in \texttt{Cache}_u[\widetilde{\mathbf{dist}}_{\tau}(r,v)]$, then this inequality holds simply from the definitions of $\mathcal{FN}$ and $\texttt{CacheIndex}$. Thus, it suffices to show that if an event occurs that could potentially cause the inequality to be violated, then we have $v\in \texttt{Cache}_u[\widetilde{\mathbf{dist}}_{\tau}(r,v)]$.
We point out that the inequality could only be violated due to three events:
\begin{enumerate}
    \item \uline{$v \in \mathcal{FN}(u)$ and $\widetilde{\mathbf{dist}}_{\tau}(r,v)$ decreases:} We observe that when $\widetilde{\mathbf{dist}}_{\tau}(r,v)$ decrements, we iterate through each vertex $w \in \normalfont{\texttt{Expire}}_v[\widetilde{\mathbf{dist}}_{\tau}(r,v)+1]$ (line~\ref{line:exp}). Since we update $\texttt{Expire}_v$ immediately after $v$ is moved in $\texttt{Cache}_u$, we have that if $\widetilde{\mathbf{dist}}_{\tau}(r,v)=\texttt{CacheIndex}(u)-1$ then $u\in \normalfont{\texttt{Expire}}_v[\widetilde{\mathbf{dist}}_{\tau}(r,v)+1]$. Thus, if $\widetilde{\mathbf{dist}}_{\tau}(r,v)$ decrements to $\texttt{CacheIndex}(u)-1$, then the loop on line~\ref{line:exp} moves $v$ to $\texttt{Cache}_u[\widetilde{\mathbf{dist}}_{\tau}(r,v)]$.
    \item 
    \uline{$h(u)$ decreases:} We note that only the procedure $\textsc{DecreaseHeaviness}(u)$ can decrease $h(u)$.
    (In particular, $h(u)$ cannot decrease in $\textsc{IncreaseHeaviness}(u)$ by Lemma~\ref{lem:sizefn}.) In $\textsc{DecreaseHeaviness}(u)$, $i'$ and $h(u)$ are each set to the expression 
    \[
    \argmax_{i \in \mathbb{N}}\{| \texttt{Cache}_u[\texttt{CacheIndex}(u,2^i), \tau_{max}]| \geq (2^i-1) \frac{6 n \log n }{\epsilon\tau} \}
    \]
    on lines \ref{line:i'dec} and \ref{line:sethd}, respectively. Between these two lines, $\widetilde{\mathbf{dist}}_{\tau}(r,u)$ remains fixed, and thus $\texttt{CacheIndex}(u,2^i)$ also remains fixed for all $i$. Between the lines \ref{line:i'dec} and \ref{line:sethd}, we move each vertex $y$ in $\texttt{Cache}_u[\texttt{CacheIndex}(u,2^{i'}), \tau_{max} ]$ to $\texttt{Cache}_u[\widetilde{\mathbf{dist}}_{\tau}(r,y)]$. By Lemma \ref{lem:cachedecreases} this can only decrease the indices of vertices in $\texttt{Cache}_u$ and therefore the size of  $\normalfont{\texttt{Cache}}_u[\texttt{CacheIndex}(u,2^{i'}), \tau_{max}]$ can only decrease. Thus, when we pick the new $h(u)$, it satisfies $h(u) \leq i'$. It follows that each vertex $y\in \mathcal{FN}(u)$
    has been moved to $\texttt{Cache}_u[\widetilde{\mathbf{dist}}_{\tau}(r,y)]$.
    
    \item \uline{$v$ is added to $\mathcal{FN}(u)$:} A vertex $v$ can be added to $\mathcal{FN}(u)$ if either the edge $(u,v)$ is inserted, the distance $\widetilde{\mathbf{dist}}_{\tau}(r,u)$ decreases to a multiple of $2^{h(u)}$, or $h(u)$ increases. If the edge $(u,v)$ is inserted then $v$ is added to $\texttt{Cache}_u[\widetilde{\mathbf{dist}}_{\tau}(r,v)]$ on line \ref{line:insertcache}. If $\widetilde{\mathbf{dist}}_{\tau}(r,u)$ decreases to a multiple of $2^{h(u)}$ then in the loop on line \ref{line:fndecrement}, if $y\in\mathcal{FN}(u)$ then $y$ is moved to $\texttt{Cache}_u[\widetilde{\mathbf{dist}}_{\tau}(r,y)]$. It remains to argue about the last case, where $h(u)$ is increased: we observe that in procedure $\textsc{IncreaseHeaviness}(u)$, we first pick a new potential heaviness $i'$ on line \ref{line:calc1} and then scan all vertices in $\texttt{Cache}_u[\texttt{CacheIndex}(u, 2^{i'}), \tau_{max}]$, moving each vertex $y$ to $\texttt{Cache}_u[\widetilde{\mathbf{dist}}_{\tau}(r,y)]$. Then, we take the new value $h(u) \leq i'$ in line \ref{line:calc2} and since we choose $h(u)$ among values smaller than $i'$, 
    each vertex $y\in \mathcal{FN}(u)$
    has been moved to $\texttt{Cache}_u[\widetilde{\mathbf{dist}}_{\tau}(r,y)]$.
\end{enumerate}

It remains to prove that there is no $v \in \mathcal{FN}(u)$ with $\widetilde{\mathbf{dist}}_{\tau}(r,v) - \widetilde{\mathbf{dist}}_{\tau}(r,u) > 2^{h(u)}$. Again, we point out that the inequality could only be violated due to three events:
\begin{enumerate}
    \item \uline{$h(u)$ decreases:} Again, only the procedure $\textsc{DecreaseHeaviness}(u)$ can decrease $h(u)$. In line \ref{line:Hdec} in $\textsc{DecreaseHeaviness}(u)$, for every vertex $v \in \mathcal{FN}(u)$ that could potentially have its distance estimate decreased, $(u,v)$ is inserted into the set $H$, which has the eventual effect that $\widetilde{\mathbf{dist}}_{\tau}(r,v)\leq\widetilde{\mathbf{dist}}_{\tau}(r,u)+1$, once $H$ is empty.
    \item \uline{$\widetilde{\mathbf{dist}}_{\tau}(r,u)$ is decremented:} Let $h^{NEW}(u)$ be the value of $h(u)$ at the point in time when we have just decremented $\widetilde{\mathbf{dist}}_{\tau}(r,u)$. Let $\ell$ be the smallest multiple of $2^{h^{NEW}(u)}$ that is at least $\widetilde{\mathbf{dist}}_{\tau}(r,u)$. Let $h_{\ell}(u)$ be the value of $h(u)$ at the point in time when $\widetilde{\mathbf{dist}}_{\tau}(r,u)$ was decremented to $\ell$. We note that if $h_{\ell}(u)\leq h^{NEW}(u)$ then $\ell$ is a multiple of $2^{h_{\ell}(u)}$. Thus, when $\widetilde{\mathbf{dist}}_{\tau}(r,u)$ was decremented to $\ell$, if $\widetilde{\mathbf{dist}}_{\tau}(r,v)>\widetilde{\mathbf{dist}}_{\tau}(r,u)+1$ then we added $(u,v)$ to $H$, which has the effect of decreasing $\widetilde{\mathbf{dist}}_{\tau}(r,v)$ to $\ell+1$. 
    Thus, once we finish processing the current edge update, we have $\widetilde{\mathbf{dist}}_{\tau}(r,v)-\ell\leq 1$. By definition, $\ell-\widetilde{\mathbf{dist}}_{\tau}(r,u)\leq 2^{h(u)}-1$, so we have $\widetilde{\mathbf{dist}}_{\tau}(r,v)-\widetilde{\mathbf{dist}}_{\tau}(r,u)\leq 2^{h(u)}$. 
    \item \uline{$v$ is added to $\mathcal{FN}(u)$:} Since we are assuming that $\widetilde{\mathbf{dist}}_{\tau}(r,v) > \widetilde{\mathbf{dist}}_{\tau}(r,u)$, the only way $v$ can be added to $\mathcal{FN}(u)$ is if the edge $(u,v)$ is inserted. In this case, if $\widetilde{\mathbf{dist}}_{\tau}(r,v)>\widetilde{\mathbf{dist}}_{\tau}(r,u)+1$, then the algorithm inserts $(u,v)$ into the set $H$, which has the eventual effect that $\widetilde{\mathbf{dist}}_{\tau}(r,v)\leq\widetilde{\mathbf{dist}}_{\tau}(r,u)+1$.
\end{enumerate}
\end{proof}


Next, we prove a lower bound on the size of the forward neighborhoods.

\begin{lemma}\label{lem:sizefn}
For all $u\in V$, $|\mathcal{FN}(u)|\geq(2^{h(u)} - 1)  \frac{6 n\log n}{\epsilon\tau}$ at all times except lines~\ref{line:begin} to \ref{line:end} and during $\textsc{DecreaseHeaviness}(u)$.
\end{lemma}
\begin{proof}
The inequality in the lemma statement could be violated due to two events: 
\begin{itemize}
    \item \uline{$h(u)$ increases:} $\textsc{IncreaseHeaviness}(u)$ is the only procedure that can increase $h(u)$. $\textsc{IncreaseHeaviness}(u)$ specifically sets $h(u)$ so that it satisfies $|\mathcal{FN}(u)|\geq (2^{h(u)} - 1)  \frac{6 n\log n}{\epsilon\tau}$.
    \item \uline{$\mathcal{FN}(u)$ shrinks:} There are two scenarios that could cause $\mathcal{FN}(u)$ to shrink. Either, 1) $h(u)$ decreases, in which case it is set so that $|\mathcal{FN}(u)|\geq (2^{h(u)} - 1)  \frac{6 n\log n}{\epsilon\tau}$, or 2) a vertex $v\in\mathcal{FN}(u)$ has its distance estimate $\widetilde{\mathbf{dist}}_{\tau}(r,v)$ decremented causing $v$ to leave $\mathcal{FN}(u)$. In this case, $v$ leaves $\mathcal{FN}(u)$ only if $\widetilde{\mathbf{dist}}_{\tau}(r,v)$ decrements to $\texttt{CacheIndex}(u)-1$ and $v$'s position in $\texttt{Cache}_u$ is updated to $\texttt{Cache}_u[\widetilde{\mathbf{dist}}_{\tau}(r,v)]$. We observe that when $\widetilde{\mathbf{dist}}_{\tau}(r,v)$ decrements, we iterate through each vertex $w\in \normalfont{\texttt{Expire}}_v[\widetilde{\mathbf{dist}}_{\tau}(r,v)+1]$ (line~\ref{line:exp}). Since we update $\texttt{Expire}_v$ immediately every time $v$ is moved to a new index in $\texttt{Cache}_u$, we have that if $\widetilde{\mathbf{dist}}_{\tau}(r,v)=\texttt{CacheIndex}(u)-1$ then $u\in \normalfont{\texttt{Expire}}_v[\widetilde{\mathbf{dist}}_{\tau}(r,v)+1]$. Thus, if $v$ has left $\mathcal{FN}(u)$, then the loop on line~\ref{line:exp} calls $\textsc{DecreaseHeaviness}(u)$, which specifically sets $h(u)$ so that it satisfies $|\mathcal{FN}(u)|\geq (2^{h(u)} - 1)  \frac{6 n\log n}{\epsilon\tau}$.
\end{itemize}
\end{proof}

\noindent
We are now ready to prove the final lemma, establishing the correctness of the algorithm.

\begin{lemma}\label{lem:correct}
After processing each edge update, for each $t \in V$ and each $\tau$, $\mathbf{dist}(r, t)\leq\widetilde{\mathbf{dist}}_{\tau}(r,t)$ and if $\mathbf{dist}(r, t) \in [\tau, 2\tau)$ then $\widetilde{\mathbf{dist}}_{\tau}(r,t) \leq (1+\epsilon)\mathbf{dist}(r, t)$.
\end{lemma}
\begin{proof}
Our main argument is a generalization of the proof of correctness from the warm-up algorithm. Fix a heaviness level $h>0$. Let $r = t_0$. Then, we define $r_{i+1}$ be the first vertex with heaviness $h$ after $t_{i}$ on $\pi_{r,t}$ and let $t_{i+1}$ be the last vertex on $\pi_{r,t}$ of heaviness $h$ whose forward neighborhood intersects with the forward neighborhood of $r_{i+1}$ (possibly $t_{i+1} = r_{i+1}$). Thus, we get pairs $(r_1, t_1), (r_2, t_2), \dots, (r_k, t_k)$. Additionally, let $r_{k+1} = t$.

By definition, the forward neighborhoods of all $r_i$'s are disjoint. By Lemma \ref{lem:sizefn}, for each $r_i$, $|\mathcal{FN}(r_i)|\geq(2^h - 1)  \frac{6 n\log n}{\epsilon\tau}$ and since all $r_i$'s have disjoint forward neighborhoods, we have at most $k$ pairs $(r_i, t_i)$ with 
\[
k \leq \frac{n}{(2^h - 1)  \frac{6 n\log n}{\epsilon\tau}} \leq \frac{\epsilon\tau}{6 (2^h-1)\log n}.
\]

For any $i$, let $v_i$ be a vertex in $\mathcal{FN}(r_i) \cap \mathcal{FN}(t_i)$ (which exists by definition of $t_i$). By Invariant \ref{inv:correct}, we have $|\widetilde{\mathbf{dist}}_{\tau}(r,r_i) - \widetilde{\mathbf{dist}}_{\tau}(r,v_i)| \leq 2^{h}$ and $|\widetilde{\mathbf{dist}}_{\tau}(r,v_i) - \widetilde{\mathbf{dist}}_{\tau}(r,t_i)| \leq 2^{h}$. Thus, $\widetilde{\mathbf{dist}}_{\tau}(r,t_i) - \widetilde{\mathbf{dist}}_{\tau}(r,r_i) \leq 2^{h+1}$.

Let $t'_i$ be the vertex on $\pi_{r,t}$ succeeding $t_i$ (except $t'_0 = s$). If $t'_i\in \mathcal{FN}(t_i)$ then by Invariant \ref{inv:correct}, we have $\widetilde{\mathbf{dist}}_{\tau}(r,t'_i) - \widetilde{\mathbf{dist}}_{\tau}(r,t_i) \leq 2^{h}$ and otherwise, $t'_i\not\in \mathcal{FN}(t_i)$ so $\widetilde{\mathbf{dist}}_{\tau}(r,t'_i) < \texttt{CacheIndex}(t_i) < \widetilde{\mathbf{dist}}_{\tau}(r,t_i)$. So regardless, we have $\widetilde{\mathbf{dist}}_{\tau}(r,t'_i) - \widetilde{\mathbf{dist}}_{\tau}(r,t_i) \leq 2^{h}$. Combining this with the previous paragraph, we have $\widetilde{\mathbf{dist}}_{\tau}(r,t'_i) - \widetilde{\mathbf{dist}}_{\tau}(r,r_i) \leq 3*2^{h}$.

Now, let $h_{max} = \log n$ be the maximum heaviness level. We handle heaviness level $h'$ (initially $h_{max}$) by finding the pairs $(r_i,t_i)$ 
for heaviness $h'$ on the path $\pi'$ (initially $\pi_{r,t}$). This partitions the path $\pi'$ into segments $\pi'[t'_i, r_{i+1}]$ and $\pi'[r_{i+1}, t'_{i+1}]$. We observe that all arc tails in these path segments have heaviness less than $h'$. We contract the path segments $\pi'[r_{i+1}, t'_{i+1}]$ to obtain the new path $\pi'$, decrement $h'$ and recurse. 
We continue this scheme until $h'$ is $0$. By the previous analysis for each heaviness level $h'$, summing over the distance estimate difference of vertex endpoints of each contracted segment we obtain at most $\frac{3(2^{h'})\epsilon\tau }{6(2^{h'}-1)\log n}\leq \frac{\epsilon\tau}{\log n}$ (since $h'>0$) total error. Thus, each heaviness level larger than 0 contributes at most $\frac{\epsilon\tau}{\log n}$ additive error and overall they only induce additive error ${\epsilon\tau}$. 

For $h' = 0$, we argue that the algorithm induces no error on edges on $\pi'$ where each arc tail is of heaviness $0$. We will show that if $u$ is vertex of heaviness $0$ and $(u,v)$ is an edge, then $\widetilde{\mathbf{dist}}_{\tau}(r,v)\leq \widetilde{\mathbf{dist}}_{\tau}(r,u)+1$. This is straightforward to see from the algorithm description, but we describe the argument in detail for completeness. Consider the last of the following events that occurred: a) edge $(u,v)$ was inserted, b) $\widetilde{\mathbf{dist}}_{\tau}(r,u)$ was decremented, or c) the heaviness of $\widetilde{\mathbf{dist}}_{\tau}(r,u)$ became $0$. Case a occurs in the $\textsc{InsertEdge}(u,v)$ procedure where the algorithm decreases $\widetilde{\mathbf{dist}}_{\tau}(r,v)$ to be at most $\widetilde{\mathbf{dist}}_{\tau}(r,u)+1$. Case b occurs in the $\textsc{Decrement}(v)$ procedure. Here, the algorithm checks whether $\widetilde{\mathbf{dist}}_{\tau}(r,v)$ is a multiple of $2^{h(v)}$, which is true since $h(v)=0$. Then the algorithm updates the distance estimate of all vertices in $\mathcal{FN}(u)$, so if $\widetilde{\mathbf{dist}}_{\tau}(r,v)>\widetilde{\mathbf{dist}}_{\tau}(r,u)+1$ then $\widetilde{\mathbf{dist}}_{\tau}(r,v)$ is decreased to $\widetilde{\mathbf{dist}}_{\tau}(r,u)+1$. Case c occurs in the $\textsc{DecreaseHeaviness}(u)$ procedure where again the algorithm updates the distance estimate of all vertices in $\mathcal{FN}(u)$. 

By definition, the path $\pi'$ above is of length at most $\mathbf{dist}(r, t)$ and therefore we obtain an upper bound on $\widetilde{\mathbf{dist}}_{\tau}(r,t)$ of $\mathbf{dist}(r, t) + {\epsilon\tau}$. Then, when $\mathbf{dist}(r, t) \geq \tau$, the additive error of ${\epsilon\tau}$ is subsumed in the multiplicative $(1+\epsilon)$-approximation, as required.
\end{proof}



\section{Running time analysis}

We will show that the total running time of each data structure $\mathcal{E}_{\tau}$ is $\tilde{O}(n^2/\epsilon)$. Since there are $O(\log n)$ values of $\tau$, this implies that the total running time of the algorithm is $\tilde{O}(n^2/\epsilon)$. For the rest of this section we fix a value of $\tau$.

We crucially rely on the following invariant, which guarantees that  the heaviness of each vertex $u$ is chosen to be maximal, in the sense that if $h(u)$ were larger then we would have an upper bound on the size of $\mathcal{FN}(u)$.

\begin{invariant}\label{inv:time}
At all times, for all $u\in V$ and all integers $i$ such that $h(u)< i \leq \log n$,
    \[|\normalfont{\texttt{Cache}}_u[\texttt{CacheIndex}(u,2^i), \tau_{max}]| \leq  (2^i - 1) \frac{12 n \log n }{\epsilon\tau}.\]
\end{invariant}
\begin{proof}
We note that the invariant is satisfied on initialization since $\texttt{Cache}_u$ is initially empty. Let us now consider the events that could cause the invariant to be violated for some fixed $i$:
\begin{enumerate}
    \item \uline{$h(u)$ is decreased:} We note that $h(u)$ is only decreased in line \ref{line:sethd} of $\textsc{DecreaseHeaviness}(u)$, where it is set to a value that satisfies the invariant.
    (In particular, $h(u)$ cannot decrease in $\textsc{IncreaseHeaviness}(u)$ by Lemma~\ref{lem:sizefn}.)
    \item \uline{A vertex $v$ is added to $\texttt{Cache}_u$:} This scenario could only occur due to an insertion of an edge $(u, v)$. However, after adding $v$ to $\texttt{Cache}_u$ (and $u$ to $\texttt{Expire}_v$), we directly invoke the procedure $\textsc{IncreaseHeaviness}(u)$, which we analyze below.
    \item \uline{$\texttt{CacheIndex}(u,2^i)$ is decreased:} Here, we note that $\texttt{CacheIndex}(u,2^i)$ decreases only if $\widetilde{\mathbf{dist}}_{\tau}(r,u)$ decreases to a multiple of $2^i$, in which case also call $\textsc{IncreaseHeaviness}(u)$.
\end{enumerate}

For the last two cases, it remains to prove that the procedure $\textsc{IncreaseHeaviness}(u)$ indeed resolves a violation of the invariant. If we do not enter the {\bf if} statement on line~\ref{line:ifloop1}, then by the definition of $i'$, the invariant is satisfied. If we do enter the {\bf if} statement, then invariant is satisfied for all $i>i'$. By Lemma \ref{lem:cachedecreases} the indices of vertices in $\texttt{Cache}_u$ can only decrease and therefore during the course of $\textsc{IncreaseHeaviness}(u)$, the size of $\normalfont{\texttt{Cache}}_u[\texttt{CacheIndex}(u,2^{i'}), \tau_{max}]$ can only decrease. Thus, when $\textsc{IncreaseHeaviness}(u)$ terminates, it is still the case that the invariant holds for all $i>i'$. On the other hand, if $i\leq i'$, then we set $h(u)$ on line \ref{line:calc2} so that the invariant is satisfied. 

\end{proof}

We can now prove the most important lemma of this section bounding the time spent in the loops starting at lines
\ref{line:fndecrement}, \ref{line:loop1}, \ref{line:loop2}, \ref{line:loop3} and \ref{line:loop4}.

\begin{lemma}\label{lem:iscan}
The total time spent in the loops starting in lines
\ref{line:fndecrement}, \ref{line:loop1}, \ref{line:loop2}, \ref{line:loop3} and \ref{line:loop4} is $O(n^2 \log^4 n / \epsilon)$.
\end{lemma}
\begin{proof}
We start our proof by pointing out that the time spent in the loop starting in line \ref{line:loop2} is subsumed by the time spent by the loop in line \ref{line:loop1} for the following reason. On line \ref{line:calc2} the heaviness is chosen so that the forward neighborhood is over a more narrow range of indices that in loop on line \ref{line:loop1}. Furthermore, By Lemma \ref{lem:cachedecreases} the indices of vertices in $\texttt{Cache}_u$ can only decrease and therefore between lines \ref{line:loop1} and \ref{line:loop2}, for all $i$ the size of $\normalfont{\texttt{Cache}}_u[\texttt{CacheIndex}(u,2^{i}), \tau_{max}]$ can only decrease.

Similarly, the running time spent in the loop starting in line \ref{line:loop4} is subsumed by the running time of the loop starting in line \ref{line:loop3}. Thus, we only need to bound the running times of the loops starting in lines \ref{line:fndecrement}, \ref{line:loop1}, and \ref{line:loop3}.

To bound their running times, we define the concept of $i$-scanning: we henceforth refer to the event of iterating through $\normalfont{\texttt{Cache}}_u[\texttt{CacheIndex}(u,2^{i}), \tau_{max}]$ by \emph{$i$-scanning $\normalfont{\texttt{Cache}_u}$}, for any $0 \leq i \leq \log n$, choosing the largest $i$ applicable.

Lines \ref{line:fndecrement}, \ref{line:loop1} and \ref{line:loop3} all correspond to $i$-scanning $\normalfont{\texttt{Cache}_u}$: the loop on line \ref{line:fndecrement} $h(u)$-scans $\normalfont{\texttt{Cache}_u}$, the loop on line \ref{line:loop1} $i'$-scans $\normalfont{\texttt{Cache}_u}$ for $i'$ chosen on line \ref{line:calc1}, and the loop at line \ref{line:loop3} $h(u)$-scans $\normalfont{\texttt{Cache}_u}$. We now want to bound the total number of $i$-scans in order to bound the total running time. 

\begin{claim}
\label{clm:iscans}
For all $u\in V$ and all integers $0\leq i\leq \log n $, the algorithm $i$-scans $\normalfont{\texttt{Cache}}_u$ at most $O(\tau \log^2 n /2^i)$ times over the course of the entire update sequence.
\end{claim}
\begin{proof}
We first observe that we $i$-scan $\texttt{Cache}_u$ on line \ref{line:fndecrement} only if we are in the procedure $\textsc{Decrement}(u', u)$ for some $u'$, and $\hat{d}_{\tau}$ is decreased to a value that is a multiple of $2^i$. Since each invocation of $\textsc{Decrement}(u', u)$, decreases $\widetilde{\mathbf{dist}}_{\tau}(r,u)$ by $1$ and since  $\widetilde{\mathbf{dist}}_{\tau}(r,u)$ is monotonically decreasing, starting at $\tau_{max} + 1$, we conclude that the number of $i$-scans on line~\ref{line:fndecrement} is bound by $O(\tau/2^i)$. 

Next, let us bound the number of $i$-scans executed in the loop starting on line \ref{line:loop1} in procedure $\textsc{IncreaseHeaviness}(u)$. We claim that between any two $i$-scans of $\texttt{Cache}_u$ on line \ref{line:loop1}, either $\widetilde{\mathbf{dist}}_{\tau}(r,u)$ becomes a multiple of $2^i$ or at least $(2^i-1) \frac{ n \log n }{\epsilon\tau}$ edges emanating from $u$ are inserted into the graph. Observe that this claim immediately implies that there can be at most $\tau/2^i + \frac{n}{2^i \frac{ n \log n }{\epsilon\tau}} = O(\tau \log n / 2^i)$ $i$-scans on line \ref{line:loop1}. 

To prove this claim, let $t_1$ and $t_2$ be two points in time at which $i$-scans occur. We will prove that if $\widetilde{\mathbf{dist}}_{\tau}(r,u)$ did not become a multiple of $2^i$ between times $t_1$ and $t_2$ then there were many edge insertions between times $t_1$ and $t_2$. Observe first, that $\texttt{CacheIndex}(u,2^i)$ only changes when $\widetilde{\mathbf{dist}}_{\tau}(r,u)$ decreases to become a multiple of $2^i$. Thus, we assume for the rest of the proof that $\texttt{CacheIndex}(u,2^i)$ remains fixed between times $t_1$ and $t_2$. Therefore, the size of $\texttt{Cache}_u[\texttt{CacheIndex}(u,2^i), \tau_{max}]$ can only be increased if a new edge $(u,v)$ is inserted with $v$ at distance $\widetilde{\mathbf{dist}}_{\tau}(r,v) \geq \texttt{CacheIndex}(u,2^i)$. 

Now, let $i'$ be such that at time $t_1$, we $i$-scan with $i' = i$ was selected in line \ref{line:calc1}. However, observe that since at $t_2$, we only $i'$-scan with $i' = i$, if $i' > h(u)$. Thus, at some point $t$ such that $t_1 \leq t < t_2$, we either decreased the heaviness to below $i'$ on line \ref{line:i'dec}, or we already set $h(u)$ to a smaller value than $i'$ at time $t_1$ in line \ref{line:calc2}. In either case we certified that
\[
    | \texttt{Cache}_u[\texttt{CacheIndex}(u,2^i), \tau_{max}]| < (2^i-1) \frac{6 n \log n }{\epsilon\tau}.
\]
Since again, at time $t_2$, we picked $i' = i$, we certified on line \ref{line:calc1} that, 
\[
    | \texttt{Cache}_u[\texttt{CacheIndex}(u,2^i), \tau_{max}]| \geq (2^i-1) \frac{12 n \log n }{\epsilon\tau}.
\]
We have shown that between times $t_1$ and $t_2$, the size of $\texttt{Cache}_u[\texttt{CacheIndex}(u,2^i), \tau_{max}]$ can only increase due to edge insertions. Thus, we conclude that at least $6(2^i-1) \frac{n \log n }{\epsilon\tau}$ edges with tail $u$ must have been inserted between times $t_1$ and $t_2$. 

Finally, we prove that the number of $i$-scans in the loop starting on line  \ref{line:loop3} is bounded. We first observe that each time an $h(u)$-scan is executed, we afterwards decrease the heaviness by at least one: By Lemma \ref{lem:cachedecreases} the indices of vertices in $\texttt{Cache}_u$ can only decrease and therefore between lines \ref{line:loop3} and \ref{line:sethd} for any $i$ the size of  $\normalfont{\texttt{Cache}}_u[\texttt{CacheIndex}(u,2^{i'}), \tau_{max}]$ can only decrease. Thus, when we pick the new $h(u)$, it satisfies $h(u) \leq i'$.

Now, we use the fact that there are at most $\log n$ heaviness values to bound the number of $i$-scans in the loop starting on line  \ref{line:loop3}. Since the number of vertices scanned when we increase $h(u)$ is more than the number of vertices scanned on line~\ref{line:loop3} when we decrease $h(u)$, the total number of vertices scanned in the loop on line~\ref{line:loop3} is at most $\log n$ times the number of vertices scanned in the loop on line~\ref{line:loop2}.
Thus, there are at most $O(\tau \log^2 n / 2^i)$ $i$-scans on line  \ref{line:loop3}.
\end{proof}

Now, the running time of each of these $i$-scans can be bound by $O(2^i \frac{n \log n}{\epsilon \tau})$ by Invariant \ref{inv:time}, so we obtain the claimed running time of
\[
    \sum_i O\left((\tau \log^2 n /2^i)\left(2^i \frac{n \log n}{\epsilon \tau}\right)\right)= O(n \log^4 n / \epsilon).
\]
\end{proof}

We can now reuse claim \ref{clm:iscans} to bound the total time spent in the loop on line \ref{line:while} in the procedure $\textsc{InsertEdge}(u,v)$.

\begin{lemma}\label{lem:insertEdge}
The total running time spent in the loop starting on line \ref{line:while} excluding calls to $\textsc{Decrement}(u,v)$ is bounded by $O(n^2 \log^4 n / \epsilon)$.
\end{lemma}
\begin{proof}
We first observe that on line \ref{line:Hinsert1} we only add newly inserted edges into $H$. Thus, we add a total of at most $n^2$ edges to $H$ during line \ref{line:Hinsert1}. The remaining edges are only inserted into $H$ during $i$-scans in the lines \ref{line:Hinsert2} and \ref{line:Hdec}. Since by claim \ref{clm:iscans} there are at most $O(\tau \log^2 n/2^i)$ $i$-scans of $\texttt{Cache}_u$ for any $u\in V$, and each $i$-scan is over at most $O(2^in \log n / \epsilon \tau)$ elements, similarly to the preceding lemma, we conclude that we iterate over at most $O(n \log^4 n/ \epsilon)$ elements in all $i$-scans of $\texttt{Cache}_u$ over all values of $i$, for a fixed $u\in V$. Since each element that we iterate over in each $i$-scan can only result in the insertion of a single edge into $H$, we can bound the total number of insertions into $H$ over the entire course of the algorithm by $O(n^2 \log^4 n / \epsilon)$. Further, we observe that each iteration of the loop in line \ref{line:while} either removes an edge from the set $H$, or decrements a distance estimate, we can bound the total number of iterations of the loop by $O(n^2 \log^4 n / \epsilon) + n \tau_{max} = O(n^2 \log^4 n / \epsilon)$. Since each iteration takes $O(1)$ time, ignoring calls to $\textsc{Decrement}(u,v)$, the lemma follows.
\end{proof}

We are now ready to finish the running time analysis.

\begin{lemma} \label{lma:mainResultUnweightedUpperBound}
The total running time of a data structure $\mathcal{E}_{\tau}$ is $O(n^2 \log^5 n / \epsilon)$.
\end{lemma}
\begin{proof}
We begin with the procedure $\textsc{InsertEdge}(u,v)$. We note that this procedure takes constant time except for the {\bf while} loop, if we ignore the calls to $\textsc{IncreaseHeaviness}(u)$. Since there are at most $n^2$ edge insertions, the running time can be bounded by $O(n^2)$. Further, the total running time spend in the {\bf while} loop starting in line \ref{line:while} excluding calls to $\textsc{Decrement}(u,v)$ is bounded by $O(n^2 \log^4 n / \epsilon)$ by lemma \ref{lem:insertEdge}.

Next, let us bound the total time spent in procedure $\textsc{Decrement}(u,v)$. We first observe that the loop on line \ref{line:exp} iterates through each vertex $w$ in $\texttt{Expire}_u[\widetilde{\mathbf{dist}}_{\tau}(r,u)+1]$ removing each $w$ from $\texttt{Expire}_v$. Clearly, the number of iterations over the course of the entire algorithm can be bounded by the total number of times a vertex is inserted into $\texttt{Expire}_v$ over all $v$. Since these insertions occur in the loops starting in lines \ref{line:loop2} and \ref{line:loop4}, we have by lemma \ref{lem:iscan}, that the time spend on the loop starting in line \ref{line:exp} is bound by $O(n^2 \log^4 n/\epsilon)$. Further, ignoring subcalls, each remaining operation in the procedure $\textsc{Decrement}(u,v)$ takes constant time. We further observe that since each invocation of the procedure $\textsc{Decrement}(u,v)$ decreases a distance estimate, the procedure is invoked at most $n\tau_{max} = O(n^2)$ times. Thus, we can bound the total time spent in procedure $\textsc{Decrement}(u,v)$ by $O(n^2 \log^4 n/\epsilon)$.

For the remaining procedures $\textsc{IncreaseHeaviness}(u)$ and $\textsc{DecreaseHeaviness}(u)$, we note that the calculations of $i'$ and $h(u)$ on lines~\ref{line:calc1}, \ref{line:calc2}, \ref{line:i'dec}, and \ref{line:sethd} can be implemented in $O(\log n)$ time using a binary tree over the elements of array $\texttt{Cache}_u$ for each $u \in V$. 
We observe that both procedures receive at most $O(n^2 \log^4 n / \epsilon)$ invocations and since we already bounded the running times of the loops that call them. Thus, the total update time excluding loops can be bound by $O(n^2 \log^5 n/\epsilon)$. The loops take total time $O(n^2 \log^4 n/\epsilon)$ by Lemma \ref{lem:iscan}. 
This concludes the proof.
\end{proof}

Using $\log n$ data structures, one for each distance threshold $\tau$, we obtain the following result.

\begin{theorem} \label{thm:mainResultUnweightedUpperBound}
There is a deterministic algorithm that given an unweighted directed graph $G=(V,E)$, subject to edge insertions, a vertex $r\in V$, and $\epsilon>0$, maintains for every vertex $v$ an estimate $\widetilde{\mathbf{dist}}(r,v)$ such that after every update $\mathbf{dist}(r,v) \leq \widetilde{\mathbf{dist}}(r,v) \leq (1+\epsilon)\mathbf{dist}(r,v)$, and runs in total time $O(n^2 \log^6 n/\epsilon)$. A query for the approximate shortest path from $r$ to any vertex $v$ can be answered in time linear in the number of edges on the path.
\end{theorem}

\section{Weighted graphs}

Finally, we show how to extend our data structure to deal with weights $[1, W]$. We first show how to handle edge weights with a linear dependency in the running time on $W$. Then, we employ a standard edge-rounding technique \cite{raghavan1987randomized, cohen1998fast, zwick2002all, bernstein2009fully, madry2010faster, bernstein2016maintaining} that decreases the dependency in $W$ to $\log W$ (we will use a set-up most similar to \cite{bernstein2016maintaining}).

\begin{lemma} \label{lma:mainResultweightedUpperBound}
There is a deterministic algorithm that given a weighted directed graph $G=(V,E, w)$, subject to edge insertions and weight changes, with weights in $[1, W]$, a vertex $r\in V$, and $\epsilon>0$, maintains for every vertex $v$ an estimate $\widetilde{\mathbf{dist}}(r,v)$ such that after every update $\mathbf{dist}(r,v) \leq \widetilde{\mathbf{dist}}(r,v) \leq (1+\epsilon)\mathbf{dist}(r,v)$ if $\mathbf{dist}(r,v) \in [\tau, 2\tau)$ for some $\tau \leq n$, and runs in total time $O(n^2 \log^6 n/\epsilon^{1.5})$. A query for the approximate shortest path from $r$ to any vertex $v$ can be answered in time linear in the number of edges on the path.
\end{lemma}
\begin{proof}
Let us first describe an almost correct approach to modify the data structure $\mathcal{E}_{\tau}$ for unweighted graphs to handle edge weights and maintains shortest-paths of weight at most $\tau_{max}$ as follows: we change the if-condition in line  \ref{line:ifThenhLoop} from $\widetilde{\mathbf{dist}}_{\tau}(r,v)>\widetilde{\mathbf{dist}}_{\tau}(r,u)+1$ to $\widetilde{\mathbf{dist}}_{\tau}(r,v)>\widetilde{\mathbf{dist}}_{\tau}(r,u)+ w(u,v)$ and similarly in line \ref{line:ifThenDecrement} to $\widetilde{\mathbf{dist}}_{\tau}(r,y)>\widetilde{\mathbf{dist}}_{\tau}(r,x)+w(x,y)$. Further, we need to adapt indices in $\texttt{Cache}_u$ and $\texttt{Expire}_u$ accordingly to reflect the additional offset which is straightforward. 

Unfortunately, whilst the running time can still be bound as before, the correctness of the algorithm could no longer be guaranteed since invariant \ref{inv:correct} is no longer true. Recall that the invariant states that if $v\in \mathcal{FN}(u)$ then $|\widetilde{\mathbf{dist}}_{\tau}(r,v) - \widetilde{\mathbf{dist}}_{\tau}(r,u)| \leq 2^{h(u)}$. However, a vertex $u$ might now have a vertex $v$ in its forward-neighborhood at large distance but have a large edge weight on $(u,v)$ so it can not decrease its distance estimate.

However, a rather simple fix suffices: whenever we compute the heaviness $i$ by setting it to 
\[
\argmax_{i \in \mathbb{N}}\{| \texttt{Cache}_u[\texttt{CacheIndex}(u,2^i), \tau_{max}]| \geq (2^i-1) \frac{6 n \log n }{\epsilon\tau} \}
\]
we now no longer want to take all vertices in $\texttt{Cache}_u[\texttt{CacheIndex}(u,2^i), \tau_{max}]$ into account but only all neighbors $v$ such that the edge $(u,v)$ is of edge weight less than $2^i$ (observe that heaviness levels now depend on different sets). Similarly, we use the restriction on the neighbors for reducing heaviness, and it is only these edges that we then consider to be in the forward neighborhood. It is straightforward to conclude that invariant \ref{inv:correct} can be restored to guarantee that $v\in \mathcal{FN}(u)$ implies $|\widetilde{\mathbf{dist}}_{\tau}(r,v) - \widetilde{\mathbf{dist}}_{\tau}(r,u)| \leq 2 * 2^{h(u)}$. 

However, this change alone is not enough to get good running time. We also stipulate that each edge $(u,v)$ is scanned only every $\epsilon w(u,v)$ levels if $v \not\in \mathcal{FN}(u)$. It is straightforward to verify that this might induce a multiplicative error of $(1+\epsilon)$ on every edge. However, by rescaling $\epsilon$ by a constant factor, we can still conclude that by the restored invariant \ref{inv:correct}, the proof \ref{lem:correct} works as before and guarantees a $(1+\epsilon)$ multiplicative error on distances in $[\tau, 2\tau)$. 

Now let us bound the running time where we only bound the running time induced by scanning the weighted edges as described above since the bounds on the remaining running time carry seamlessly over from lemma \ref{lma:mainResultUnweightedUpperBound}. It can be verified that invariant \ref{inv:time} is still enforced for our new definition. Thus, if the heaviness is $h(u) = i$ for some vertex $u$, then the number of edges of weight in $(2^j, 2^{j+1}]$ for $j > i$ is at most $(2^j-1) \frac{12 n \log n }{\epsilon\tau}$. Since we scan these edges only every $\epsilon 2^j$ decrements of $\widetilde{\mathbf{dist}}_{\tau}(r,u)$, we obtain that the total running time required for all edge scans can be bound by
\[
    \sum_{v \in V} \sum_{j \in (0, \log n]} O\left( \left(2^j \frac{ n \log n }{\epsilon\tau} \right) \left(\frac{\tau_{max}}{\epsilon2^j} \right) \right) = O(n^2 \log^2 n /\epsilon^2).
\]
We point out that rebalancing terms slightly, we can reduce the $\epsilon$ dependency to $1/\epsilon^{1.5}$.
\end{proof}

We now prove the following lemma which implies \Cref{thm:ContributionIncrSSSPResult} as a corollary by maintaining a data structure $\mathcal{E}_{\tau_{hop}, \tau_{depth}}$ with parameters $\tau_{hop} = 2^i$ and $\tau_{depth}=2^j$, for every $i \in [0, \log n)$ and $j \in [0, \log nW)$. We point out that we define \emph{length} subsequently as the number of edges on a path and \emph{weight} as the sum over all edge weights on a path.

\begin{lemma}
There is a deterministic data structure $\mathcal{E}_{\tau_{hop}, \tau_{depth}}$ that given a weighted directed graph $G=(V,E, w)$, subject to edge insertions and weight changes, with weights in $[1, W]$, that takes parameters $\tau_{hop}$ and $\tau_{depth} \geq \tau_{hop}$, a vertex $r\in V$, and $\epsilon>0$, and maintains for every vertex $v$ with some shortest path in $G$ consisting of $[\tau_{hop}, 2\tau_{hop})$ edges and of weight in $[\tau_{depth}, 2\tau_{depth})$, an estimate $\widetilde{\mathbf{dist}}(r,v)$ such that after every update $\mathbf{dist}(r,v) \leq \widetilde{\mathbf{dist}}(r,v) \leq (1+\epsilon)\mathbf{dist}(r,v)$ and runs in total time $O(n^2 \log^8 n/\epsilon^{2.5})$. A query for the approximate shortest path from $r$ to any vertex $v$ can be answered in time linear in the number of edges on the path.
\end{lemma}
\begin{proof}
Let us start by defining some constant $\alpha = \frac{ \epsilon \tau_{depth} }{ \tau_{hop} }$ (we assume that $\alpha$ is integer by slightly perturbing $\epsilon$). Then, we let $G_{\alpha}$ be the graph $G$ after rounding each edge up to the nearest multiple of $\alpha$. We claim that for every vertex $t \in V$, for which we have a shortest path $\pi_{r,t}$ from $r$ to $t$ of length in $[\tau_{hop}, 2\tau_{hop})$ and weight in $[\tau_{depth}, 2\tau_{depth}]$, we have
\[
    w_{G_{\alpha}}(\pi_{r,t}) \leq w_G(\pi_{r,t}) \leq (1+2\epsilon)w_{G_{\alpha}}(\pi_{r,t}).
\]
To see this observe that each edge incurs additive error at most $\alpha$. However, since the path is of length at most $2 \tau_{hop}$, the additive error has to be bound by $2\alpha \tau_{hop} =  2\frac{ \epsilon \tau_{depth} }{ \tau_{hop} } \tau_{hop} = 2 \epsilon \tau_{depth}$. But since the path $\pi_{r,t}$ is of weight at least $\tau_{depth}$, we have overall at most a $(1+3\epsilon)$-approximation and therefore by rescaling $\epsilon$ by a constant factor, the claim follows.

Next, we let $G^*_{\alpha}$ be the graph $G_{\alpha}$ where each edge is scaled down by factor $\alpha$ and note that weights are all integral and positive. We next claim that for every vertex $t \in V$, for which we have a shortest path $\pi_{r,t}$ from $r$ to $t$ of length in $[\tau_{hop}, 2\tau_{hop})$ and weight in $[\tau_{depth}, 2\tau_{depth}]$, we have
\[
    w_{G^*_{\alpha}}(\pi_{r,t}) \leq \tau_{hop}/\epsilon
\]
To see this, observe that the path $\pi_{r,t}$ in $G_{\alpha}$ has weight at most $(1+2\epsilon)2\tau_{depth}$ by our preceding claim. Thus, scaling it down by $\alpha$, the path has weight at most
\[
    (1+2\epsilon)2\tau_{depth} / \alpha = (1+2\epsilon)2\tau_{depth} \frac{ \tau_{hop} }{ \epsilon \tau_{depth} } = (1+2\epsilon)2 \tau_{hop} / \epsilon \leq 8 \tau_{hop} / \epsilon.
\]
in $G^*_{\alpha}$. It now remains to run a data structure $\mathcal{E}_{\tau}$ on $G^*_{\alpha}$ with $\tau = \tau_{hop}$ as described in Theorem \ref{lma:mainResultweightedUpperBound}, however run to depth $8\tau_{hop}/ \epsilon$ (instead of $\tau_{max}$ which increases the running time by an $1/\epsilon$ factor. We then forward for each vertex $t$, the distance estimate $\widetilde{\mathbf{dist}}_{\tau}(r,t)$ scaled up by $\alpha$. This concludes the lemma.
\end{proof}